\documentclass[twoside,11pt]{article}

\usepackage{jmlr2e}

\ShortHeadings{Passport Option}{J. Teichmann, H. Wutte}
\firstpageno{1}

\begin{document}

\title{Machine Learning-powered Pricing of the Multidimensional Passport Option}

\author{\name Josef Teichmann \email jteichma@math.ethz.ch \\
       \addr Department of Mathematics\\
       ETH Zurich\\
       Zurich, Switzerland
              \AND
       \name Hanna Wutte \email hanna.wutte@math.ethz.ch \\
       \addr Department of Mathematics\\
       ETH Zurich\\
       Zurich, Switzerland
       }

\editor{}

\maketitle

\begin{abstract}
Introduced in the late 90s, the \emph{passport option} gives its holder the right to trade in a market and receive any positive gain in the resulting traded account at maturity. Pricing the option amounts to solving a stochastic control problem that for $d>1$ risky assets remains an open problem. Even in a correlated Black-Scholes (BS) market with $d=2$ risky assets, no optimal trading strategy has been derived in closed form.
In this paper, we derive a discrete-time solution for multi-dimensional BS markets with uncorrelated assets.
Moreover, inspired by the success of deep reinforcement learning in, e.g., board games, we propose two machine learning-powered approaches to pricing general options on a portfolio value in general markets.
These approaches prove to be successful for pricing the passport option in one-dimensional and multi-dimensional uncorrelated BS markets.
\end{abstract}


\section{Introduction}

A \emph{passport option} gives its holder the right to freely trade in a certain market within pre-specified constraints, over a predetermined time interval. The agent may take long and short, but bounded positions and receives at maturity the maximum of the account balance accrued and some floor value. The holder thus gets to keep any profits of her trading exceeding a threshold, without taking the loss.

Passport options are options on a traded account that come with various specifications such as trading constraints or terminal thresholds (see \citep{Shreve00optionson} for a concise overview). Moreover, they can be seen as a generalization to American options, allowing multiple exercise. A classic example of passport options are variable annuities: during the course of paying her premiums, the insured may select a certain portfolio with which to participate in the market, and she is guaranteed to receive at retirement her portfolio value, capped from below by a guaranteed benefit (floor value).

Pricing passport options (or more general options on a traded account) means solving specific stochastic optimal control problems. For many years, these problems have only been solved in markets with a single risky asset. 

\subsection{Prior Work}
Initially introduced by Bankers trust for an FX market \citep{hyer1997passport}, passport options have been studied intensively for the one-asset Black-Scholes (BS) market.

\cite{andersen} treat the 1D case with continuous and discrete switching rights, general dividend rate, and both European and American exercise. In experiments, they consider partial differential equation (PDE) methods (Crank Nicholson finite difference) to solve the Hamilton Jacobi Bellman (HJB)-PDE connected to the pricing problem. 
\cite{Penaud1999} treat exotic passport options on one underlying. Previous numerical approaches all suggest solving HJBs with PDE methods. \cite{penaud_various} derive the HJB for the multi-asset case, and, by analyzing this PDE, give properties of the price and optimal strategy for pricing the passport option. In particular, they show that optimal strategies attain values in the extreme points of the constraint set. Moreover, they analyze discrete constraints on the trading strategy.
\cite{nagayama19986} considers 1D passport options with general constraints (i.e., also the asymmetric case) and analyzes the value functions' properties.
\cite{delbaen2002passport} give a correspondence between pricing the 1D passport option and H1 martingales. They derive a discrete-time optimal solution, that we generalize to multiple assets in this paper. Moreover, they derive the price of the passport option as limit of discrete-time optimization problems. However, they note that there is no optimal strategy for the continuous time problem if the filtration is generated by the original Brownian motion. Moreover, in Chapter 8, they treat the case of a non-zero dividend rate.
\cite{Shreve00optionson} give a nice introduction to passport options and a good overview of results for 1D markets. They specifically deal with vacation calls and puts, options on a traded account with only non-positive/non-negative investment and derive the value for these instruments. For general boundary conditions, \cite{Shreve00optionson} derive optimal strategies, optimal values and put call parities in 1D.
\cite{henderson2000local} use concepts of local times, and stochastic coupling to derive the price and optimal strategy for the 1D, symmetric passport option. Furthermore, the authors extended their results to stochastic volatility models in \cite{henderson2001passport}, showing that under certain conditions (diffusion coefficient non-decreasing in the price) the optimal strategy remains unchanged.
\cite{malloch2011passport} give an alternative proof of the 1D case and also treats a discrete-time case with the Binomial model showing that the optimal strategy coincides in the discrete and continuous time settings and that the value determined in the Binomial model converges in distribution to the one of the continuous BS setting.
\cite{kanaujiya2018numerical} summarizes numerical approaches to price the 1D passport option, and suggests a novel addition to these algorithms. All considered algorithms focus on solving the HJB-PDE corresponding to the pricing problem.

To the best of our knowledge, pricing the multidimensional passport option for $d>1$ potentially dependent risky assets remains a challenging problem. Even in a Black-Scholes market with $d=2$ risky assets, no optimal trading strategy has been derived in closed form.


However, recent advances in deep learning (DL) for dynamic decision making have given rise to several algorithms that can be applied to \emph{approximate} optimal controls via neural networks (NNs) in fairly general control problems. 
One strand of research focuses on algorithms in which these NNs are  optimized purely by backward passes on expected accrued cost. 
These algorithms have been proposed in various contexts such as, e.g., in valuation and hedging problems \cite{buehler2019deep,DeepHedgingStochasticControl,reppen2022deep}, in production planning \cite{Reppen2022} or in gas storage \cite{Bachouch_2021,Curin2021}.
Motivated by applications in the field of economics, \citet{EMControl} introduce \emph{EM-Control}, another ML algorithm to solve finite-horizon stochastic optimal control problems based on the classical expectation-maximization algorithm (\cite{EMAlgo}). 
All of these approaches are simulation-based, suitable for high dimensional problems, and can be applied to general, not necessarily Markovian stochastic state processes. To stress this particular aspect, in this paper, we term algorithms of these sorts \emph{generalized policy approximation}.

Another strand of algorithms employs the \emph{dynamic programming principle (DPP)} that holds in classic optimal control theory for Markov dynamics of the controlled process. These algorithms, also known as \emph{deep reinforcement learning (deep RL)}, can be purely strategy-based (e.g.,  \emph{NNcontPI} in \cite{Hur__2021}) or additionally involve estimates of value functions \citep{ActorCriticAlgorithms,surveyA2C}. 
In settings with continuous Markov states, algorithms incorporating additional estimates of value functions have proven to be most successful \citep{VarianceReductionA2C}.

Despite the success of RL and generalized policy approximation methods in various control tasks, they have not been tested for pricing options on traded accounts. This paper seeks to fill this gap by applying RL as well as generalized policy approximation algorithms to price multivariate passport options.

\subsection{Our Contributions}
In the present paper, we make the following contributions.

\begin{itemize}
    \item  We derive a closed-form solution for the optimal trading strategy to price the multi-dimensional passport option for independent assets in discrete time.
    \item We introduce a generalized policy approximation algorithm for pricing multi-dimensional passport options for general dependence structures in the market. Moreover, we contrast this algorithm to a standard RL approach in simulated settings. Our source code is available on GitHub: \url{https://github.com/HannaSW/ML4PassportOptions}.

    \item We discuss the potential pitfalls of deep hedging in a classification context, and the difficulty of small step sizes in RL approximations to continuous time solutions. 
    \item We show that the ML approaches successfully recover the optimal strategies in both the well-known one asset, and the independent multi-asset cases. 
    We further analyze trained strategies for markets with correlated assets, where no analytical solution is known.  
\end{itemize}

In this paper, we discuss the pricing of passport options in multivariate BS markets. This setting we make precise in \Cref{sec:preliminaries}. In \Cref{sec:DiscreteTimeSolution}, we then proceed to give our main result: a discrete-time solution for multi-dimensional BS markets with uncorrelated assets. We introduce and discuss our ML algorithms in \Cref{sec:MLApproaches}, and test them in simulated experiments in \Cref{sec:experiments}.
\section{Preliminaries}\label{sec:preliminaries}
In this section, we define the market setting and the portfolio process underlying the passport option (\Cref{subsec:setting}) and discuss how to price passport options (\Cref{subsec:PricingThePassportOption}).
\subsection{Setting}\label{subsec:setting}
Given a probability space $(\Omega,\mathcal{F},\Q)$, consider a BS market with \emph{risk-neutral} measure $\Q$, consisting of $d\in\N$ risky assets

\begin{align}\label{eq:Sprocess}
    dS_t^i&=r S_t^i\,dt+\sigma^i S_t^i\,dW_t^{i}, \quad \sigma^i>0,\quad i=1,\ldots,d,\\
    dW^{i}_tdW^{j}_t&=\rho^{ij}\, dt, \quad i\neq j,\notag
\end{align}
with interest rate $r\in\mathbb{R}$, $\Q$-Brownian motions $W^i$ with correlations $-1\le\rho^{ij}\le 1$ and volatilities $\sigma^{i}>0$, {\small $i,j=1,\ldots,d$}.

Otherwise put,
\begin{align*}
   Cor(W_t^i,W_t^j)=\rho^{ij}, W=(W^1,\ldots,W^d)=A\tilde{W}
\end{align*}
with $\tilde{W}$ a $d$-dim standard BM, 
\begin{align*}
    AA^\top=\left(\begin{matrix}1 & \cdots & \rho^{1d}\\
    \vdots & \ddots &\vdots\\
   \rho^{d1}  & \cdots & 1\end{matrix}\right)
\end{align*}and
\begin{align*}
    S_t^i=S_0^i\exp\left(\left(r-\frac{(\sigma^i)^2}{2}\right)t+\sigma^i\langle a^i, \tilde{W}_t\rangle\right),
\end{align*}
where $a^i$ denotes the $i$\textsuperscript{th} row of $A$.
For a predictable process $q = (q_{t})_{0\le t\le T}$, the \emph{trading strategy}, we denote the portfolio value at final time $T$ by
\begin{equation}\label{eq:PVprocess}
    X_T = x_0 + \int_0^Tr(X_t-\sum_{i=1}^dq_t^{i}S_t^{i})\,dt+ \sum_{i=1}^d\int_0^T q_{t}^{i} \, dS_{t}^{i}.
\end{equation}


\subsection{Pricing the Passport Option}\label{subsec:PricingThePassportOption}
For pricing a passport option, the option seller considers the worst-case expected payoff among all strategies an option holder could choose within the pre-specified trading constraints. In this paper, these trading constraints allow the option holder to go short or long at most one unit in any of the underlying risky assets.\footnote{This is a standard formulation of the passport option. W.l.o.g., it can be extended to different bounds on allowed investments \citep{Shreve00optionson}.} The option seller's goal, therefore, is to find a trading strategy $q$ solving
\begin{equation}\label{eq:PPOprice}\tag{P}
   \max_{q=(q_t)_{0\le t\le T}}\mathbb{E}_\Q\left[e^{-rT}X_T^+\right]\quad\text{s.t. } ||q_t||_1\le 1, t\in[0,T].
\end{equation}

We are thus interested in solving a constrained stochastic optimal control problem \eqref{eq:PPOprice} for the (discounted) controlled Markov process\footnote{The dynamics of this Markov process can be obtained via integration by parts and using equations \eqref{eq:Sprocess} and \eqref{eq:PVprocess}.}
\begin{align}\label{eq:discountedMDP}
    de^{-rt}\left(\begin{matrix}X\\
    S
    \end{matrix}\right)_t=\left(\begin{matrix}
   q_t^1\sigma^1e^{-rt}S_t^1 & \cdots & q_t^d\sigma^de^{-rt}S_t^d\\
   \sigma^1e^{-rt}S_t^1 & \cdots & 0\\
    \vdots & \ddots &\vdots\\
   0 & \cdots & \sigma^de^{-rt}S_t^d\end{matrix}\right)A \,d\tilde{W}_t.
\end{align}

Observing the HJB of this control problem \eqref{eq:PPOprice} we note that the optimal strategy takes values in the corner points $\Diamond:=\{\pm e_i,i=1,\ldots,d\}$ where $e_i\in\R^d$ denotes the $i$\textsuperscript{th} unit vector (see also \citep{penaud2000optimal}).
Moreover, following \Cref{le:absEquiv}, throughout this paper we identify the problem of pricing the passport option with the control problem \eqref{eq:ABSobjective}. In \eqref{eq:ABSobjective}, we look for the trading strategy maximizing the expected absolute value of the portfolio value at terminal time, instead of its positive part.

\begin{lemma}\label{le:absEquiv}
The strategy $q^*$ is a solution to \eqref{eq:PPOprice} if and only if it solves  
\begin{equation}\label{eq:ABSobjective}\tag{AP}
   \max_{q=(q_t)_{0\le t\le T}}\mathbb{E}_\Q\left[e^{-rT}|X_T|\right]\quad\text{s.t. } ||q_t||_1\le 1, t\in[0,T].
\end{equation}
\end{lemma}
\begin{proof}
Since $(\cdot)^+=(\cdot)^-+id$
\begin{align*}
\E_\Q[e^{-rT}(X_T)^+]&=\underbrace{e^{-rT}\E_\Q[X_T]}_{=x_0}+e^{-rT}\E_\Q[(X_T)^-]\\
\iff \E_\Q[e^{-rT}(X_T)^+]&=\frac{x_0+e^{-rT}\E_\Q[(X_T)^-]+e^{-rT}\E_\Q[(X_T)^+]}{2}\\
&=\frac{e^{-rT}(\E_\Q[(X_T)^+]+\E_\Q[(X_T)^-])+x_0}{2}\\
&=\frac{e^{-rT}\E_\Q[|X_T|]+x_0}{2},
\end{align*}
and thus 
\begin{align*}
    \max_q e^{-rT}\E_\Q[(X_T)^+]=\frac{1}{2}  \max_q e^{-rT}\E_\Q[|X_T|]+x_0/2.
\end{align*}
\end{proof}

\begin{assumption}
In what follows, we consider discounted values $(e^{-rt}X_t, e^{-rt}S_t)$ and omit the discounting factor in the notation.
\end{assumption}

\section{A Discrete-Time Solution for the Multi-Dimensional Case with Independent Assets}\label{sec:DiscreteTimeSolution}
The control problem \eqref{eq:PPOprice} for pricing the passport option has been unsolved analytically for general BS markets with more than one risky asset. In this section, we consider the discrete-time setting $0=t_0<t_1<\ldots<t_N=T$. We derive a discrete-time solution to the pricing problem \eqref{eq:ABSobjective} similarly to how \citet{delbaen2002passport} proceed to solve the pricing problem for a one-dimensional market.

Let $q = (q_{t_n})_{1\leq n \leq N}$ be a predictable process in discrete-time, the \emph{trading strategy}, then the (discounted) portfolio value at final time $T$ is given as

\begin{equation}\label{eq:DiscreteTimeX}
    X^q_T = x_0 + \sum_{n=1}^N\sum_{i=1}^d q_{t_n}^{i} \Delta S_{t_n}^{i},
\end{equation}
where $\Delta S_{t_n} := (S_{t_n} - S_{t_{n-1}})$ are the returns of discrete-time (discounted) asset processes

\begin{align}\label{eq:DiscreteTimeS}
    S_{t_n}^i=S_0^i\exp\left(\left(-\frac{(\sigma^i)^2}{2}\right){t_n}+\sigma^i\langle a^i, \tilde{W}_{t_n}\rangle\right), i=1,\ldots, d.
\end{align}

As in \citep{delbaen2002passport}, we find a discrete-time solution for the (finite-horizon) control problem \eqref{eq:ABSobjective} via the \emph{dynamic programming principle (DPP)} for the $(d+1)$-dimensional Markov decision process (MDP) with

	 \begin{itemize}
	 	\item state space $\mathcal{X}:= \mathbb{R}^d_+\times\mathbb{R}$, where $(s,x)\in\mathcal{X}$ with $s\in\R^d_+$ and $x\in\R$ being the (discounted) states of risky assets and portfolio value respectively,
	 	\item action space ${\Diamond}:=\{\pm e_i,i=1,\ldots,d\}$, that contains the corner points of the $d$-dimensional $\ell_1$-ball, 
   \item transition probabilities given by the discrete-time dynamics of \eqref{eq:DiscreteTimeX}, and \eqref{eq:DiscreteTimeS}, 
   \item terminal reward $R(s,x):=|x|$, and
	 	\item policies ${\bf q}=(q_{t_n})_{1\le n\le N}$, $q_{t_n}:\mathcal{X}\to {\Diamond}$.
	 \end{itemize}
Solving \eqref{eq:ABSobjective} in discrete time then means to find the optimal policy ${\bf q^*}$ with $q^*_{t_n}:\mathcal{X}\to{\Diamond}$ that solves 
	 	\begin{equation}\label{eq:MDPObjective}\tag{MDPO}\max_{\bf q}\mathbb{E}\left[|X^q_T|\right].\end{equation}
To this end, we define the value functions $V_k:\R\times\R^d_+\to\R_+$ as
\begin{align*}
    V_T(x, s)&:=|x|,\\
    V_{k}(x, s)&:= \max_{q\in\Diamond}\Exsk\left[V_{k+1}\left(x+\sum_{j=1}^dq^js^j\left(\frac{S^j_{t_{k+1}}}{s^j}-1\right), S_{t_{k+1}}\right)\right],
\end{align*}
where $\Exsk[\cdot]:=\E[\cdot\mid X_{t_{k}}=x, S_{t_{k}}=s]$.

In \Cref{thm:optimalStrat}, we give a closed-form solution for the optimal policy ${\bf q^*}$ for the case when risky assets are uncorrelated. Even under this simplifying market assumption, ${\bf q^*}$ has been unknown for over decades.

\begin{theorem}[Independent assets]\label{thm:optimalStrat}
Let $\rho^{ij}=0$ for all $i,j=1,\ldots,d$, $i\neq j$. The optimal strategy $q^*=(q^*_{t_n})_{1\le n\le N}$ for problem \eqref{eq:MDPObjective} is given as
\begin{align}\label{eq:optimalStrat}
    q^*_{t_n}(x,s)&=-\sign(x)e_{j^*},\\
j^*& \in \argmax_{j\in\{1,\ldots,d\}} 
(s^i+|x|) \Phi(d_1^i)-s^i\Phi(d_2^i), \notag\\
    d_1^i&=\frac{\log(1+|x|/s^i)+\frac{1}{2}(\sigma^i)^2\Delta t_{n}}{\sigi\sqrt{\Delta t_{n}}},\notag\\
    d_2^i&= d_1^i - \sigi\sqrt{\Delta t_{n}}.\notag
\end{align}
Here, $e_j$ denotes the $j$\textsuperscript{th} unit vector and $\Delta t_{n}=t_{n}-t_{n-1}$.
\end{theorem}

\begin{remark}
By \Cref{thm:optimalStrat}, the optimal discrete-time trading strategy ${\bf q^*}$ for pricing the passport option is to invest at a time point $t_n$ the negative sign of the current portfolio value into the asset $S^i$ with highest
$CP^i(S_{t_n}^i/\kappa^i)\kappa^i$, where $CP^i(S_{t_n}^i/\kappa^i)$ is the price for a call with maturity $t_{n+1}-t_n$ and strike $S_{t_n}^i/\kappa^i$, with $\kappa^i=\frac{|X_{t_n}|+S_{t_n}^i}{S_{t_n}^i}$.

\end{remark}

\begin{proof}[\Cref{thm:optimalStrat}]
By DPP, $q^*$ is a solution to \eqref{eq:MDPObjective} if and only if 

\begin{equation}\label{eq:1stepobjective}
     q^*_{t_{k+1}} \in \argmax_{q\in\Diamond} \Exsk\left[V_{k+1}\left(x+\sum_{j=1}^dq^js^j\left(\frac{S^j_{t_{k+1}}}{s^j}-1\right), S_{t_{k+1}}\right)\right]
\end{equation}
for all $k$ in the recursion above.
By independence of $S^i$, $i=1\ldots,d$, we may apply Fubini's theorem to split the expectations, and due to the specific structure of $\Diamond$, the objective of \eqref{eq:1stepobjective} can be split into
\begin{align*}
    \max_{i, q^i\in\{\pm 1\}} \Esimxsk\left[\Esicxsk\left[V_{k+1}\left(x+q^is^i\left(\frac{S^i_{t_{k+1}}}{s^i}-1\right), S_{t_{k+1}}\right)\right]\right],
\end{align*}
where $S^{i-}$ denotes all but the i\textsuperscript{th} asset. We use the notations $\Esimxsk$ and $\Esicxsk$ for the expectation under the distribution of $S^{i-}_{t_{k+1}}$ respectively $S^i_{t_{k+1}}$, conditioned on $\{X_{t_{k}}=x, S_{t_{k}}=s\}$.
With \Cref{le:VposHom} of Appendix \ref{app:Results}, and further with a change of measure $\frac{dQ}{d\mathbb{Q}}=\frac{S^i_{t_{k+1}}}{s^i}$, we get
\begin{align*}    &\max_{i, q^i\in\{\pm 1\}} \Esimxsk\left[\Esicxsk\left[\frac{S^i_{t_{k+1}}}{s^i}V_{k+1}\left(x\left(\frac{S^i_{t_{k+1}}}{s^i}\right)^{-1}+q^is^i\left(1-\left(\frac{S^i_{t_{k+1}}}{s^i}\right)^{-1}\right), \left(s^{i},S_{t_{k+1}}^{i-}\right)\right)\right]\right]\\
    =&\max_{i, q^i\in\{\pm 1\}} \Esimxsk\left[\mathbb{E}_{{Q}}\left[V_{k+1}\left(x\left(\frac{S^i_{t_{k+1}}}{s^i}\right)^{-1}+q^is^i\left(1-\left(\frac{S^i_{t_{k+1}}}{s^i}\right)^{-1}\right), \left(s^{i},S_{t_{k+1}}^{i-}\right)\right)\right]\right]\\
    =&\max_{i, q^i\in\{\pm 1\}} \Esimxsk\left[\Esicxsk\left[V_{k+1}\left(x\left(\frac{S^i_{t_{k+1}}}{s^i}\right)+q^is^i\left(1-\left(\frac{S^i_{t_{k+1}}}{s^i}\right)\right), \left(s^{i},S_{t_{k+1}}^{i-}\right)\right)\right]\right].\\
\end{align*}
Here, the last step follows since the distribution of $\frac{S^i_{t_{k+1}}}{s^i}$ under $\Q$ is equal to the one of $\frac{s^i}{S^i_{t_{k+1}}}$
under $Q$. By \Cref{le:Vrepresentation} of Appendix \ref{app:Results} we can re-write $V_{k+1}$ as an integral w.r.t. a  probability measure $\muim$ on $\R_+$ that depends on the values of $S^{i-}$ and $s^i$ (but not on $S^{i}$). With this, and then using once more Fubini's theorem, we further re-write the one-step objective of Eq. \eqref{eq:1stepobjective} as
{\small\hspace{-5mm}
\begin{align*}
         \max_{i, q^i\in\{\pm 1\}}\,  \Esimxsk\bigg[\int_{\R_+}\Esicxsk\bigg[\max\bigg\{\underbrace{\bigg|x\left(\frac{S^i_{t_{k+1}}}{s^i}\right)+q^is^i\left(1-\left(\frac{S^i_{t_{k+1}}}{s^i}\right)\right)\bigg|}_{=\bigg||x|\left(\frac{S^i_{t_{k+1}}}{s^i}\right)+\sign(x)q^is^i\left(1-\left(\frac{S^i_{t_{k+1}}}{s^i}\right)\right)\bigg|}, z\bigg\}\,\bigg]d\muim(z)\bigg].
         \end{align*}
}

We further distinguish for each fixed asset $i$ the actions of going short or long, i.e., we distinguish the actions $q^i=1$ and $q^i=-1$ and define the respective expected one-step values in \Cref{def:pmInvestment}.
\begin{definition}[One-step investment in $i$\textsuperscript{th} asset at time $t_{k}$]\label{def:pmInvestment}
\begin{align*}
    \varphi^i_+(z)&:=\Esicxsk\left[\max\left\{\bigg||x|\left(\frac{S^i_{t_{k+1}}}{s^i}\right)+s^i\left(1-\frac{S^i_{t_{k+1}}}{s^i}\right)\bigg|,z\right\}\right],\\
        \varphi^i_{-}(z)&:=\Esicxsk\left[\max\left\{\bigg||x|\left(\frac{S^i_{t_{k+1}}}{s^i}\right)-s^i\left(1-\frac{S^i_{t_{k+1}}}{s^i}\right)\bigg|,z\right\}\right].
\end{align*}
Here, $\varphi^i_+(z)$ and $\varphi^i_-(z)$ characterize the one-step objective for investment $q^i=\sign(x)$ and $q^i=-\sign(x)$, respectively, in the $i$\textsuperscript{th} asset at time $t_{k}$.
\end{definition}

With the notation of \Cref{def:pmInvestment}, the one-step objective can alternatively be written as
\begin{align}\label{eq:onestepobjectivepm}
    \max_{i=1\ldots,d}\max\left\{\Esimxsk\left[\int_{\R_+}\varphi^i_+(z)\,d\muim(z)\right],\Esimxsk\left[\int_{\R_+}\varphi^i_-(z)\,d\muim(z)\right]\right\}.
\end{align}
By \Cref{le:negSignMax} of Appendix \ref{app:Results} and \Cref{re:varphiRepresentations} of Appendix \ref{app:Results}, \eqref{eq:onestepobjectivepm} is equivalent to
\begin{align*}
    \max_{i=1\ldots,d}\,\Esimxsk\left[\int_{\R_+}\varphi^i_-(z)\,d\muim(z)\right].
\end{align*}
In other words, for every asset $i$, $q^i=-\sign(x)$ yields a higher objective than $q^i=\sign(x)$.
We further investigate when an investment in one asset is to be preferred over an investment in any other.

We define for each asset $i$ the value of investing in that asset at a step $k$ 
\begin{align*}
    V^{i}_k(x,s):= \Esimxsk\left[\int_{\R_+}\varphi^i_-(z)\,d\muim(z)\right].
\end{align*}

The claim then follows from \Cref{le:whichAssetIsMax} of Appendix \ref{app:Results}.

\end{proof}

A special state in terms of the optimal trading action according to \Cref{thm:optimalStrat} is attained when the portfolio value reaches zero: first, both going long and short in an asset is of equal value. The convention in this paper is to set $\sign(0)=1$. Second, the optimal asset choice simplifies. Details can be seen in \Cref{cor:PV0}.

\begin{corollary}\label{cor:PV0}
For $x=0$,

\begin{align*}
    q^*_{t_n}(0,s)&=-\sign(0)e_{j^*},\\
j^*&=\argmax_{j=1,\ldots,d} s^j\left(2\Phi(\frac{1}{2}\sigma^j\sqrt{\Delta t_{n}})-1\right).
\end{align*}

In addition, for $\max_{1\le n\le N}|\Delta t_n|\approx 0$, the optimal strategy at $x=0$ is
\begin{align*}
    q^*_{t_n}(0,s)&=-\sign(0)e_{j^*},\\
j^*&\approx\argmax_i s^j\sigma^j.
\end{align*}
\end{corollary}

\begin{proof}
We revisit the proof of \Cref{thm:optimalStrat}. If the current portfolio value $x=0$, then $\varphi_+^i=\varphi_-^i$ for all $i=1,\ldots, d$ in \Cref{def:pmInvestment} and w.l.o.g., the one-step objective reduces to
\begin{align*}
    \max_{i=1,\ldots,d}\int_{\R_+}\varphi^i_-(z).
\end{align*}
An application of \Cref{le:whichAssetIsMax} in Appendix \ref{app:Results} with $x=0$ finally yields the result.

As can be seen in \Cref{le:whichAssetIsMax} in Appendix \ref{app:Results} in the one-step problem, the decision in which asset $S^i$ to invest in at time $t_{n-1}$ is equal to deciding for the asset with maximal price for an at-the-money call with maturity $\Delta t_n$. For decreasing mesh size in the discretized time grid $0=t_0<\ldots<t_N=T$, these call prices can be approximately modeled as $\sigma^is^i$ for every $i$ (see e.g., \citep[Proposition 5.1.]{relationshipCallVola}), thus the second statement follows.
\end{proof}

\section{Deep Learning Approaches}\label{sec:MLApproaches}
In this section, we discuss how to price passport options using DL. We point to certain pitfalls of common approaches in \Cref{sec:RelaxingTheProblem} and present two DL algorithms in \Cref{sec:PGA} and \Cref{sec:A2C} that we will use in \Cref{sec:experiments} to price passport options in simulated market settings. 

There are many ways of how to frame the pricing of passport options as a DL task. In the following, we give a short overview of popular approaches that are commonly grouped into value-based and action-based approaches.
\paragraph{Value-based Approaches.} First, one could think of learning the value function $V_t$ for each time point $t$. For continuous time, i.e. $t\in[0,T]$, this involves solving a fully non-linear PDE \citep{penaud2000optimal}. Classic numerical approaches for solving such an intricate task are scarce and often do not scale to higher dimensions \citep{FullyNonLinearPDEsNNs}. Recently, \citet{FullyNonLinearPDEsNNs} proposed a DL method to approximate (smooth, unique) solutions of fully non-linear PDEs.
However, in this paper, we turn to learning the optimal pricing strategy instead.
\paragraph{Action-based Approaches.} The literature is rich on DL algorithms that learn a strategy to maximize some expected terminal reward.
First, pricing the passport option can be framed in spirit of deep hedging \citep{buehler2019deep}, where the trading strategy is dynamically parametrized with (a collection of) NNs and is then optimized with some form of stochastic gradient descent on minimizing negative expected payoff. While there are universal approximation theorems \citep{buehler2019deep} that guarantee expressiveness of these models, it is not clear if the training method succeeds at finding optimal trading strategies. Typically, these deep hedging approaches perform well, which might be explained by implicit/explicit regularization in these models and/or the fact that there are no bad local optima. In the case of the passport option however, the latter is not true (cp. \Cref{re:problemwithdeephedging}) and we repeatedly find convergence to (bad) local optima (see \Cref{fig:1DPVdist}, where average terminal payoffs (magenta) are indistinguishable from those obtained by taking random actions in $\Diamond$ (yellow)). 
We thus argue in more detail in \Cref{re:problemwithdeephedging} that for the pricing of passport options, a standard deep hedging approach is ill-posed.

\begin{remark}[Local Optima - The Problem with Standard Deep Hedging]\label{re:problemwithdeephedging}
    The problem of finding a strategy $q: ||q||_1\le 1$ that maximizes $\E_\Q[|X_T^q|]$ with a gradient method is ill-posed: 
    the portfolio value for some time point $t_n$ at time point $t_{n-1}$ is given as $X^q_{t_n}=x_{t_{n-1}}-\langle q_{t_n},s_{t_{n-1}}\rangle+\langle q_{t_n},S_{t_{n}}\rangle$. For $q_{t_n}\in\Diamond$, $X^q_{t_n}$ is a random variable with
    $$\E_{s_{t_{n-1}},x_{t_{n-1}},t_{n-1}}[X^q_{t_n}]= x_{t_{n-1}}$$
    and non-zero variance, while for $q_{t_n}=0$ $$X^q_{t_n}\equiv x_{t_{n-1}}.$$
    Intuitively, since the value function $V_k(x,s)$ is convex in $x$ for any time point $k$, the value $\E_{s_{t_{n-1}},x_{t_{n-1}},t_{n-1}}[V_{N-n}(x_{t_{n-1}},{S_{t_n}})]$ for the deterministic portfolio value $x_{t_{n-1}}$ that one gets by not investing at all in a risky asset (i.e., action $q_{t_n}=0$) is smaller than the value $\E_{s_{t_{n-1}},x_{t_{n-1}},t_{n-1}}[V_{N-n}(X^q_{t_n},{S_{t_n}})]$ for the random variable $X^q_{t_n}$ with non-zero variance around $x_{t_{n-1}}$ that one gets for any $q_{t_n}\in\Diamond$.
    Thus, intuitively, the value for each of the actions in $\Diamond$ is bigger than the one of action $q_{t_n}=0$.
    Therefore, whenever a strategy $q^\theta$ for approximating $q^*$ is initialized s.t. ,e.g., for some time-point $t$ $0<{q^\theta_t}^i<1$, $i=1,\ldots,d$, a gradient step will pull towards the locally better corner point $q_{\text{local}}\in\Diamond$, even if the globally optimal action were attained at another point $q^* \in\Diamond, q^*\neq q_{\text{local}}$.
    In \Cref{fig:DHstrats}, we visualize this problem in a one-dimensional BS market. We see that at initialization, the network actions $0<{q^\theta_t}<1$ at time point $t$ take random values in $[-1,1]$ (for each $t$, the network $q^\theta_t$ is almost constant, as is typical in standard initializations). After training in spirit of deep hedging\footnote{We initialized a separate NN for each of the $T=32$ time steps and trained the entire architecture to minimize a MC estimate of $\E_\Q[|X_T^{q^{\theta}}|]$ over $2^{13}$ paths, for $2^7$ epochs, with batch size $2^8$, $\ell_2$-regularization and entropy regularization 1e-18.}, we find that each network $q^\theta_t$ converges to constant extreme points $\{-1,1\}$ (right sub-figure in \Cref{fig:DHstrats}), depending on which extreme point it had been closer to at initialization.
    \begin{figure}
\makebox[\textwidth][c]{%
		\includegraphics[scale=0.4]{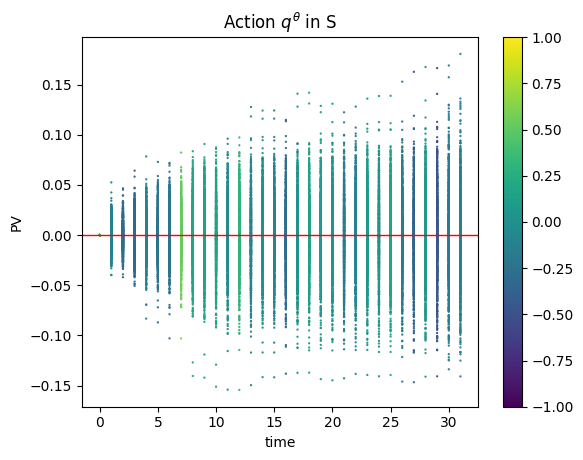}
		\includegraphics[scale=0.4]{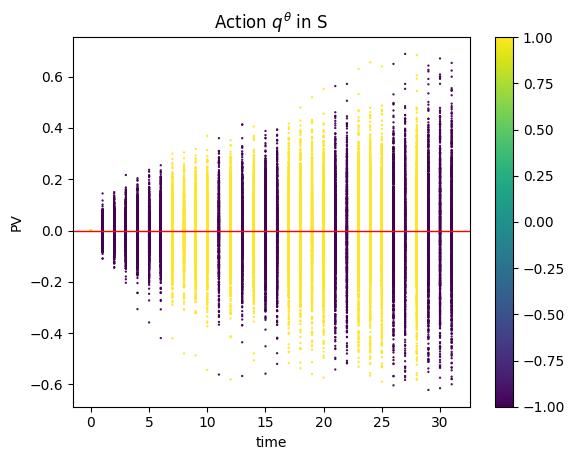}

}
    \caption{Evolution of portfolio values (PV) $X^{q^\theta}$ over time until maturity $T=32$ for actions taken according to a deep hedging strategy $q^\theta$ at initialization (left) and trained with a standard deep hedging algorithm \citep{buehler2019deep} (right) respectively, over 1000 asset paths in a BS market with one risky asset. The color indicates the trading action in $[-1,1]$ that the respective network $q^\theta$ takes in the risky asset.}
    \label{fig:DHstrats}
\end{figure}
\end{remark}

Approximating the optimal trading strategy to price a passport option via a standard deep hedging approach that searches the entire $\ell_1$-ball might be an ill-posed problem, however, this approach also doesn't use all information that the theory suggests.
In particular, recall that in \Cref{subsec:PricingThePassportOption}, we noted that for general asset dimensions $d$ optimal actions take values in the corner points $\Diamond$ of the $d$-dimensional $\ell_1$-ball.
Therefore, we are dealing with a classification problem with $2d$ possible actions.
Learning the optimal action is then commonly framed via learning a probability distribution over these finitely many possible actions.
Thus, pricing the passport option in the discrete-time BS market can be seen as solving a particular Markov decision process with \emph{continuous} state and \emph{finite} action spaces (with $2d$ possible actions).

\subsubsection{Our Approach}
For pricing passport options in this paper, we take the following approach. First (as is common practice for classification problems), we relax the pricing problem of \Cref{eq:MDPObjective} by moving to probabilistic actions in \Cref{sec:RelaxingTheProblem}. Second, we introduce a purely action-based approach to learning the optimal trading strategy for the relaxed problem in \Cref{sec:PGA}. We argue that this approach does not bear the risk of being trapped in local minima and discuss how to deal with noisy estimates involved in the algorithm. Third, in \Cref{sec:A2C}, we discuss how to use a popular action- and value-based approach for pricing passport options.

\subsection{Relaxing the Pricing Problem}
\label{sec:RelaxingTheProblem}
For numerically solving problem \eqref{eq:MDPObjective}, we consider the relaxed MDP with action space $\mathcal{P}({\Diamond})$, the probability measures on $\Diamond$, i.e.,

	 \begin{itemize}
	 	\item state space $\mathcal{X}:= \mathbb{R}^d_+\times\mathbb{R}$
	 	\item action space $\mathcal{P}({\Diamond}):=\{p\in[0,1]^{2d}:\sum_{j=1}p_j=1\}$,
    \item transition probabilities given by the discrete-time dynamics \begin{align*}
        X_0^\pi=x_0, \quad  X_{t_n}^\pi=X_{t_{n-1}}^\pi+\sum_{a\in\Diamond}\left(\pi_{t_n}(X_{t_{n-1}}^\pi, S_{t_{n-1}})(a) \sum_{i=1}^da^i\Delta S_{t_n}^i\right),\quad 1\le n\le N,
    \end{align*} and \eqref{eq:DiscreteTimeS}, 
   \item terminal reward $R(s, x):=|x|$, and
   
	 	\item policies $\pi=(\pi_{t_n})_{1\le n\le N}$, $\pi_{t_n}:\mathcal{X}\to \mathcal{P}({\Diamond})$.
	 \end{itemize}
In the relaxed problem, we search for the optimal policy $\pi^*$ with $\pi^*_{t_n}:\mathcal{X}\to\mathcal{P}({\Diamond})$ that solves 
	 	\begin{equation}\label{eq:RLObjective}\tag{RLO}\min_{\pi}-\mathbb{E}_{\Q}\left[|X^\pi_T|\right].\end{equation}

\begin{remark}
This kind of relaxation is genuine, it convexifies the set of actions, and the infimum over relaxed actions equals the one over strict actions \citep{ControlledDiffusions}. Thus, w.l.o.g. we may turn to numerically solving the relaxed pricing problem \eqref{eq:RLObjective} instead of \eqref{eq:MDPObjective}. Note that the convexification of the action space we introduce by moving to probability distributions differs from the original convexification where the action space was the $d$-dimensional $\ell_1$-ball.
Computationally, relaxing the problem in this way bears the benefit of not being trapped in local minima (cp. \Cref{re:problemwithdeephedging}). This can be seen since in each step, the expectation w.r.t. an action $p\in\mathcal{P}(\Diamond)$ is linear in the measure $p$. Thus, if one measure $p_1\in\mathcal{P}(\Diamond)$ leads to higher value than another measure $p_2\in\mathcal{P}(\Diamond)$, also any convex combination $\alpha*p_1+(1-\alpha)*p_2$ leads to higher value than $p_2$. Thus, a gradient algorithm will not get stuck at local optima.
\end{remark}
Problem \eqref{eq:RLObjective} admits a possibly time-dependent, a.k.a., inhomogeneous Markov solution $\pi^*$. We hence approximate the optimal policy $\pi^*$ by a NN $\pi^\theta:\R_+\times\mathcal{X}\to {\mathcal{P}(\Diamond)}$ mapping time and state space into a probability measure over possible actions, with $\pith_{t_n}(\cdot):=\pith(t_n,\cdot)$ and parameters $\theta$. In the following, we call $\pith$ \emph{strategy network}. The objective 
then is to solve
			
			\begin{equation}\label{eq:NNObjective}
		   \min_{\theta}-\mathbb{E}_{\Q}\left[|X^{\pith}_T|\right].
			\end{equation}
\begin{remark}[Relaxed Deep Hedging a.k.a. REINFORCE]\label{re:REINFORCE}
    Modeling a distribution over a discrete set of optimal actions is as simple as choosing the softmax activation $\phi(x):=\max(x,0), x\in\R$ for strategy networks $\pith$. In spirit of deep hedging, the parameters of strategy networks $\pith$ can then be optimized via a gradient method for \Cref{eq:NNObjective}. Via policy gradient theorems \citep{policy_gradient_learning} the gradient for \Cref{eq:RLObjective} can be reformulated as\footnote{We introduce the notation $\mathbb{E}_{\Q,\pith}\left[|X^{\pith}_T|\right]:=\mathbb{E}_{\Q}\left[|X^{\pith}_T|\right]$ to highlight that in the objectives of \Cref{eq:RLObjective,eq:NNObjective}, we also average w.r.t. the probabilistic actions. When training a NN estimate $\pith$ on the objective of \Cref{eq:NNObjective} with a gradient method, we are thus taking derivatives of the measures appearing in \Cref{eq:NNObjective}.}
    \begin{equation}\label{eq:PGthm}
		   \nabla_\theta\mathbb{E}_{\Q,\pith}\left[|X^{\pith}_T|\right]=\mathbb{E}_{\Q,\pith}\left[|X^{\pith}_T|\sum_{n=0}^{N-1}\nabla_\theta \log\pith_{t_{n}}\right].\
			\end{equation}
   Thanks to this reformulation, we can decouple the NN gradient from the MDP dynamics that appear in the measure on the l.h.s. of \Cref{eq:PGthm}. This reformulation allows for Monte Carlo (MC) estimation of the gradient on samples of assets $S$ and actions $\pith$. 
    When the same NNs $\pith$ are chosen in each time step (like in our setting), this algorithm corresponds to the classic REINFORCE algorithm \citep{policy_gradient_learning}. 
\end{remark}

A general challenge of algorithms that try to learn an optimal strategy for long-term horizons is to
determine the long-term consequences of actions at time points far from maturity. 
Both deep hedging and the relaxed REINFORCE algorithm, weigh the gradient of all actions in an episode with the terminal reward earned from trading according to the actions in that sequence. Instead of valuing every single action with the same (noisy, when MC-estimated) weight $\E_{\pith}[|X^{\pith}_T|]$ (cp. \Cref{eq:PGthm}), we try to get a more accurate per-action value estimate for every single action's state-action value.

An abundance of algorithms has been proposed throughout the RL literature to improve this temporal credit assignment.
In this paper, we employ two such approaches, both of which aim to learn the optimal strategy solving problem \eqref{eq:RLObjective}: we parametrize the trading strategy as a single feed forward NN and train this strategy network a) taking a specific policy gradient (see \Cref{sec:PGA}) and b) following a standard advantage actor critic (A2C) approach (see \Cref{sec:A2C}).
For both algorithms, we consider the relaxed problem \eqref{eq:RLObjective}.

Both algorithms iterate between \emph{evaluating} (E) the current NN policy $\pith$ and \emph{updating} (U) its parameters. As such, they both are examples of \textit{generalized policy iteration} as termed in
\citet{policy_gradient_learning}.

\subsection{A Policy Gradient Algorithm}\label{sec:PGA}
In this section, we consider learning the optimal policy $\pi^*$ of \Cref{eq:RLObjective} via a specific policy gradient (PG) algorithm (see \citep{DegrisPilarskiSutton} for an overview of PG algorithms). In our version of this approach, \Cref{alg:policyGrad}, we proceed as follows. Going backward in time (code line 3), we 
\begin{itemize}
    \item[(E)] collect noisy training data for that time point $t$ (code lines 4-7). We generate a noisy training data point for time $t$ in \Cref{alg:dataGen} as follows. First, we simulate a state $(s_t,x_t)$ with the current NN strategy $\pith$ (code line 4). Then, we determine at this state $(s_t,x_t)$ a state-action value for each of the $2d$ possible actions (code lines 6-7), and determine the action $a_t^*$ that maximizes the current estimate of state-action value (code lines 8-10). The state-action value is determined as a MC estimate of expected (discounted) terminal reward, given the initial state and action, and when trading according to the current NN policy in subsequent steps (code lines 6-7). We repeat this evaluation process a certain number of times and collect these (noisy) training data $((s_t,x_t),a_t^*)$ in a training data set $D_t$ (code lines 5-7 in \Cref{alg:policyGrad}).
    \item[(U)] Then, we greedily update the current NN strategy at this state (code lines 8-14 in \Cref{alg:policyGrad}). With the training data $D_t$ from step (E), we minimize the total variation distance of $\pith_t(s_t,x_t)$ and $\delta_{a_t^*}$, a Dirac delta at the optimal action $a_t^*$.
\end{itemize}

\begin{algorithm}[ht]
    \caption{\texttt{PG\textunderscore 4PPO}}
    \label{alg:policyGrad}
\begin{algorithmic}[1]
\STATE{\textbf{input:} \texttt{dppt} \COMMENT{no.\ training points per time step}, $B$ \COMMENT{no.\ MC paths}, $\gamma$ \COMMENT{learning rate}, epochs \COMMENT{no.\ epochs per time step}, b\textunderscore sizes \COMMENT{batch sizes per time step}, $T$ \COMMENT{terminal time}, \text{market\textunderscore args}}
\STATE $\pith = \texttt{initialize\textunderscore NN}()$
\FOR{$t= T-1$ to $1$}
    \STATE\COMMENT{{\bf(E) }collect dppt[t] noisy training data at time $t$:}
    
    \FOR{$d=1$ to dppt[t]}
    \STATE 
    $D_t = \texttt{data\textunderscore gen}(\pith, B, t, T, \text{market\textunderscore args}) $
    \ENDFOR
    \STATE
    \COMMENT{{\bf(U)} update parameters of $\pith$ by minimizing total variation distance TV between $\pi^\theta$ and training data actions}
    \STATE batches = \texttt{split2batches}(data=$D_t$, batch\textunderscore size=b\textunderscore sizes[t])
    \FOR{$e=1$ to epochs[t]}
       \FOR{$B$ in batches}
        \STATE $\theta = \theta -\gamma\nabla\frac{1}{|B|}\sum_{((x_t,s_t),a_t)\in B}\text{TV}(\delta_{a_t^*},\pith(s_t,x_t))$
        \ENDFOR
        \ENDFOR
\ENDFOR
\end{algorithmic}
\end{algorithm}

\begin{algorithm}[ht]
    \caption{\texttt{data\textunderscore gen}}
    \label{alg:dataGen}
\begin{algorithmic}[1]
\STATE \textbf{input:} $\pith$ \COMMENT{current strategy network}, $B$ \COMMENT{no.\ MC paths}, $t$ \COMMENT{current time}, $T$ \COMMENT{terminal time}, \texttt{market\textunderscore args}
\STATE {terminal payoff $R:=0$}
\WHILE{$R == 0$}
   \STATE {$x_t, s_t$ = \texttt{sample\textunderscore state}$(t, \pith, \text{market\textunderscore args})$}
     \FOR{$a$ in $\Diamond$}
         \STATE{$x_{T}, s_{T} =$ \texttt{sample\textunderscore state}$(T, x_t,s_t,a, \pith, B, \text{market\textunderscore args})$}
        \STATE {$r = \texttt{mean}(|x_{T}|)$}
        \IF{$r>R$}
            \STATE {$R = r$, $a_t^* = a$}
        \ENDIF
     \ENDFOR
\ENDWHILE
\STATE \textbf{return:} $x_t,s_t,a_t^*$
\end{algorithmic}
\end{algorithm}

The algorithm, summarized in \Cref{alg:policyGrad}, shares elements with Monte Carlo Tree Search (expansion and simulation phases coincide, however, the selection of tree nodes is very specific (backward in time) and backpropagation affects all nodes simultaneously), and Least Square Monte Carlo (we use a recursive scheme backward in time, however, we do not estimate value functions). Moreover, it can be seen as a policy gradient, where we greedily update the strategy network in the direction of the optimal action, instead of an average direction weighted by estimated state-action values.
Furthermore, we discuss an alternative, probabilistic view on \Cref{alg:policyGrad} in the following \Cref{re:probabilisticInference}.
\begin{remark}[Probabilistic Inference]\label{re:probabilisticInference}
In spirit of \cite{levine2018reinforcement}, one can also view \Cref{alg:policyGrad} as performing probabilistic inference to approximate the optimal (in terms of \Cref{eq:RLObjective}) distribution $p_\tau$ of a trajectory $\tau=(S_0, X_0,a_0, S_1, X_1, a_1,\ldots,S_T,X_T)$, with
\begin{align*}
    p_\tau(A_{t_0},a_{t_1},\ldots,A_{t_N})=\Q(S_{t_0},X_{t_0} \in A_{t_0})&\prod_{n=1}^N \bigg(\delta_{a_{t_{n}}^*(S_{t_{n-1}},X_{t_{n-1}})}(a_{t_{n}})\\
    &\cdot\Q(S_{t_n},X_{t_n} \in A_{t_n} \mid S_{t_{n-1}},X_{t_{n-1}}a_{t_n})\bigg).
\end{align*}
Here, $a_{t_n}^*$ denotes the optimal action at time point $t_n$, in the sense that for $n\in\{1,\dots, N\},$
\begin{equation}\label{eq:optimalaction}
    a_{t_n}^*(s_{t_{n-1}},x_{t_{n-1}}):=\argmax_a\mathbb{E}_{p_\tau}\left[|X_{T}|\mid s_{t_{n-1}}, x_{t_{n-1}}, a_{t_n}=a\right].
\end{equation}

More specifically, we introduce the parametric distribution
\begin{align*}
    p^\theta_\tau(A_{t_0},a_{t_1},\ldots,A_{t_N})=\Q(S_{t_0},X_{t_0} \in A_{t_0})&\prod_{n=1}^N \bigg(\pith_{t_n}(S_{t_{n-1}},X_{t_{n-1}})(a_{t_{n}})\\
    &\cdot\Q(S_{t_n},X_{t_n} \in A_{t_n} \mid S_{t_{n-1}},X_{t_{n-1}},a_{t_n})\bigg),
\end{align*}
and choose the parameters $\theta$ to minimize the total variation distance
\begin{align*}
    D_{TV}(p_\tau,p^\theta_\tau)&=\mathbb{E}_{p^\theta_\tau}\left[\left|1-\frac{dp_\tau}{dp^\theta_\tau}\right|\right]\\
    &=\mathbb{E}_\Q\left[\left|\prod_{n=1}^N \pith_{t_n}(s_{t_{n-1}},x_{t_{n-1}})(a_{t_{n}})-\prod_{n=1}^N \delta_{a_{t_{n}}^*(s_{t_{n-1}},x_{t_{n-1}})}(a_{t_{n}})\right|\right]
\end{align*}
where with slight notational overload
\begin{equation}\label{eq:dtv}
    \mathbb{E}_\Q\left[\cdot\right]:=
    \int\sum_{a_{t_0}\in\Diamond}\dots\int(\cdot)\,d\Q(s_{t_N},x_{t_N}|x_{t_{n-1}},x_{t_{n-1}},a_{t_N})\dots \,d\Q(s_{t_0},x_{t_0}).
\end{equation}
\Cref{eq:dtv} attains its minimum at \begin{align*}
    \pith_{t_n}(s_{t_{n-1}},x_{t_{n-1}})(a_{t_{n}})=\delta_{a_{t_{n}}^*(s_{t_{n-1}},x_{t_{n-1}})}(a_{t_{n}}),\quad\forall n=1,\ldots N.
\end{align*} We hence
\begin{equation*}
    \min_\theta \mathbb{E}_\Q\left[\left| \pith_{t_n}(s_{t_{n-1}},x_{t_{n-1}})({a}_{t_{n}})- \delta_{\hat{a}_{t_{n}}^*(s_{t_{n-1}},x_{t_{n-1}})}(a_{t_{n}})\right|\right]
\end{equation*}
for approximate/noisy targets
\begin{equation}\label{eq:optimalactionestimate}
    \hat{a}_{t_n}^*(s_{t_{n-1}},x_{t_{n-1}}):=\argmax_a\frac{1}{|B|}\sum_{b\in B}\left|X^{\pith}_{T}(\omega_b\mid s_{t_{n-1}}, x_{t_{n-1}}, a_{t_n}=a)\right| , n=1,\dots, N,
\end{equation}
where for every $b$ in batch $B$, $X^{\pith}_{T}(\omega_b\mid s_{t_{n-1}}, x_{t_{n-1}}, a_{t_n}=a)$ denotes a sample of terminal portfolio value under $p^\theta_\tau$ conditioned on choosing action $a$ in state $s_{t_{n-1}}, x_{t_{n-1}}$ at time $t_{n-1}$.
We proceed backwards in time, i.e., for $n$ from $N$ to $1$, in order to train on targets $\hat{a}_{t_n}^*$
 with as little noise as possible. 
\end{remark}

The targets that \Cref{alg:policyGrad} obtains in the evaluation step (E) (code line 5) and then trains on in the updating step (U) (code line 7) are highly noisy estimates of truly optimal actions. In \Cref{re:noiseintargets}, we discuss the different types of noise that occur within this algorithm and describe how we handle them. 
    
\begin{remark}[Mitigate the Noise in Targets]\label{re:noiseintargets}
    The sources of noise in \Cref{alg:policyGrad} are multiple.
    \begin{enumerate}
        \item     First, the MC simulation that is used to estimate conditional expectations introduces noise. Typically, for a fixed strategy, MC errors can be controlled with a moderate sample size.
        \item A second source of noise is introduced by estimating continuation values, i.e., the conditional expectation in \Cref{eq:optimalactionestimate}, based on approximations $\pith$ of the optimal strategy $\pi^*$. Going backward in time is thus crucial to best as possible mitigate errors caused by sub-optimal continuation trades. 

    \end{enumerate}
        We employ several further regularization techniques to deal with noisy classification data. Especially for increasing dimensions, we introduce \emph{entropy regularization} to prevent producing over-confident predictions.\footnote{The entropy $H$ of a prediction $p\in\mathcal{P}(\Diamond)$ is defined as $H(p):=-\sum_{i=1}^{2d}p^i\log(p^i)$. Entropy is maximized for uniform distributions $p^i=1/2d$ for all $i=1,\ldots, 2d$. Entropy regularization then subtracts entropy as a regularization term to the objective of \Cref{eq:RLObjective} in order to favor action diversity.} Moreover, note that the \emph{total variation (TV) distance} between the NN's predicted optimal actions and the noisy optimal targets acts as regularizing loss function: in \Cref{fig:KL_vs_TV}, we see that compared to the popular \emph{Kullback Leibler (KL) divergence}, the TV distance does not punish as harshly network predictions far from the observed (noisy) probability. In this simple example in \Cref{fig:KL_vs_TV}, we consider classification with two actions $a$ and $b$, and $p$ estimates the probability of taking the first action $a$. We assume further that we have a training data point $x_{\text{tr}}=a$ that we transform to a one-hot encoded training data point $(p_{\text{tr}}, 1-p_{\text{tr}})=(1,0)$. Assume then that our estimate $p$ for the probability of action $a$ is close to zero, meaning it assigns a high probability to action $b$. When we do an update of our estimate $p$ at towards our observed training point $p_{\text{tr}}=1$ we clearly see in \Cref{fig:KL_vs_TV} that a gradient w.r.t.\ the TV loss is not as steep as one w.r.t.\ the KL divergence. Thus, if $b$ were the optimal action a gradient step would push less strongly towards the (wrongly observed) noisy target action $a$ with the TV loss as it would with the KL divergence.

        \begin{figure}[h!]
            \centering
            \includegraphics[scale=.5]{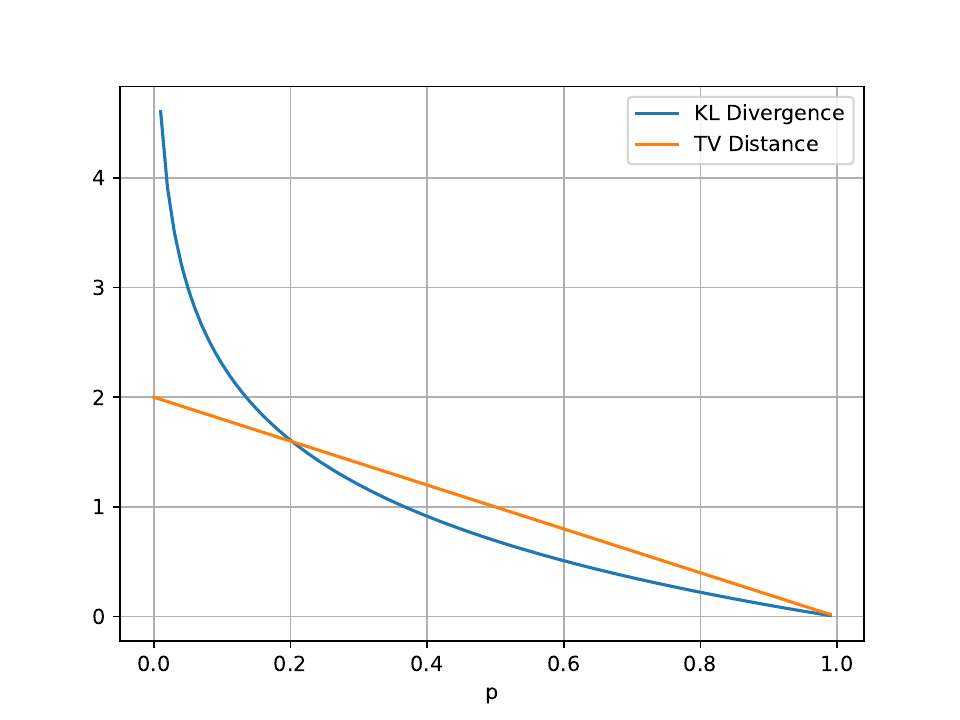}
            \caption{KL-divergence and TV distance for the observed probability $p_{\text{tr}}=1$, i.e., we plot d\textsubscript{KL}$((p,1-p),(1,0))$ and d\textsubscript{TV}$((p,1-p),(1,0))$.}
            \label{fig:KL_vs_TV}
        \end{figure}

\end{remark}

\begin{remark}[Beyond Markovian Dynamics]
Although presented for a Black-Scholes market, the policy gradient algorithm of \Cref{sec:PGA} can be applied to more general, and in particular to non-Markovian settings.
\end{remark}

\subsection{An A2C Approach}\label{sec:A2C}
In this section, we discuss an approach to learning the optimal strategy for our pricing problem \eqref{eq:RLObjective} that combines elements of both action-based and value-based approaches:
the popular \emph{advantage actor critic (A2C)} algorithm. 

In action-based policy gradient algorithms, we train the NN policy's parameters $\theta$ with a gradient method. 
Purely action-based algorithms (such as e.g. the original REINFORCE algorithm of \Cref{re:REINFORCE}
) 
are however known to suffer from high variance \citep{surveyA2C}. To reduce this variance, it is common practice to consider other unbiased estimates of the gradient that frequently depend on the state value function \citep{VarianceReductionA2C}.
This value function can itself be approximated by a NN and is trained to function as critic in the training of strategy networks.  

In A2C, the unbiased estimate of the value function is given by the \emph{advantage function} (see \Cref{eq:advantagefunction}) that tells how much value can be gained or lost in one step by choosing a specific action. In the following
\Cref{le:AdvantagepolicyGrad}, we state the unbiased gradient for updating the policy in the A2C algorithm for our objective \eqref{eq:RLObjective}. 

\begin{lemma}[Advantage Policy Gradient for Continuous State Space]\label{le:AdvantagepolicyGrad} Consider the relaxed pricing problem \eqref{eq:RLObjective}. It holds that
\begin{equation}\label{eq:AdvantagepolicyGrad}
    \nabla_\theta \mathbb{E}_{\pthtau}\left[|X^{\pith}_T|\right]= \mathbb{E}_{\pthtau}\left[\sum_{n=0}^{N-1}{\nabla_\theta \log\pith(t_{n+1}, S_{t_{n}}, X^{\pith}_{t_{n}})(\atn)}\,A_{t_{n}}(S_{t_{n}}, X_{t_{n}}, \atn)\right],
\end{equation}
with advantage function \begin{equation}\label{eq:advantagefunction}A_{t_{n}}(s,x,a):=\mathbb{E}_{\pthtau}\left[|X^{\pith}_T|\mid S_{t_{n}}=s, X_{t_{n}}^{ \pith}=x, a_{t_{n}}=a\right]-V_{t_{n}}(s,x),\end{equation}
and value function
\begin{equation}\label{eq:valuefunction}
    V_{t_{n}}(s,x):=\mathbb{E}_{\pthtau}\left[|X^{\pith}_T|\mid S_{t_{n}}=s, X_{t_{n}}^{ \pith}=x\right].
\end{equation}
Here $\pthtau$ denotes the path measure on $\left\{\R^{d+1}\times\Diamond\right\}^N$ where $S$ and $X$ evolve according to $\Q$ and when trading with strategy $\pith$ (see \Cref{re:probabilisticInference}).
\end{lemma}

\begin{proof}
For the original proof in an infinite horizon setting see \citep{policy_gradient_learning}.
\end{proof}
Algorithms that make use of a gradient as in \Cref{eq:AdvantagepolicyGrad} to update the NN policy's parameters are termed advantage actor critic (A2C).
In A2C with function approximation, the state value function \eqref{eq:valuefunction} appearing in the advantage function of \Cref{le:AdvantagepolicyGrad} is parametrized by a NN $V^\phi$, i.e., we model
\begin{align*}
    A_{t_{n}}^\phi(s,x,a):=\mathbb{E}_{\pthtau}\left[|X^{\pith}_T|\mid S_{t_{n}}=s, X_{t_{n}}^{ \pith}=x, a_{t_{n+1}}=a\right]-V_{t_{n}}^\phi(s,x),
\end{align*}
with $\Vph:\R_+\times\mathcal{X}\to \R$ and $\Vph_{{t_{n}}}(\cdot)=\Vph({t_{n}},\cdot)$ and parameters $\phi$.
To learn both $\Vph$ and $\pith$, we then sample $B$ trajectories $(s_0^\omega,a_0^\omega,_0^\omega,\ldots,s_T^\omega,x_T^\omega)$ according to $\pthtau$ based on the current NN strategy and iteratively
\begin{itemize}
    \item[(E)]\begin{align*}
    \min_\phi \frac{1}{B}\sum_\omega\sum_{n=0}^N ( A_{t_{n}}^\phi(s_{t_{n}}^\omega,x_{t_{n}}^\omega,a_{t_{n+1}}^\omega))^2 
\end{align*}
    \item[(U)] \begin{align*}
    \min_\theta -\frac{1}{B}\sum_\omega\sum_{n=0}^N \left[\sum_{a\in\Diamond}{\log\pith_{t_{n+1}}(s_{t_{n}}^\omega, x_{t_{n}}^\omega)(a_{t_{n+1}}^\omega)}\,A_{t_{n}}^\phi(s_{t_{n}}^\omega,x_{t_{n}}^\omega,a)\right].
\end{align*}
\end{itemize}

As a side benefit, this algorithm yields estimates $\Vph$ of the value process, (i.e., the state value functions $V_{t_n}, n=0,\ldots,N$) for pricing the passport option. Thus, prices for the passport option can be conveniently estimated by evaluating the state value network $\Vph$ at time $t=0$, instead of computing an additional MC estimate based on the estimated optimal strategy $\pi^{\theta^*}$.

A detailed description of the A2C algorithm for pricing the passport option is given in \Cref{alg:A2C,alg:A2C_forward}. It utilizes a financial market environment (\texttt{market\textunderscore env}) in which asset prices and the agent's portfolio value evolve. After initializing strategy and value networks in code lines 2 and 3 of \Cref{alg:A2C}, a virtual agent then trades in the market until terminal time, based on the current strategy network $\pith$ and collects terminal reward, value functions and log-probabilities for each of the $|B|$ paths (\Cref{alg:A2C_forward}, called in code line 6 of \Cref{alg:A2C}). Based on these, MC estimates of the advantage functions, of the policy gradient, and of the gradient for updating the strategy network are computed (code lines 7-11). These gradients are then used to update both the strategy and the value network's parameters $\theta$ and $\phi$ (code lines 12,13). 
Moreover, we use entropy regularization in the training step (U) to keep the NN strategy closer to a uniform encouraging exploration (following, e.g., \citep{mnih2016asynchronous}).

A number of libraries have been created that implement such financial market environments\footnote{See, e.g., \url{http://finrl.org/} or the recent project \url{https://github.com/PawPol/PyPortOpt} from researchers at Stony Brook University.}. Many of them also include pipelines for popular RL algorithms. For the experiments in this paper, i.e., for \Cref{sec:experiments}, we built our own software.\footnote{See \url{https://github.com/HannaSW/ML4PassportOptions} for the corresponding code.}

\begin{algorithm}[ht]
    \caption{\texttt{A2C\textunderscore 4PPO}}
    \label{alg:A2C}
\begin{algorithmic}[1]
\STATE \textbf{input:} {market\textunderscore env} \COMMENT{market environment}, {niter} \COMMENT{number of iterations}, $B$ \COMMENT{number of paths per iteration}, $\tau$ \COMMENT{regularization parameter}, $\gamma$ \COMMENT{discount factor}
\STATE $\pi^{\theta_0} = \texttt{initialize\textunderscore NN\textunderscore actor}()$
\STATE $V^{\phi_0} = \texttt{initialize\textunderscore NN\textunderscore critic}()$
\FOR{$k=0$ to niter-1}
\STATE $s_0,x_0 = $ \texttt{market\textunderscore env.reset}()
\STATE {critics}, {log\textunderscore pis}, $e, x_T$ = \texttt{forward}({market\textunderscore env}, $\pi^{\theta_k}, V^{\phi_k}, s_0, x_0, \gamma$) \COMMENT{(E) step}
\FOR{$t=1\ldots,T$}
    \STATE $A_t = \gamma^{T-t}|x_T| -$ {critics}$[t]$ \COMMENT{compute advantages}
\ENDFOR
    \STATE actor loss $=-\frac{1}{B}\sum\left( \frac{1}{T}\sum_{t=0}^T\texttt{log\textunderscore pis}[t]A_t-\tau e\right)$
    \STATE critic loss $=\frac{1}{B}\sum \frac{1}{T}\sum_{t=0}^T\left(A_t\right)^2 $
    \STATE$\pi^{\theta_{k+1}} = \texttt{train\textunderscore NN}$($\theta_k$, actor loss) \COMMENT{(U) step}
    \STATE$V^{\phi_{k+1}} = \texttt{train\textunderscore NN}$($\phi_k$, critic loss)
\ENDFOR
\end{algorithmic}
\end{algorithm}

\begin{algorithm}[ht]
    \caption{\texttt{forward}}
    \label{alg:A2C_forward}
\begin{algorithmic}[1]
\STATE \textbf{input:} market\textunderscore env, $\pi^{\theta}, V^{\phi}, s_0, x_0, \gamma$ \COMMENT{discount factor}
\STATE $e=0$ \COMMENT{entropy regularization term}
\STATE critics $ = []$
\STATE log\textunderscore pis $= []$
\FOR{$t= 0$ to $T-1$}
    \STATE sample action $a_t \sim \pith(s_t, x_t)$
    \STATE $s_{t+1}, x_{t+1} = $ market\textunderscore env.\texttt{step}$(s_t, x_t, a_t)$
    \STATE critics $=$ critics $\cup\, V^\phi(s_{t},x_t)$ \COMMENT{collect critics}
    \STATE log\textunderscore pis $=$ log\textunderscore pis $\cup\log\pith(s_t,x_t)(a_t))$ \COMMENT{collect log-pis}
    \STATE $e = e+ \gamma\sum_a\pith(s_t)(a)\log\pith(s_t,x_t)(a)$ \COMMENT{update entropy}
\ENDFOR
    \STATE critics $=$ critics $\cup\, V^\phi(s_{T},x_T)$
\STATE \textbf{return:} critics, log\textunderscore pis, $e, x_T$
\end{algorithmic}
\end{algorithm}

\subsubsection{The Challenge of Fine Time Grids}\label{subsubsec:FinerTimegrids}
Ideally, we would like to select a fine discrete-time grid in order to best as possible  approximate a continuous time solution with DL.
However, algorithms involving a policy gradient suffer from variance explosion for vanishing time steps \citep{JMLR:v7:munos06b,NEURIPS2021_024677ef}.

In particular, we show in \Cref{le:VarianceA2C} below that even the variance of the gradient (\Cref{eq:AdvantagepolicyGrad}) in the \emph{variance-reduced} A2C algorithm scales at least linearly with the number of time steps. We show this by giving a formal lower bound on the trace of the variance of the A2C gradient \Cref{eq:AdvantagepolicyGrad} that is linear in the number of time steps $N$ (see the proof of \Cref{le:VarianceA2C}). 

\begin{lemma}\label{le:VarianceA2C}
    If $\pith$ has high entropy, i.e., if there is some $\epsilon>0$ s.t. $|\pith(a)-0.5|<\epsilon)$ for all $a\in\Diamond$, then, the variance of the advantage policy gradient from \Cref{eq:AdvantagepolicyGrad}
    \begin{align*}
        C(N):=\V_{\pthtau}\left[\sum_{n=0}^{N-1}{\nabla_\theta \log\pith(t_{n+1}, S_{t_{n}}, X^{\pith}_{t_{n}})(\atn)}\,A_{t_{n}}(S_{t_{n}}, X_{t_{n}}, \atn)\right]
    \end{align*}
    scales at least linearly in the number of time steps $N$ in the time discretization, i.e., $C(N)\in\Omega(N)$ in big-omega notation.
\end{lemma}
\begin{proof}
    As in \citep{NEURIPS2021_024677ef}, we bound by below the trace of the covariance matrix $C(N)$ by considering the first entry of the gradient from \Cref{eq:AdvantagepolicyGrad}, i.e., the derivative $\nabla_b$ w.r.t. the last layer bias $b$ of the strategy network $\pith$ at output $e_1\in\Diamond$.
    Let $\tau=(s_0, x_0,a_0, s_1, x_1, a_1,\ldots,s_T,x_T)$ be a sampled trajectory from distribution $\pthtau$. We denote by $\NNtha$ the pre-activated output (i.e., before applying the softmax activation) of $\pith$ at the component corresponding to $a$. Then we have (writing $\pith$ for $\pith(t_{n+1}, s_{t_{n}}, x^{\pith}_{t_{n}})$ to ease notation)
    \begin{align*}
        \nabla_b \log(\pith)(\atn)&=\frac{\nabla_b\pith(\atn)}{\pith(\atn)}\\
        &=\frac{1}{\pith(\atn)}\frac{\nabla_b\NNthatn e^{\NNthatn}\sum_a e^{\NNtha}-\nabla_b\NNthatn e^{\NNthatn}}{\left(\sum_a e^{\NNtha}\right)^2}\\
        &=\frac{1}{\pith(\atn)}\pith(\atn)\nabla_b\NNthatn\frac{\sum_a e^{\NNtha}-e^{\NNthatn}}{\sum_a e^{\NNtha}}\\
        &=\left\{\begin{matrix}
            1-\pith(\atn), & \atn=e_1\\
            0, & \atn\neq e_1.
        \end{matrix},\right.\\
         &=\left\{\begin{matrix}
            1-\pith(e_1), & \text{ with probability } \pith(e_1)\\
            0, & \text{ with probability } 1-\pith(e_1).
        \end{matrix}\right.
    \end{align*}
    Thus, $\nabla_b\log(\pith)
    (\atn)$ is a discrete random variable that attains with probability $\ptn:=\pith(e_1)$ the positive value $ 1-\ptn$.
    With this, we get that
    \begin{align*}
        \mathrm{trace}(C)(N)&\ge \V_{\pthtau}\left[\sum_{n=0}^{N-1}{\nabla_b \log\pith(t_{n+1}, S_{t_{n}}, X^{\pith}_{t_{n}})(\atn)}\,A_{t_{n}}(S_{t_{n}}, X_{t_{n}}, \atn)\right]\\
        &=\V_{\pthtau}\left[\sum_{n=0}^{N-1}{\1_{\{\atn=e_1\}}} (1-\ptn)\,A_{t_{n}}(S_{t_{n}}, X_{t_{n}}, \atn)\right]\\
        &\ge\E_{S}\left[\V_{\pith}\left[\sum_{n=0}^{N-1}{\1_{\{\atn=e_1\}}} (1-\ptn)\,A_{t_{n}}(S_{t_{n}}, X_{t_{n}}, \atn)\mid S_T,\ldots,S_0\right]\right],
    \end{align*}
    where the last inequality follows by the law of total variance.
    Furthermore, conditioned on $S_T,\ldots,S_0$, the random variables ${\1_{\{\atn=e_1\}}} (1-\ptn)\,A_{t_{n}}(S_{t_{n}}, X_{t_{n}}, \atn)$ are independent and we get
    \begin{align*}
         \mathrm{trace}(C)(N) &\ge\E_{S}\left[\sum_{n=0}^{N-1}\V_{\pith}\left[{\1_{\{\atn=e_1\}}} (1-\ptn)\,A_{t_{n}}(S_{t_{n}}, X_{t_{n}}, \atn)\mid S_T,\ldots,S_0\right]\right]\\
         &=\E_{S}\left[\sum_{n=0}^{N-1}(1-\ptn)^2A_{t_{n}}(S_{t_{n}}, X_{t_{n}}, e_1)^2(1-\ptn)\ptn\right]\\
         &\ge\sum_{n=0}^{N-1} (0.5-\epsilon)^4\underbrace{\E_S\left[A_{t_{n}}(S_{t_{n}}, X_{t_{n}}, e_1)^2\right]}_{=:A>0}\\
         &=N(0.5-\epsilon)^4A\in\mathcal{O}(N).
    \end{align*}
\end{proof}

\begin{figure}[h!]
\makebox[\textwidth][c]{%
 \subfloat[\label{subfig:gradientevol}]{%
		\centering\includegraphics[scale=0.1]{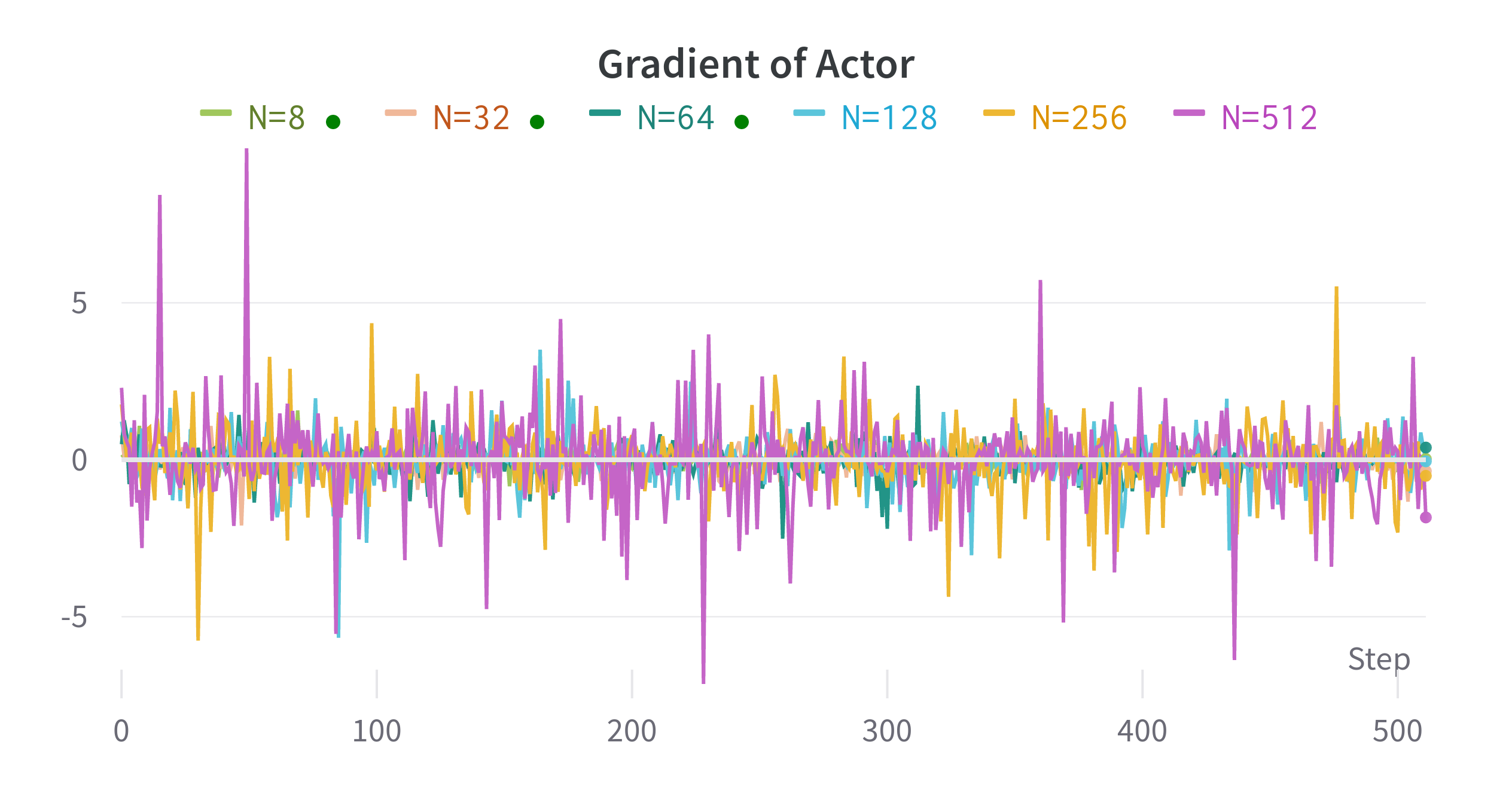}
		}%
  }
  \makebox[\textwidth][c]{%
    \subfloat[\label{subfig:linearregression}]{%
    \includegraphics[scale=0.1]{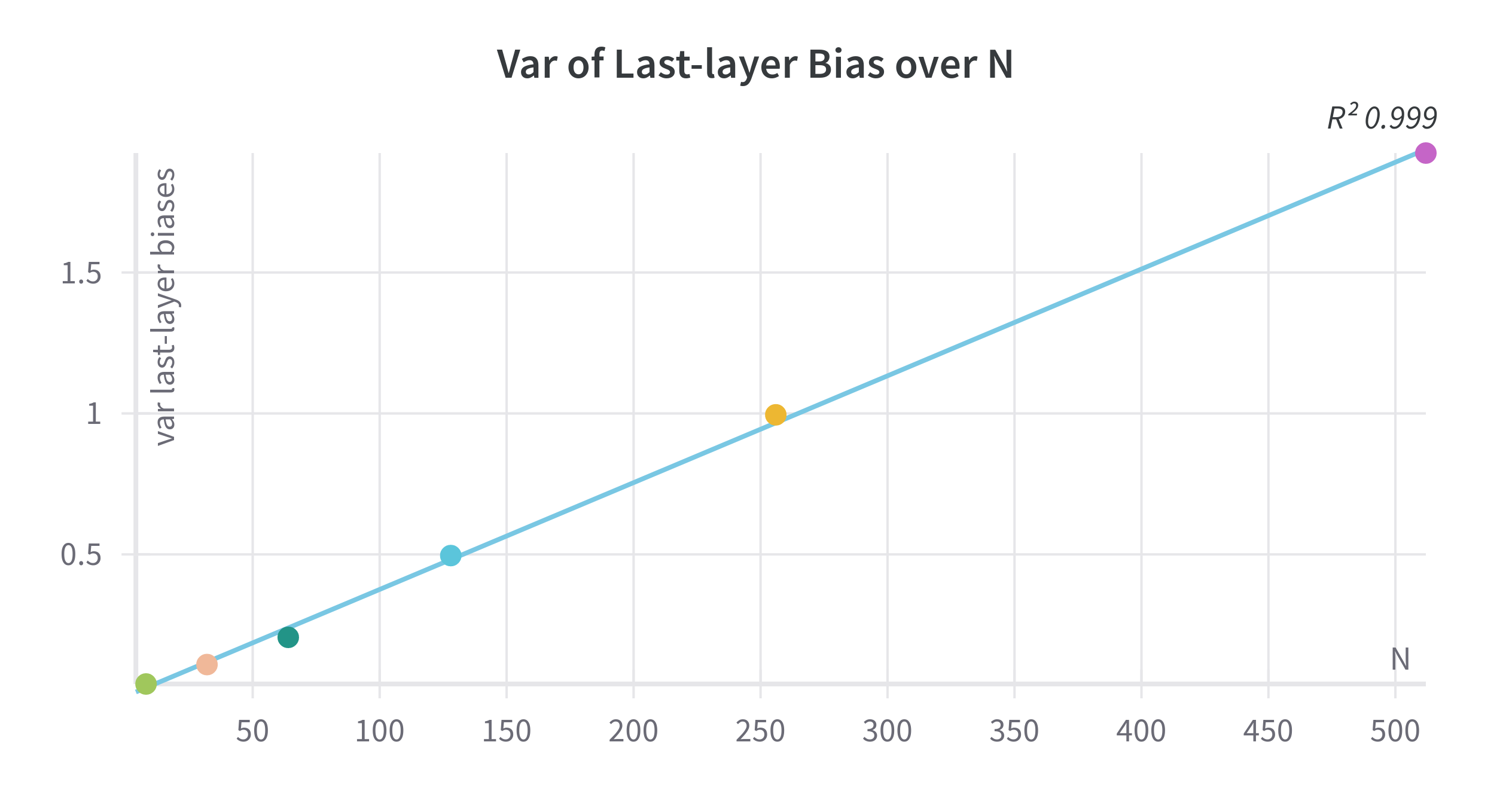}
	
   }
}
    \caption{(a) Gradient of actor network w.r.t. last-layer bias for different time discretizations $N$. (b) Standard deviations for these gradients across time steps over time discretizations $N$.}
    \label{fig:gradientevolution}
\end{figure}
\begin{remark}
     In \Cref{le:VarianceA2C}, we proved a lower bound on the gradient's variance for NNs that output distributions with high entropy across the input space. Such distributions do generally occur during the training process. First, for standard parameter initializations, strategy networks start with high entropy. Moreover, entropy regularization very often is included to keep the network from producing over-confident predictions too fast. Thus also during the training, there is a tendency for strategy networks to stay at high entropy.  
\end{remark}

In the proof of \Cref{le:VarianceA2C}, we show that the A2C gradient with respect to the last-layer bias grows linearly in the number of time steps $N$. We also observed this fact in our experiments as \Cref{fig:gradientevolution} illustrates. In this experiment, we tracked the gradient of the strategy network $\pith$ w.r.t. one terminal-layer bias (i.e., the gradient considered in the proof of \Cref{le:VarianceA2C}) during 512 iterations of the A2C \Cref{alg:A2C} (for $|B|=1$ path), for a varying number of time steps $N$. 
\Cref{subfig:gradientevol} shows the evolution of these gradients over training iterations for $N=2^k, k\in\{3,5-9\}$. We observe that large spikes in gradients become more frequent with decreasing step size $1/N$. \Cref{subfig:linearregression} confirms the linear fit through sample estimates of variances of these gradients over the 512 training steps.

Besides introducing higher variance, a finer time grid worsens exploration and increases the issue of temporal credit assignment \citep{NEURIPS2021_024677ef} (even for the A2C approach). 




\section{Experiments}\label{sec:experiments}
In this section, we experimentally evaluate our proposed algorithms PG 
(\Cref{alg:policyGrad,alg:dataGen}) and A2C (\Cref{alg:A2C}) for pricing passport options in one- and two-dimensional BS markets. \Cref{subsec:1Dexperiments} and treats the one-dimensional case, where we show that both algorithms recover the well-known solution of \Cref{eq:1Dsolution}. In \Cref{subsec:MDexperiments} we test the algorithms on the case of two uncorrelated assets and find that, also in this setting, they recover the solution derived in \Cref{thm:optimalStrat} in \Cref{subsubsec:2Duncorrelated}. We then move to grounds where no solution is known so far and analyze the strategies learned by our proposed PG and A2C algorithms of \Cref{sec:MLApproaches} in BS markets with two correlated assets (\Cref{subsubsec:2Dcorrelated}).
Detailed configurations of the experiments conducted in this section and the code used to run them can be found under \url{https://github.com/HannaSW/ML4PassportOptions}.
\subsection{The 1d Case}\label{subsec:1Dexperiments}
In a BS market consisting of $d=1$ risky assets, pricing (and hedging) the passport option is well understood. In this case, the optimal trading strategy, i.e., the solution to problem \eqref{eq:RLObjective}, is to go short when ahead (when the portfolio value $x>0$), and long when behind (when the portfolio value $x<0$), i.e., formally,
\begin{align}\label{eq:1Dsolution}
    q^*(t, s, x) = -\sign(x),
\end{align}
for all $(t,s,x)\in\R_+\times\R_+\times\R$. In particular, the strategy is independent of time $t$ and the asset's value $s$. 
Next, we investigate how the PG and A2C algorithms proposed in \Cref{sec:PGA,sec:A2C} perform in this well-known single-asset market environment.

\begin{figure}[t]
\makebox[\textwidth][c]{%
		\includegraphics[scale=0.4]{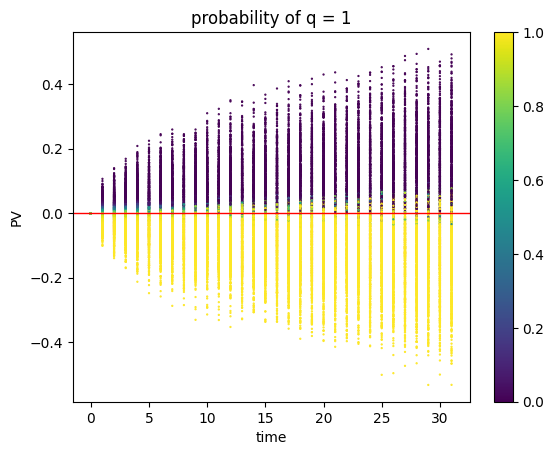}
		\includegraphics[scale=0.4]{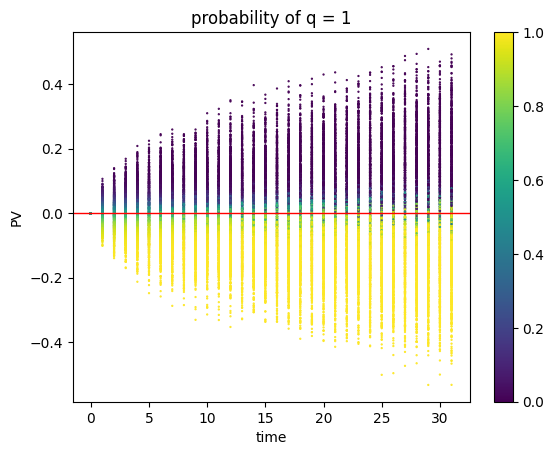}

}
    \caption{Evolution of portfolio values (PV) $X^{\pith}_t$ over time $t$ until maturity $T=32$ for actions taken according to $\pith$ trained with PG (\Cref{alg:policyGrad}) (left) and A2C (\Cref{alg:A2C}) (right) respectively, over a test set of 1000 asset paths. The probability that the respective network $\pith$ assigns to taking action $q=1$ is shown in color.}
    \label{fig:1Dstrats}
\end{figure}

\paragraph{How do ML-powered Pricing Approaches Perform?}
A sensible DL approach to pricing the passport option should be able to replicate the well-known analytical solution from \Cref{eq:1Dsolution}.
We test both algorithms of \Cref{sec:MLApproaches}, i.e., the PG algorithm (\Cref{alg:policyGrad,alg:dataGen}) from \Cref{sec:PGA} and the standard A2C algorithm (\Cref{alg:A2C}) from \Cref{sec:A2C}, on a BS market with risk free rate of return $r=0.2\%$, volatility $\sigma=20\%$ and initial capital $x=0$ over $T=32$ trading days.

\Cref{fig:1Dstrats} shows evolutions of portfolio values $X^{\pith}_t$ over time $t$ until maturity $T=32$, for a test set of 1000 asset paths and when actions are taken based on the trained network strategies $\pith$. In each of the sub-figures, the color coding shows the probability that the respective trained network strategy $\pith$ assigns to taking action $q=1$. We observe that both ML approaches yield strategies that assign a high probability to action $q=1$ when the portfolio value (PV) $X^{\pith}_t$ is negative, and that they output almost zero probability for taking action $q=1$ when PV is positive. Thus, PG and A2C both manage to capture the optimal strategy $ q^*(t, s, x) = -\sign(x)$.

\begin{figure}[t]
\makebox[\textwidth][c]{%
 \subfloat[MC estimate of {$e^{-rT}\mathbb{E}{[(X^{\pith}_T)^{+}]}$} with strategies $\pith$ obtained from \Cref{alg:policyGrad}.\label{subfig:1DvaluesPG}]{%
		\includegraphics[scale=0.4]{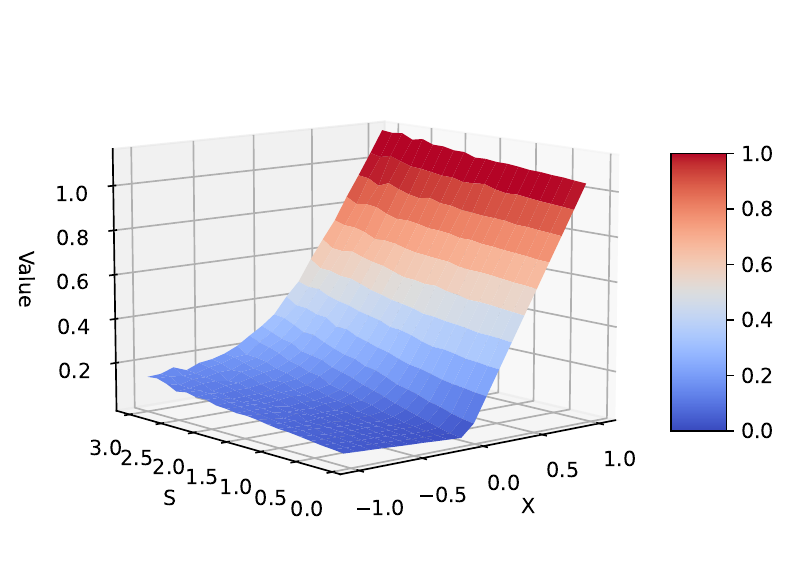}
  }
  }
  \makebox[\textwidth][c]{%
   \subfloat[MC estimate of {$e^{-rT}\mathbb{E}{[(X^{\pith}_T)^{+}]}$} with strategies $\pith$ obtained from \Cref{alg:A2C} (left), and price surface $(\Vph_0(s,x)+x)/2$ corresponding to the trained critic $\Vph_0$ (right).\label{subfig:1DvaluesA2C}]{%
		\includegraphics[scale=0.4]{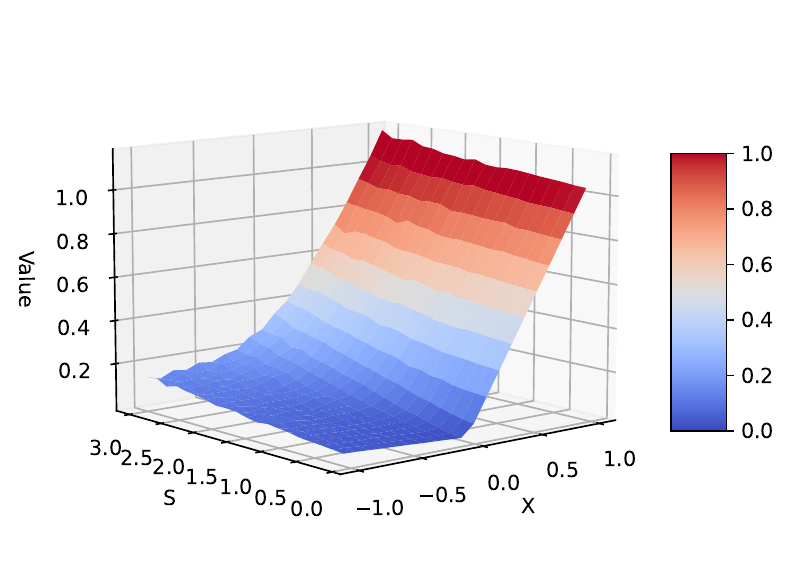}
		\includegraphics[scale=0.4]{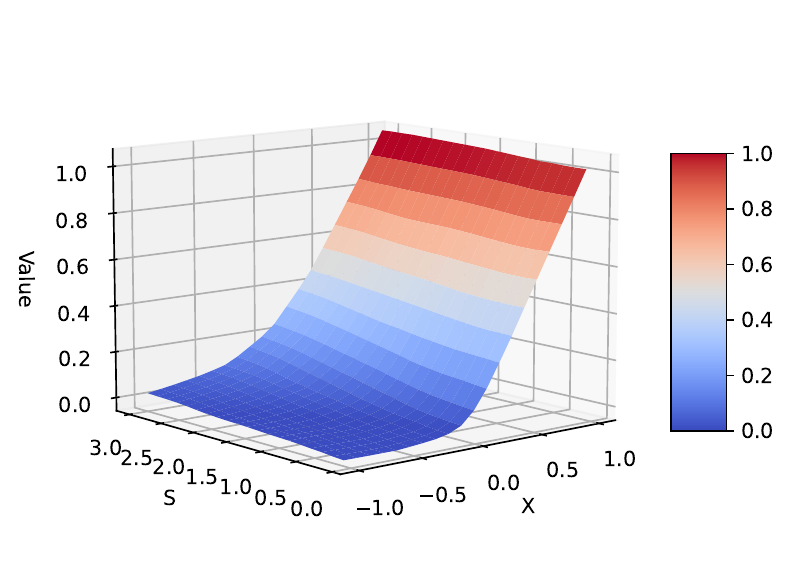}
    }
}
\caption{Estimated price surfaces for the passport option. 
}
    \label{fig:1Dvalues}
\end{figure}
Moreover, \Cref{fig:1Dvalues} shows MC estimates of the price surfaces obtained by the respective NN trading strategies $\pith$ over a grid of asset and portfolio values $(s,x)$ at time point $t=0$, i.e., a MC estimate of $e^{-rT}\E_{\Q}[(X_T^{\pith})^+\mid X_0^{\pith}=x, S_0=s]$. Not surprisingly, since $\pith$ approximate well $\pi^*$ (cp \Cref{fig:1Dstrats}) for both algorithms, also these prices coincide. Moreover, \Cref{fig:1Dvalues} also shows the estimate of the price of the passport option corresponding to the trained value function network $\Vph_0$ of \Cref{sec:A2C}.\footnote{ In order to get the estimate of the price of the passport option corresponding to the critic $\Vph_0$, we scale and shift $\Vph_0$ as in \Cref{le:absEquiv}. In \Cref{fig:1Dvalues}, we hence plot $(\Vph_0(s,x)+x)/2$.} This trained critic $\Vph_0$ can be quickly evaluated to obtain estimated prices (instead of doing a MC estimation). Note however that in regions far off the training data range (i.e., for asset values $s$ with $s>2.5$ and large absolute portfolio values $|x|>0.5$ approximately), the critic smoothly generalizes (as is typical for NN estimates), but the critic's price estimate might deviate more from a true price surface than within the range of training data (i.e., for asset values within (0.5, 2) and portfolio values within (-0.55, 0.55) approximately).

\begin{figure}[t]
\makebox[\textwidth][c]{%
		\includegraphics[scale=0.4]{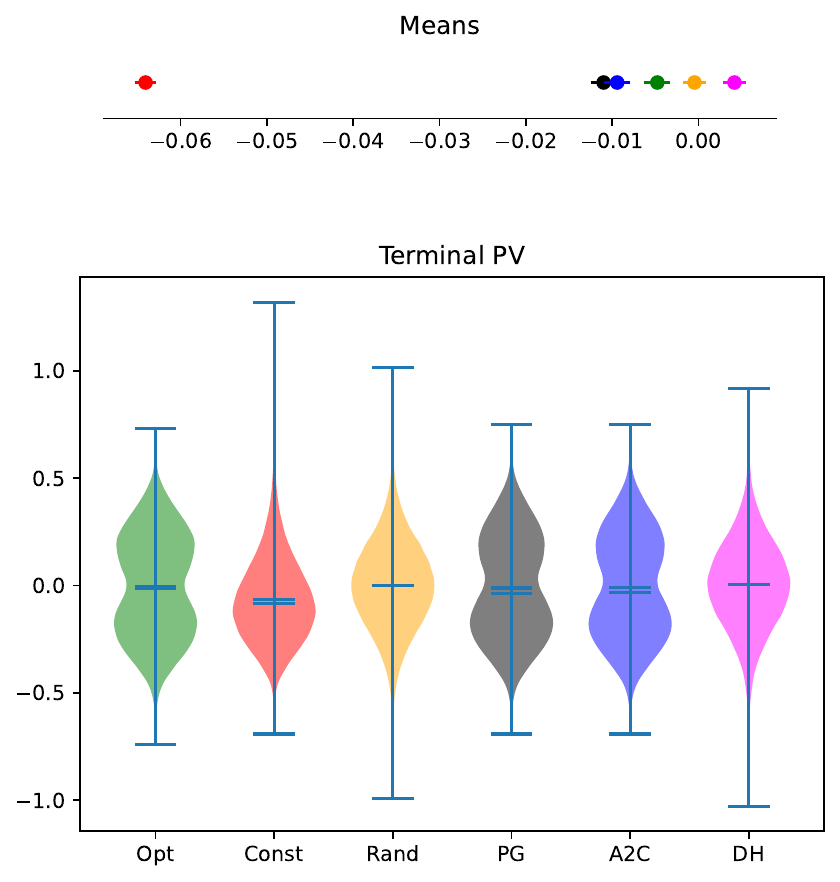}
		\includegraphics[scale=0.4]{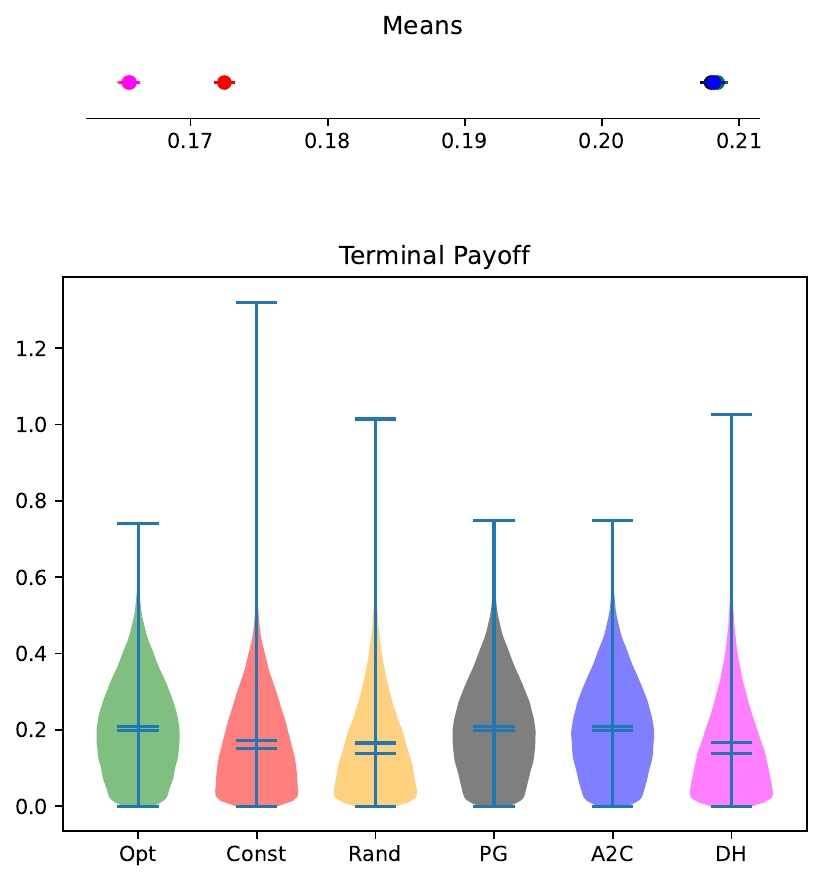}
	
}    \caption{Distributions of terminal portfolio values $X_T^{\pith}$ and terminal payoffs (absolute portfolio values $|X_T^{\pith}|$) over a hundred thousand test paths of the asset with $x_0=0$, for the optimal (Opt), a constant (Const), a random (Rand), PG and A2C strategies and a trained deep hedging (DH) strategy (left to right). Means with a student-t $95\%$-confidence interval are shown on top.}
    \label{fig:1DPVdist}
\end{figure}

Finally, in \Cref{fig:1DPVdist}, we present the empirical distribution of portfolio values $X_T^{\pith}$ and absolute portfolio values $|X_T^{\pith}|$ at terminal time $T=32$, obtained by trading with different strategies on a hundred thousand test paths. We contrast these distributions for the following strategies (from left to right in \Cref{fig:1DPVdist}): the optimal strategy obtained as a solution of \Cref{eq:1Dsolution}, a constant buy-and-hold strategy, i.e., $\pith_t(s,x)(q)=1$ for $q=1$ and all $t\in\R_+$ and $(s,x)\in\mathcal{X}$, a strategy that randomly chooses from actions $\{-1,1\}$ in each step, i.e., $\pith_t(s,x)(q)=0.5$ for $q\in\{-1,+1\}$, the trained PG and A2C strategies from \Cref{alg:policyGrad,alg:A2C}, and a (non-probabilistic) deep hedging strategy as discussed in \Cref{re:problemwithdeephedging}.

We observe that distributions of terminal payoffs, i.e., terminal absolute values are fairly similar for the optimal strategy and the trained NN strategies of \Cref{alg:policyGrad} and \Cref{alg:A2C} (Opt, PG, and A2C in \Cref{fig:1DPVdist}). The top rows in \Cref{fig:1DPVdist}, show means of the corresponding distributions with a student-t $95\%$-confidence interval. We see on the r.h.s. of \Cref{fig:1DPVdist} that these intervals overlap for the distributions corresponding to the optimal, the PG, and the A2C strategies. Therefore, on a $95\%$-confidence level, the means of absolute values achieved by trading with optimal and trained NN, i.e., PG and A2C strategies coincide.
Likewise, the means over terminal absolute values of deep hedging and random strategies are indistinguishable on a $95\%$-confidence level.

\subsection{The Multi-dimensional Case}\label{subsec:MDexperiments}
We have seen in the previous section that both RL algorithms PG and A2C could successfully recover the optimal trading strategy, i.e., the solution of \Cref{eq:1Dsolution} in a BS market with a single risky asset. In this section, we turn to the multi-asset case. First, we are interested to verify that the RL algorithms PG and A2C we introduced in \Cref{sec:MLApproaches} recover the optimal solution of \Cref{thm:optimalStrat} in a market setting with multiple \emph{uncorrelated} risky assets. Second, we analyze the strategies learned by these ML approaches in BS markets with multiple \emph{correlated} assets, where the solution for the problem of \Cref{eq:RLObjective} is still unknown.

Both DL approaches can be readily applied to markets with multiple risky assets. 
However, with increasing dimension, we need to tune hyper-parameters. Key factors to consider with increasing dimension are increasing sample size in \Cref{alg:policyGrad} to reduce variance and increasing entropy regularization in \Cref{alg:A2C} to foster exploration. 
\begin{figure}[t]
\makebox[\textwidth][c]{%
 \includegraphics[scale=.35]{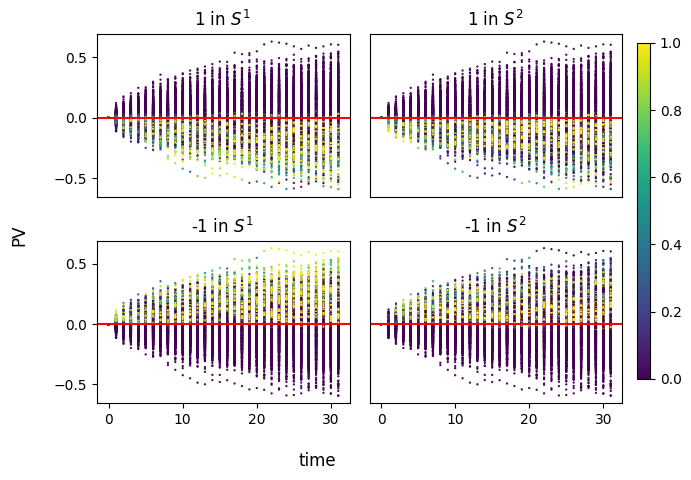}
 		\includegraphics[scale=.35]{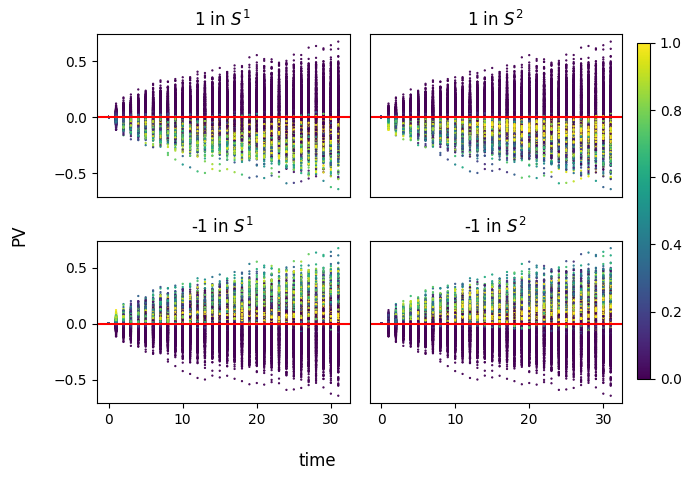}
}
    \caption{
    Evolution of portfolio values (PV) $X^{\pith}_t$ over time until maturity $T=32$ for actions taken according to $\pith$ trained with \Cref{alg:policyGrad} (left) and \Cref{alg:A2C} (right) respectively, over a test set of 1000 asset paths in a BS market with $x_0=0, \sigma^1=0.2=\sigma^2$ and $\rho=0$. In each of the sub-plots, the probability that the respective network $\pith$ assigns to taking the action $q\in\Diamond$ indicated in the sub-plots' titles is shown in color.}
    \label{fig:2DstratsalongPV}
\end{figure}
\subsubsection{2D Market with Uncorrelated Assets}\label{subsubsec:2Duncorrelated}
\paragraph{Symmetric Market.}
Consider first a symmetric BS market as described in \Cref{subsec:setting} with $d=2$ risky assets $S^1$ and $S^2$, with volatilities $\sigma_1=\sigma_2=0.2$, initial values $S_0^1=S_0^2=1$, correlation $\rho=0$ and interest rate $r=0.2\%$. In such an uncorrelated market, the optimal strategy for solving problem \eqref{eq:RLObjective} is derived in \Cref{thm:optimalStrat}.
To approximate this strategy, we let both algorithms (PG from \Cref{sec:PGA} and A2C from \Cref{sec:A2C}) run to learn optimal trades over $T=10$ equidistant time periods.

\begin{figure}[h!]
\makebox[\textwidth][c]{%
		\includegraphics[scale=0.4]{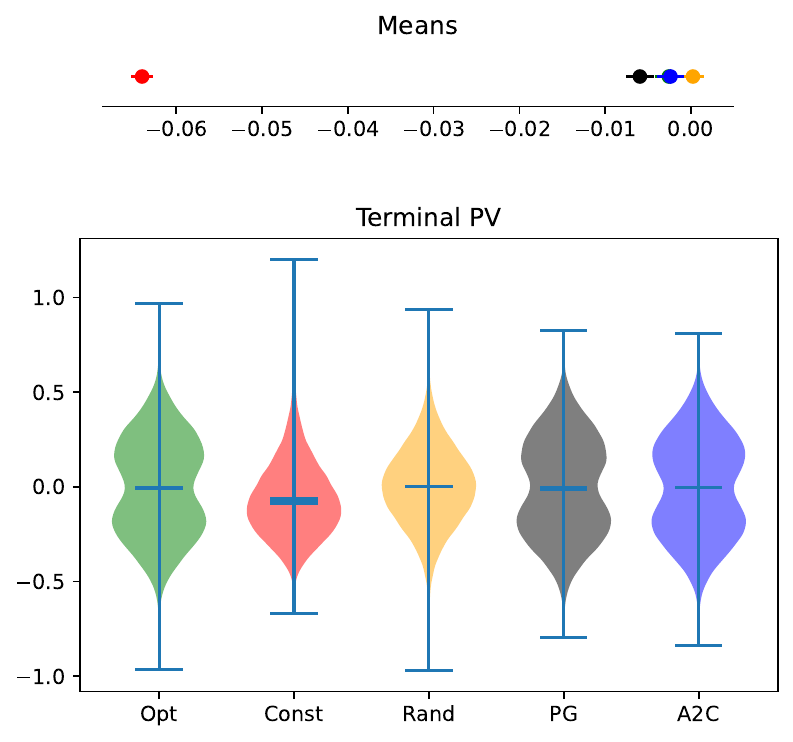}
		\includegraphics[scale=0.4]{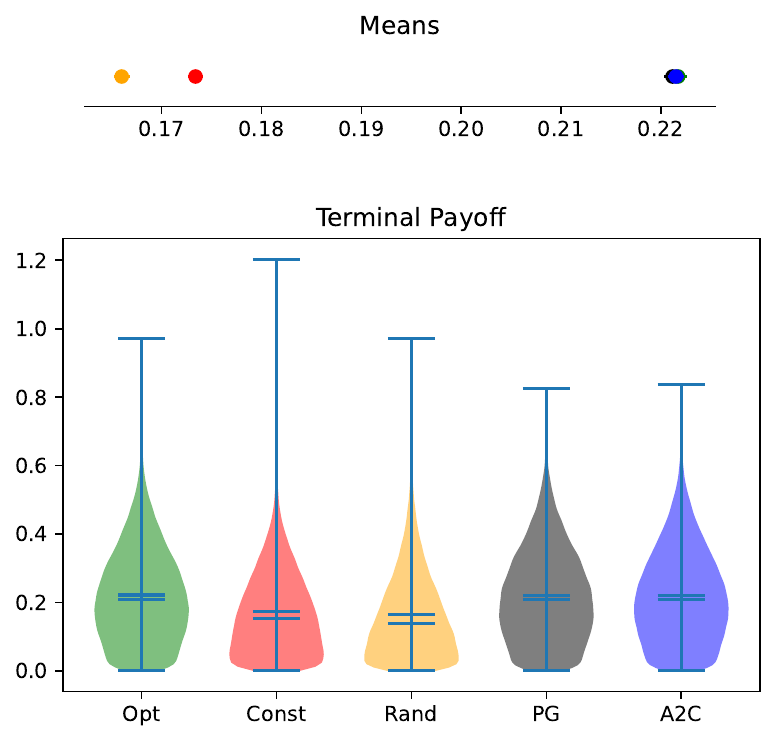}
	
} \caption{
Distributions of terminal PV $X^{\pith}_T$ and terminal payoff $|X^{\pith}_T|$ over a test set of 100 000 asset paths with $x_0=0$, for the optimal (Opt), a constant (Const), a random (Rand) and both trained NN strategies (PG and A2C) (left to right) in a BS market with $x_0=0, \sigma^1=\sigma^2=0.02$ and $\rho=0$. Means with a student-t $95\%$-confidence interval are shown on top.}
    \label{fig:2DPVDist}
\end{figure}

In \Cref{fig:2DstratsalongPV}, we again show evolutions of portfolio values $X^{\pith}_t$ over time $t$ until maturity $T=32$, for a test set of 1000 test asset paths and when actions are taken based on the trained network strategies $\pith$. In each of the sub-figures of \Cref{fig:2DstratsalongPV}, the color shows the probability that the respective trained network strategy $\pith$ assigns to taking each respective action $q\in\Diamond$ listed in the sub-titles. In \Cref{fig:2DstratsalongPV}, we see that all trained strategies $\pith$ choose to go short for negative portfolio values (i.e., $X_t^{\pith}<0$) and long otherwise, as does the optimal strategy derived in \Cref{eq:optimalStrat} in \Cref{thm:optimalStrat}. In both the left and right sub-figure of \Cref{fig:2DstratsalongPVAsym}, whenever portfolio values are negative, the NN strategies only assign a positive probability to only go long in either asset (first row of the sub-figure). Likewise, when portfolio values $X_t^\theta$ lie above zero, the NN strategies  assign positive probability only to go short in either asset (second row of the sub-figures). Moreover, asset preferences appear to be symmetric for both NN strategies, i.e., the probability of investing in asset $S^1$ (first columns of the sub-figures) is as evenly spread as the one of investing in asset $S^2$ (second columns of the sub-figures).

Furthermore, in \Cref{fig:2DPVDist}, we contrast empirical distributions of terminal values $X_T^{\pith}$ and terminal payoffs $|X_T^{\pith}|$ over a test set of 100 000 asset paths achieved by different strategies. As in the experiment in \Cref{subsec:1Dexperiments}, we consider the optimal strategy as in \Cref{eq:optimalStrat}, a constant buy-and-hold strategy, i.e., $\pith_t(s,x)(q)=1$ for strategy $q\in\Diamond$ fixed for all $t\in\R_+$ and $(s,x)\in\mathcal{X}$, a strategy that randomly chooses from actions $q\in\Diamond$ in each time step, i.e., $\pith_t(s,x)(q)=0.25$ for $q\in\Diamond$, and the trained PG and A2C strategies obtained from \Cref{alg:policyGrad,alg:A2C}. We see that both PG and A2C yield strategies that perform statistically on par with the optimal one, as indicated by overlapping confidence intervals of the means of terminal payoff distributions (right sub-plot in \Cref{fig:2DPVDist}).

Moreover, we again show price surfaces estimated by the trained NN strategies in \Cref{fig:2DvaluesCor0}. We show MC estimates of {$e^{-rT}\mathbb{E}{[(X^{\pith}_T)^{+}\mid X_0^{\pith}=0, S_0=s]}$} for both PG and A2C strategies $\pith$ for a grid of initial asset values $s=(s^1,s^2)\in[0,3]^2$. The price corresponding to the trained critic $\Vph_0$ from \Cref{alg:A2C} is shown in \Cref{subfig:2DvaluesCor0A2C}. (For each initial value $s$, we scale $\Vph(s,0)$ as in \Cref{le:absEquiv} to obtain an estimate $\Vph(s,0)/2$ of $V(s,0)$.) As in the one-asset market (cp. \Cref{fig:1Dvalues}), the critic is smoother than the MC approximations and assigns gives lower prices for asset values far from the training range ($s\in(0.5,1.5)^2$ approximately).

\begin{figure}[!t]

\makebox[\textwidth][c]{%
    \subfloat[MC estimate of {$e^{-rT}\mathbb{E}{[(X^{\pith}_T)^{+}\mid X_0^{\pith}=0, S_0=s]}$} with strategies $\pith$ obtained from PG (\Cref{alg:policyGrad}), plotted over a grid of asset values $s=(s^1, s^2)$.\label{subfig:2DvaluesCor0PG}]{%
		\includegraphics[scale=0.4]{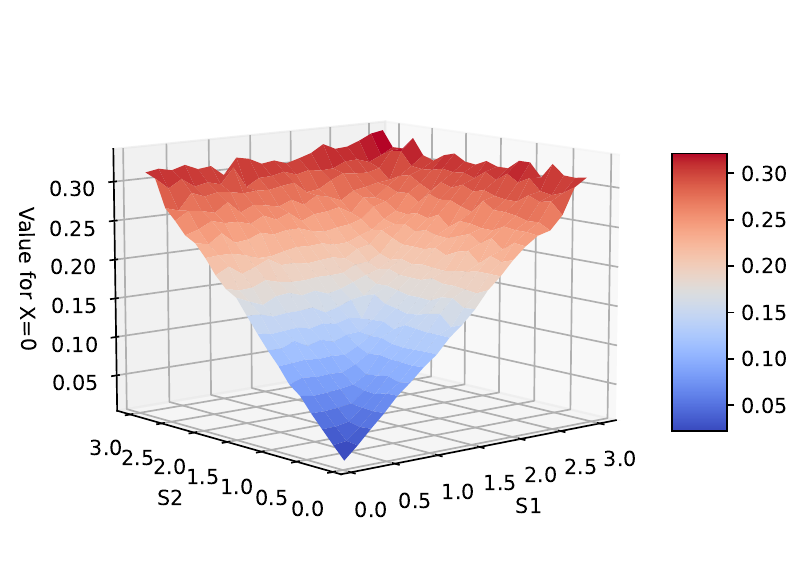}
		}
  }
  \makebox[\textwidth][c]{%
   \subfloat[MC estimate of {$e^{-rT}\mathbb{E}{[(X^{\pith}_T)^{+}\mid X_0^{\pith}=0, S_0=s]}$} with strategies $\pith$ obtained from A2C (\Cref{alg:A2C}) (left), and price surface $\Vph_0(s,0)/2$ corresponding to the trained critic $\Vph_0(s,0)$ (right), plotted over a grid of asset values $s=(s^1, s^2)$.\label{subfig:2DvaluesCor0A2C}]{%
  \includegraphics[scale=0.4]{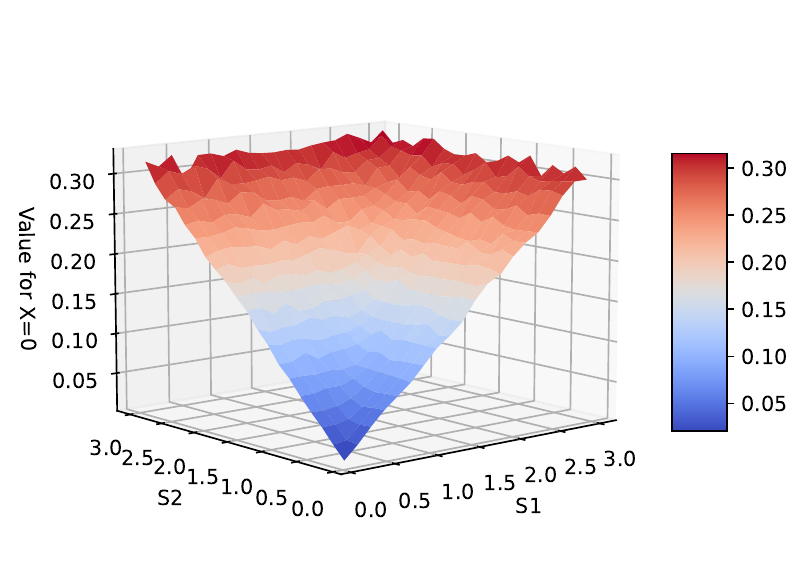}
		\includegraphics[scale=0.4]{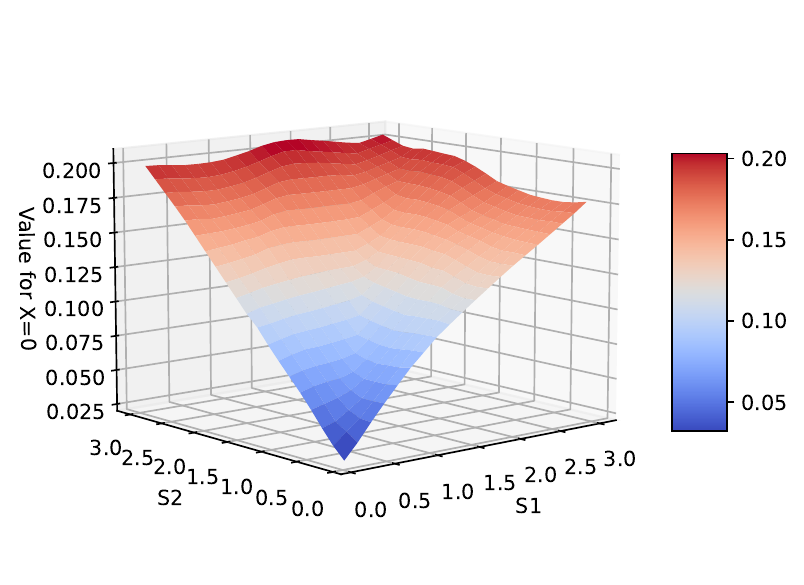}
	
    }
}   
\caption{Estimated price surfaces for a 2D passport option in a BS market as in \Cref{sec:preliminaries} with $\sigma^1=\sigma^2=0.2$ and $\rho=0$.}
    \label{fig:2DvaluesCor0}
\end{figure}

\paragraph{Asymmetric Market.} For an asymmetric setting with $d=2$ uncorrelated, risky assets, we consider variances $(\sigma^1)^2=0.04, (\sigma^2)^2=0.03$. Thus in this setting, asset $S^1$ is more volatile than asset $S^2$.

\begin{figure}[!t]
\makebox[\textwidth][c]{%
 \includegraphics[scale=.35]{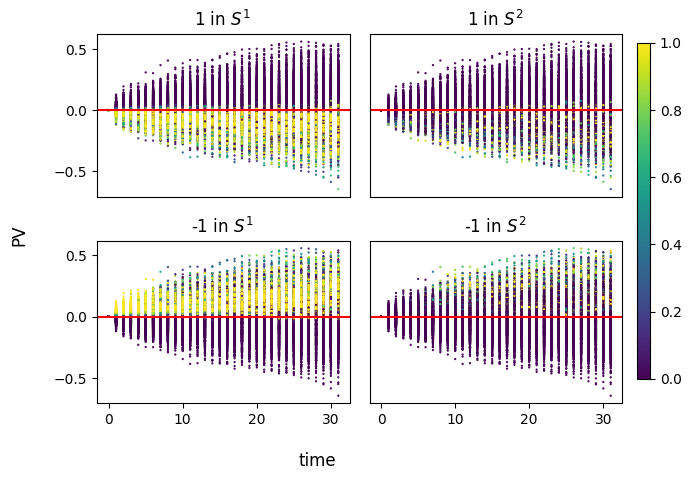}
		\includegraphics[scale=.35]{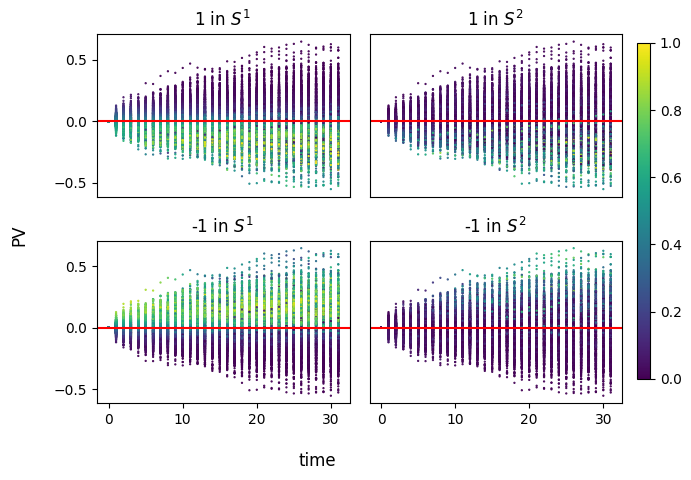}
	

}
    \caption{Evolution of portfolio values (PV) $X^{\pith}_t$ over time $t$ until maturity $T=32$ for actions taken according to $\pith$ trained with \Cref{alg:policyGrad} (left) and \Cref{alg:A2C} (right) respectively, over a test set of 1000 asset paths in an asymmetric BS market with $x_0=0, \sigma^1=0.2, \sigma^2=\sqrt{0.03}$ and $\rho=0$. In each of the sub-plots, the probability that the respective network $\pith$ assigns to taking the action $q\in\Diamond$ indicated in the sub-plots' titles is shown in color.}
    \label{fig:2DstratsalongPVAsym}
\end{figure}
\begin{figure}[t]
\makebox[\textwidth][c]{%
		\includegraphics[scale=0.4]{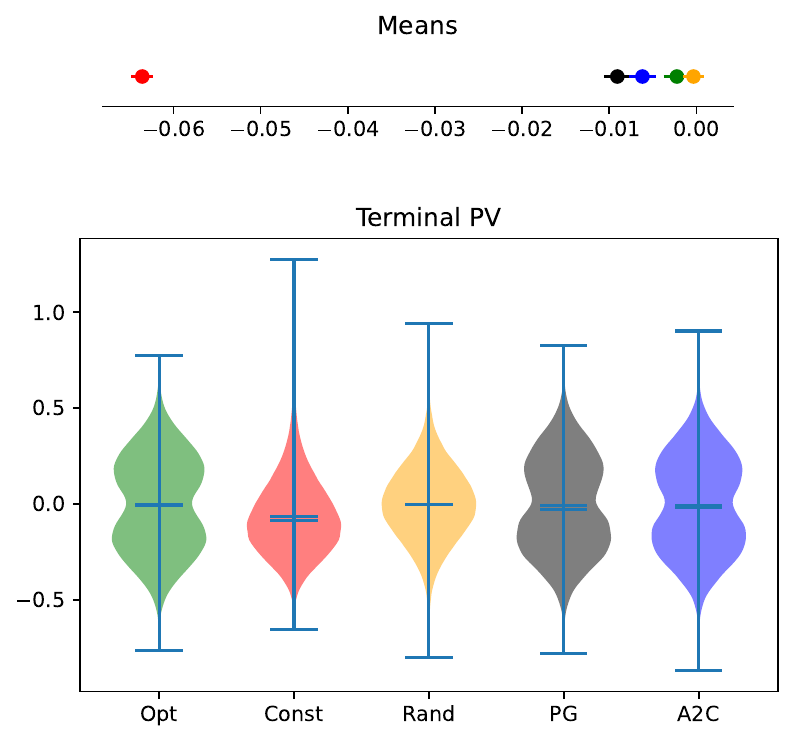}
		\includegraphics[scale=0.4]{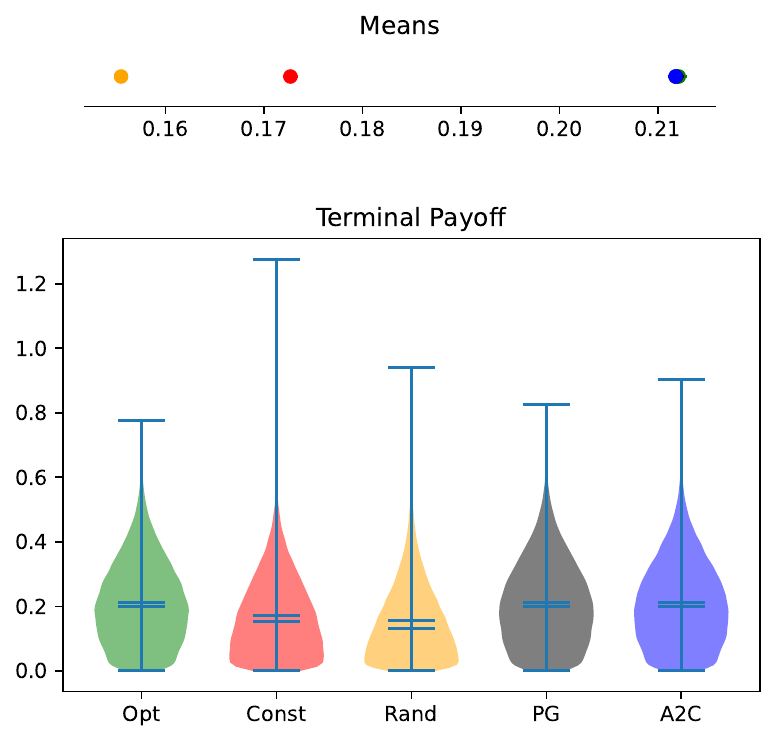}
	
} \caption{
Distributions of terminal PV $X^{\pith}_T$ and terminal payoff $|X^{\pith}_T|$ over a test set of 100 000 asset paths with $x_0=0$, for the optimal (Opt), a constant (Const), a random (Rand) and both trained NN strategies (PG and A2C) (left to right) in a BS market with $x_0=0, \sigma^1=0.2, \sigma^2=\sqrt{0.03}$ and $\rho=0$. Means with a student-t $95\%$-confidence interval are shown on top.}
    \label{fig:2DPVDistAsym}
\end{figure}

Similarly to before, \Cref{fig:2DstratsalongPVAsym} shows actions taken by the trained network strategies for both algorithms, PG and A2C, over a test set of 1000 asset paths. Again, we see that for both algorithms (PG and A2C), trained NN strategies tell to go short for negative portfolio values and long otherwise, as does the optimal strategy (cp. \Cref{thm:optimalStrat}). In both the left and right sub-figure of \Cref{fig:2DstratsalongPVAsym}, whenever portfolio values are negative, the NN strategies only assign a positive probability to only go long in either asset (first row of the sub-figure). Likewise, when portfolio values $X_t^\theta$ lie above zero, the NN strategies  assign positive probability only to go short in either asset (second row of the sub-figures). Moreover, both strategies assign higher probabilities to investing in the more volatile asset $S^1$ (first columns of the sub-figures).
While these asset preferences appear to be similar for both PG and A2C strategies, the PG strategy is a bit more confident in its actions than the A2C one. (To prevent over-confidence increasing entropy regularization is an option for the PG algorithm as well.)

In \Cref{fig:2DPVDistAsym}, we visualize the empirical distributions of both terminal portfolio values and terminal payoffs for the corresponding strategies. As in the symmetric setting, we see that on the test paths, distributions of terminal PV and terminal payoff achieved by trading with the trained NN strategies are indistinguishable from the optimal ones. 

As previously in the symmetric experiment, we show price surfaces estimated by the trained NN strategies $\pith$ in this asymmetric, uncorrelated market in \Cref{fig:2DvaluesAsym} in Appendix \ref{app:sec:figures}, where we observe similar patterns as in the symmetric case.

\subsubsection{2D Market with Correlated Assets}\label{subsubsec:2Dcorrelated}

In this section, we investigate the outcome of applying the DL algorithms of \Cref{sec:MLApproaches} to a BS market with two \emph{correlated} risky assets $S^1$ and $S^2$. In such a setting, the solution $\pi^*$ to the pricing problem of \Cref{eq:RLObjective} is \emph{unknown}. We can however still train a NN strategy $\pith$ according to our PG and A2C algorithms to approximate $\pi^*$.

We consider the previous asymmetric setting from \Cref{subsubsec:2Dcorrelated}, with volatilities $\sigma^1=0.2, \sigma^2=\sqrt{0.03}$, but in a market with \emph{non-zero} correlation $\rho\in\{-0.9,-0.5,-0.1,0.1,0.5,0.9\}$.
As in the experiments of the previous \Cref{subsubsec:2Duncorrelated}, we let both algorithms (PG from \Cref{sec:PGA} and A2C from \Cref{sec:A2C}) run to learn optimal trades over $T=10$ equidistant time periods in this correlated market. 
\begin{figure}[h!]
\makebox[\textwidth][c]{%
\subfloat[$\rho=0.1$]{%
 \includegraphics[scale=.35]{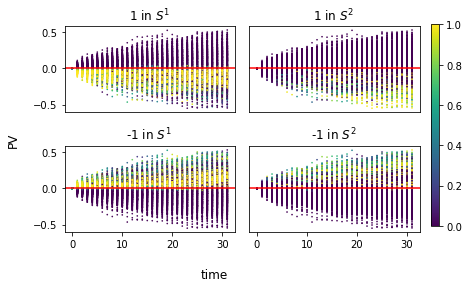}
		\includegraphics[scale=.35]{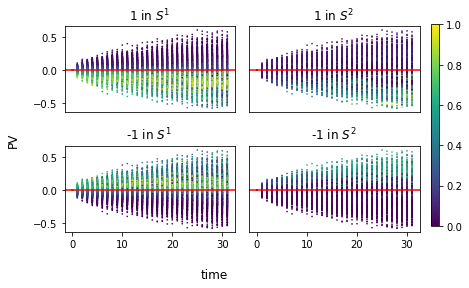}
    }
 }
 \makebox[\textwidth][c]{%
\subfloat[$\rho=0.5$]{%
 \includegraphics[scale=.35]{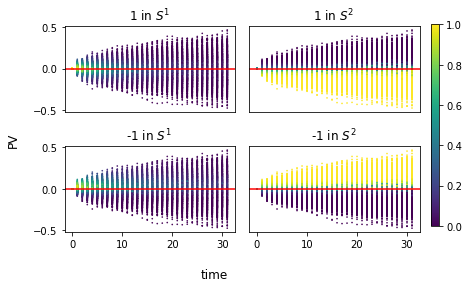}
		\includegraphics[scale=.35]{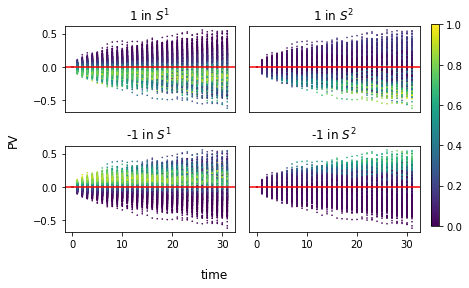}
    }
 }
\makebox[\textwidth][c]{%
\subfloat[$\rho=0.9$]{%
 \includegraphics[scale=.35]{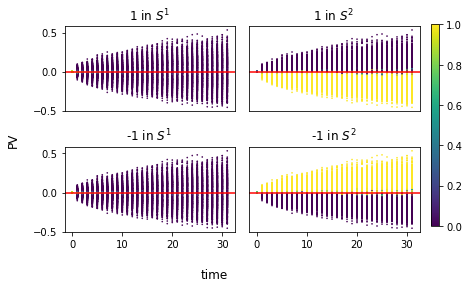}
		\includegraphics[scale=.35]{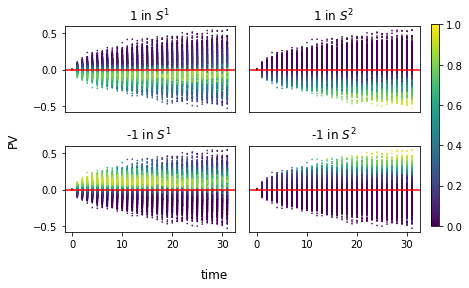}
    }}
    \caption{
    Evolution of portfolio values (PV) $X^{\pith}_t$ over time $t$ until maturity $T=32$ for actions taken according to $\pith$ trained with \Cref{alg:policyGrad} (left) and \Cref{alg:A2C} (right) respectively, over a test set of 1000 asset paths in an asymmetric BS market with $x_0=0, \sigma^1=0.2, \sigma^2=\sqrt{0.03}$ and positive correlations $\rho$. In each of the sub-plots, the probability that the respective network $\pith$ assigns to taking the action $q\in\Diamond$ indicated in the sub-plots title is shown in color.}
    \label{fig:2DstratsalongPVposcor}
\end{figure}

Analogously to before, we then sample a test set of 1000 asset paths and compute the evolutions of portfolio values (PV) $X^{\pith}$ over time until maturity $T=32$, when actions are taken based on the trained network strategies $\pith$. We show these PV evolutions for correlations $\rho\in\{0.1,0.5,0.9\}$ in \Cref{fig:2DstratsalongPVposcor}. (See \Cref{fig:2DstratsalongPVnegcor} in Appendix \ref{app:sec:figures} for the results in a market with negative correlation.) In each of the sub-figures of \Cref{fig:2DstratsalongPVposcor}, the color again signifies the probability that the trained network strategy $\pith$ assigns to taking each respective action $q\in\Diamond$ listed in the sub-titles. 
As in the uncorrelated setting, we see in this figure that across correlations, all trained strategies $\pith$ choose to go short for negative portfolio values and long otherwise. (Even though unproven, this is a strong indication that this is the case for the optimal solution of \Cref{eq:RLObjective} that is unknown in this setting.) Moreover, we observe that for both network strategies the trading action in asset $S^2$ (the less risky asset) increases with increasing correlation in the market. While as in the asymmetric uncorrelated setting, there is a clear preference for the more volatile asset $S^1$ over $S^2$ (cp. \Cref{fig:2DstratsalongPVAsym,fig:2DstratsalongPVposcor}), we see in \Cref{fig:2DstratsalongPVposcor} that now as $\rho$ increases, the probability of investing in asset $S^2$ increases. 
While this effect is only slightly visible for $\pith$ trained with the A2C algorithm of \Cref{sec:A2C} (right subplots of \Cref{fig:2DstratsalongPVposcor}), the tendency is much stronger for the network strategy $\pith$ trained with the PG algorithm of \Cref{sec:PGA} (left subplots of \Cref{fig:2DstratsalongPVposcor}). For $\rho \in\{0.5, 0.9\}$, the network strategy even learned to (almost in the case of $\rho=0.5$) purely invest in the less volatile asset $S^2$. 

\begin{figure}
\centering
\makebox[\textwidth][c]{%
\subfloat[$\rho=0.1$]{%
 \includegraphics[scale=.35]{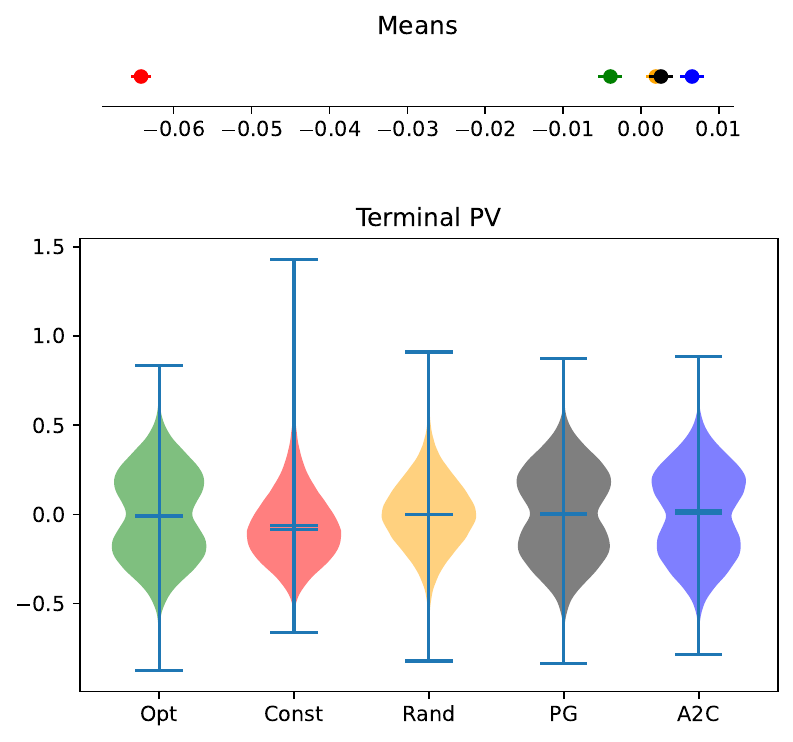}
		\includegraphics[scale=.35]{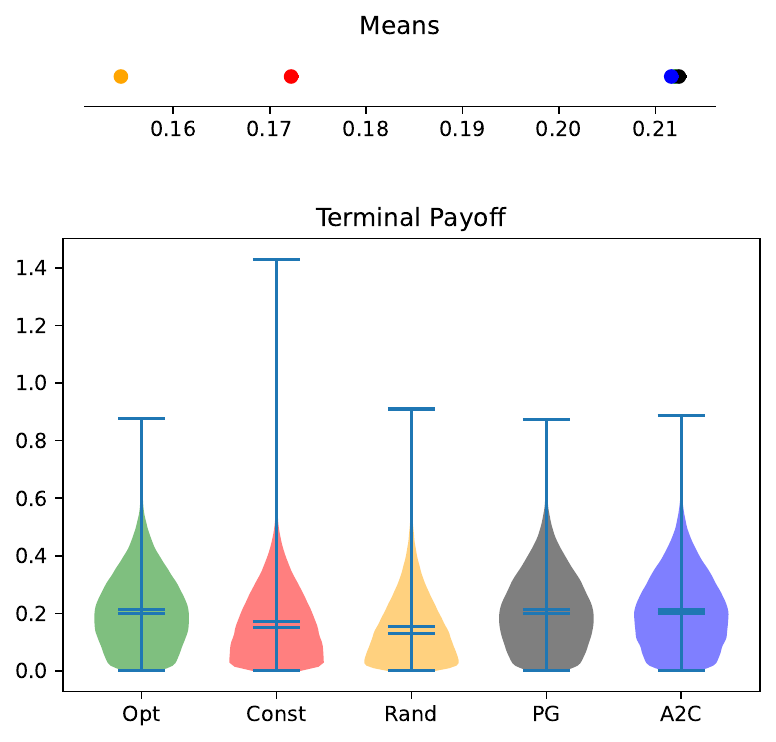}
    }
 }
 \makebox[\textwidth][c]{%
\subfloat[$\rho=0.5$]{%
 \includegraphics[scale=.35]{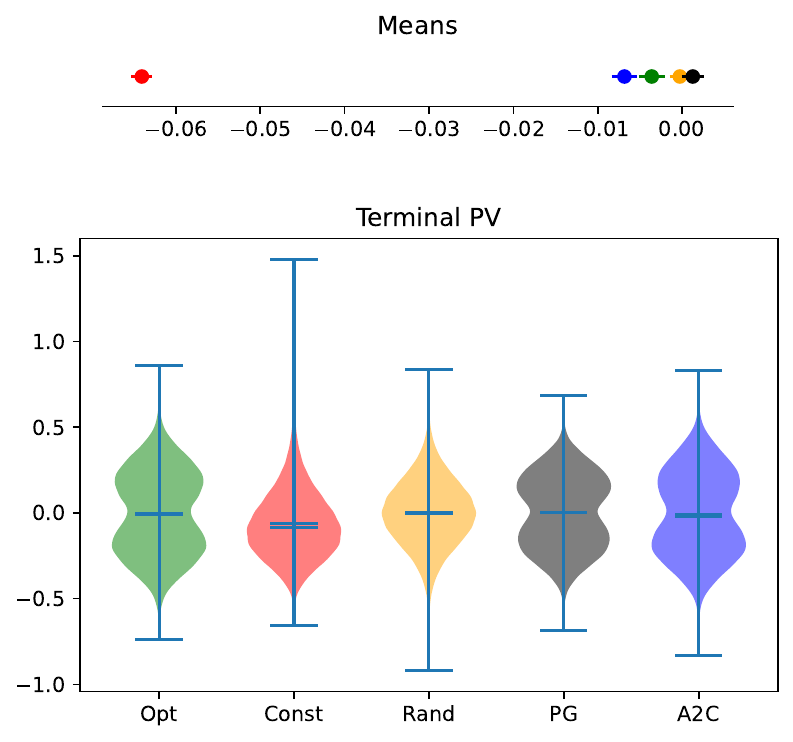}
		\includegraphics[scale=.35]{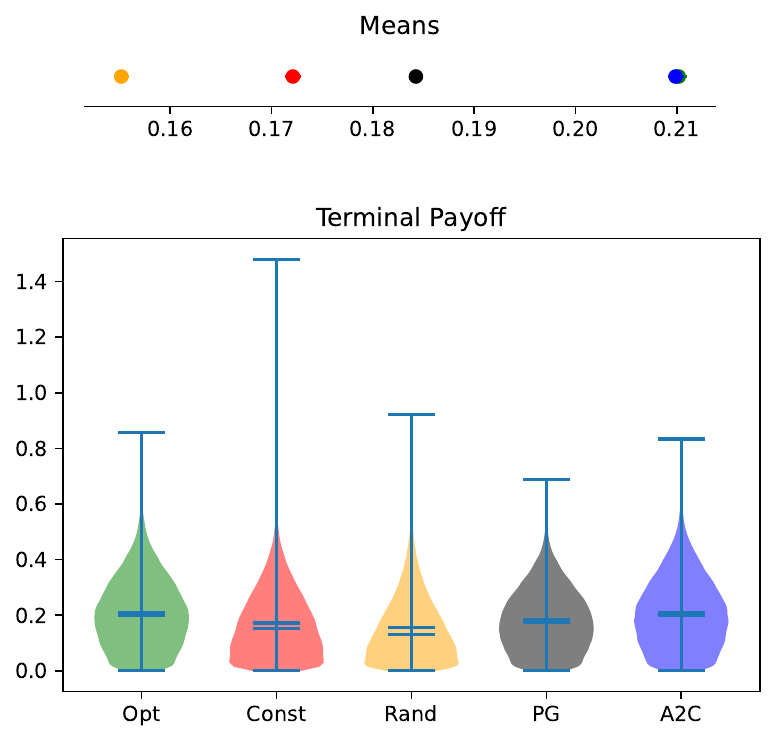}
    }
  }
\makebox[\textwidth][c]{%
\subfloat[$\rho=0.9$]{%
 \includegraphics[scale=.35]{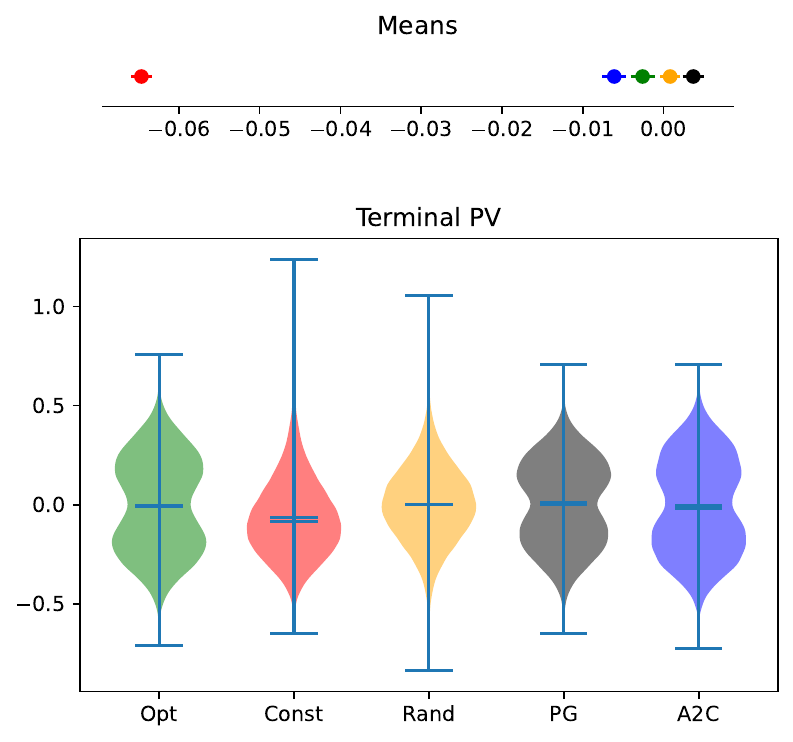}
		\includegraphics[scale=.35]{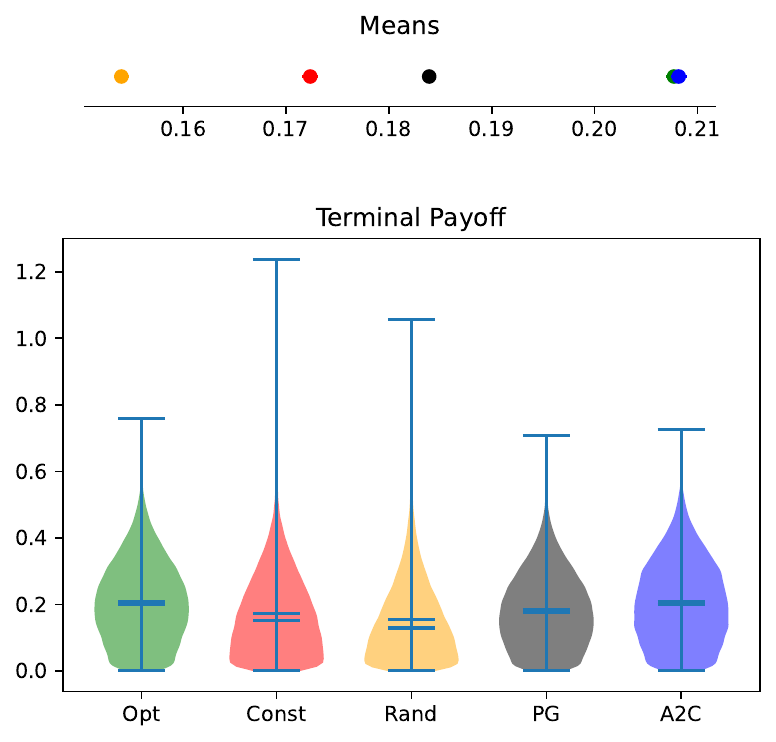}
    }} \caption{
Distributions of terminal PV $X^{\pith}_T$ and terminal payoff $|X^{\pith}_T|$ over a test set of 100 000 asset paths with $x_0=0$, for the strategy of \Cref{thm:optimalStrat} (Opt), a constant (Const), a random (Rand) and both trained NN strategies (PG and A2C) (left to right) in a BS market with $x_0=0, \sigma^1=0.02, \sigma^2=\sqrt{0.03}$ and positive correlations $\rho$. Means with a student-t $95\%$-confidence interval are shown on top.}
    \label{fig:2DPVDistposcor}
\end{figure}

We further also contrast distributions of terminal values $X_T^{\pith}$ and terminal payoffs $|X_T^{\pith}|$ over a test set of 100 000 asset paths achieved by different strategies in \Cref{fig:2DPVDistposcor}.\footnote{See Appendix \ref{app:sec:figures} \Cref{fig:2DPVDistnegcor} for the analogous figure for markets with negative correlation.} Analogously to the previous experiments of \Cref{subsec:1Dexperiments,subsubsec:2Duncorrelated}, we consider empirical distributions resulting from the following strategies $\pith$: first, the trained PG and A2C strategies obtained from \Cref{alg:policyGrad,alg:A2C} (PG and A2C in \Cref{fig:2DPVDistposcor}), second, a constant buy-and-hold strategy, i.e., $\pith(t,s,x)(q)=1$ for strategy $q\in\Diamond$ fixed for all $(t,s,x)\in\R_+\times\mathcal{X}$, and a strategy that randomly chooses from actions $q\in\Diamond$ in each step, i.e., $\pith(t,s,x)(q)=0.5$ for $q\in\Diamond$ (Const and Rand in \Cref{fig:2DPVDistposcor}), and third, the strategy of \Cref{thm:optimalStrat} (Opt in \Cref{fig:2DPVDistposcor}). We keep referring to the latter strategy as ``optimal strategy'', where optimal now means that it would be optimal in the same BS market \emph{without} correlation.
Across correlations, we note that the trained network strategies $\pith$ and the ``optimal'' strategy of \Cref{thm:optimalStrat} outperform the random and constant strategies, since the estimates of expected terminal payoffs (i.e., the means in the right sub-plots of \Cref{fig:2DPVDistposcor}) corresponding to the former are significantly larger than the ones resulting from trading with the latter strategies.

Moreover, we observe in \Cref{fig:2DPVDistposcor} that the A2C yields strategies that perform comparably to the strategy of \Cref{thm:optimalStrat} that was optimal only in the market with $\rho=0$. This is indicated by overlapping confidence intervals of the means of terminal payoff distributions (right sub-plots in \Cref{fig:2DPVDistposcor}).
The strategy to only go long or short in the less volatile asset $S^1$ that our PG simulation converged to for correlations $\rho\in\{0.5,0.9\}$ in the market, yields higher expected payoff than the random strategy (Rand) and the buy-and-hold (Const). However, it is significantly outperformed by the ``optimal'' and A2C strategies (Opt and A2C).

\section{Conclusion}
DL methods are on the rise in many applications within the field of finance. For the pricing of options on a portfolio value, prices are often given in terms of maximal expected payoffs when trading with a worst-case (from the seller's perspective) trading strategy. Thus, such tasks can be framed as stochastic control problems for which policy approximation methods are particularly well suited.  

In this paper, we have presented two ML-powered approaches to numerically approximate such maximizing trading strategies: a policy gradient method in \Cref{sec:PGA} and a advantage actor critic (A2C) method in \Cref{sec:A2C}. While conceptually, these algorithms can be generally applied to all sorts of options on traded accounts, in this paper we have focused on the multi-dimensional passport option.
Pricing this option has been a long-standing challenging problem even for simple Black-Scholes markets with two independent assets. 
Optimal actions in the control problem for pricing passport options are known to be of ``bang-bang type'', meaning that it is optimal to either go long or short in one of the market's assets. Thus, pricing a passport option is a classification task that comes with computational challenges due to noisy environments. We have discussed these challenges around data noise in classification and also how to deal with approximations of continuous-time solutions in \Cref{sec:MLApproaches}.

Within this paper, we have contributed to solving the problem of pricing the multi-dimensional passport option both analytically and numerically.
First, we have derived a discrete-time analytic solution for the optimal trading strategy in a multi-asset, uncorrelated BS market (see \Cref{thm:optimalStrat}). The optimal (worst-case) strategy in this problem demands to go long when the portfolio value is negative and short otherwise, and to do so in the asset for which a certain call price takes the highest value. 
Second, we have shown in \Cref{sec:experiments} that our ML-powered approaches are able to successfully recover these optimal solutions, also in the well-known single-asset setting. 
As part of future work, these algorithms can be applied to approximate worst-case pricing strategies in general, correlated BS- or even more general financial markets.

\acks{Heartfelt thanks go to Jakob Heiss, Jakob Weissteiner and Alexis Stockinger for the most fruitful discussions and contributions in coding.}


\newpage

\clearpage

\vskip 0.2in
\bibliography{passport_bib}

\clearpage
\appendix

\section{Proofs}\label{app:Results}
\begin{lemma}\label{le:VposHom}
For any $i\in\{1,\ldots,d\}$ and $k\in\{0,\ldots,T-1\}$ let $S^{i}\overset{(d)}{=}S^i_{t_{k+1}}$ and define
\begin{align*}
    V_k^{i}(y, \left( s^{i},s^{i-}\right)):=\E_{S^{i-}|s^{i-}}\left[\E_{S^{i}| s^{i}}\left[\lambda V_{k+1}\left(y+q^is^i\left(\left(\frac{S^i}{s^i}\right)-1\right), \left(S^{i},S^{i-}\right)\right)\right]\right].
\end{align*}
Then we have that for every $\lambda>0$, $V_k^{i}\left(\lambda y, \left(\lambda s^{i},s^{i-}\right)\right)=\lambda V_k^{i}\left( y, \left( s^{i},s^{i-}\right)\right)$. 
\end{lemma}
\begin{proof}
First,
\begin{align*}
   V_{T-1}^i\left(\lambda y, \left(\lambda s^{i},s^{i-}\right)\right)&=\E_{S^{i-}|s^{i-}}\left[\E_{S^{i}| \lambda s^{i}}\left[V_{T}\left(y\lambda+q^is^i\lambda\left(\left(\frac{S^i}{s^i\lambda}\right)-1\right), \left(S^{i}\frac{\lambda}{\lambda},S^{i-}\right)\right)\right]\right]\\
   =&\E_{S^{i-}|s^{i-}}\left[\E_{S^{i}| \lambda s^{i}}\left[\left|y\lambda+q^is^i\lambda\left(\left(\frac{S^i}{s^i\lambda}\right)-1\right)\right|\right]\right]\\
   =&\lambda \E_{S^{i-}|s^{i-}}\left[\E_{S^{i}| \lambda s^{i}}\left[\left|y+q^is^i\left(\left(\frac{S^i}{s^i\lambda}\right)-1\right)\right|\right]\right]\\
   =&\lambda \E_{S^{i-}|s^{i-}}\left[\E_{S^{i}|  s^{i}}\left[\left|y+q^is^i\left(\left(\frac{S^i}{s^i}\right)-1\right)\right|\right]\right]\\
   =&\lambda V_{T-1}^i\left( y, \left( s^{i},s^{i-}\right)\right),
\end{align*}
where the penultimate equality follows from a change of measure. Analogously, the same can be shown for $k\in\{0,\ldots,T-2\}$.
\end{proof}

\begin{lemma}\label{le:Vrepresentation}
For any fixed $s\in\R_+^d$,
and $k\in\{1,\ldots,N\}$, define the function $\psi:\R\to\R_+$, 
$\psi(y):=V_k(y,s)$. 
There exist a probability measure $\mu$ (that depends on $s$) on $\R_+$ s.t.
\begin{equation}
    \psi(y)=\int_{\R_+}\max(|y|,z)\,d\mu(z).
\end{equation}
\end{lemma}
\begin{proof}
We first show by induction over k that uniformly in $s$
\begin{equation*}
    \lim_{x\to\infty}V_k(x,s)-x =0, \quad \forall k\in\{1,\ldots,N\}.
\end{equation*}

For $k=0$, we have $\psi(x)-x=V_0(x,s)-x=|x|-x$, which goes to zero for $x\to\infty$ uniformly in $s$. Assume now that for some $k\in\Np, \lim_{x\to\infty}V_k(x,s)-x =0$ uniformly in $s$. Then
\begin{align*}
    V_{k+1}(x,s)-x&=\max_{q\in\Diamond}\Exsk\left[V_k\left(\underbrace{x\left(\frac{S^i_{t_{k+1}}}{s^i}\right)+q^is^i\left(1-\left(\frac{S^i_{t_{k+1}}}{s^i}\right)\right)}_{=:\tilde{x}}, S_{t_{k+1}}\right)-x\right]\\
    &=\max_{q\in\Diamond}\Exsk\left[V_k\left(\tilde{x}, S_{t_{k+1}}\right)-\tilde{x}+\sum_{j=1}^dq^js^j\left(\frac{S^j_{t_{k+1}}}{s^j}-1\right)\right].
\end{align*}
This implies that
\begin{align*}
    \lim_{x\to\infty} V_{k+1}(x,s)-x&=\max_{q\in\Diamond}\Exsk\left[\underbrace{\left(\lim_{x\to\infty}V_k\left(\tilde{x}, S_{t_{k+1}}\right)-\tilde{x}\right)}_{=0}+\sum_{j=1}^dq^js^j\left(\frac{S^j_{t_{k+1}}}{s^j}-1\right)\right]\\
    &=\max_{q\in\Diamond}\Exsk\left[\sum_{j=1}^dq^js^j\left(\frac{S^j_{t_{k+1}}}{s^j}-1\right)\right]\\
    &=\max_{q\in\Diamond}\sum_{j=1}^dq^js^j\Exsk\left[\frac{S^j_{t_{k+1}}}{s^j}-1\right]=0.
\end{align*}

Analogously to \citep[Lemma 6.3.]{delbaen2002passport}, one can further show that
\begin{enumerate}
    \item $\psi$ is convex,
    \item $\psi(-x)=\psi(x)$,
    \item $\psi(x)\ge |x|$ and $\lim_{x\to\infty}\psi(x)/x=1$.
\end{enumerate}
The result then follows from the proof of \cite[Lemma 6.4.]{delbaen2002passport}.

\end{proof}

\begin{remark}\label{re:varphiRepresentations}
Note that with the notation of \Cref{def:pmInvestment}, we have (dropping the time index in the notation of the asset)
\begin{align*}
 \varphi^i_+(z)&:=\Esicxsk\left[\max\left\{\bigg|S^i\left(\frac{|x|-s^i}{s^i}\right)+s^i\bigg|,z\right\}\right],\\
        \varphi^i_{-}(z)&:=\Esicxsk\left[\max\left\{\bigg|S^i\left(\frac{|x|+s^{i}}{s^i}\right)-s^i\bigg|,z\right\}\right].
\end{align*}
for $S^i\sim\logN\left(\log(s^i)-\frac{(\sigma^i)^2\Delta t_{k+1}}{2}, \sigi\sqrt{\Delta t_{k+1}}\right)$, with $\Delta t_{k+1}:=t_{k+1}-t_{k}$. Further, let $\kappa:=\frac{|x|+s^{i}}{s^i}, \tilde{\kappa}:=\frac{|x|-s^{i}}{s^i}$. Then we get
 \begin{align*}
     \varphi^i_-(z)&=\Esicxsk\left[z+\left(S^i\kappa-s^i-z\right)_++\left(-S^i\kappa+s^i-z\right)_+\right],\\
     &=z+\Esicxsk\left[\left(S^i\kappa-(s^i+z)\right)_+\right]\\
     &\hspace{3mm}+\Esicxsk\left[\left(S^i\kappa-(s^i-z)\right)_+\right]-\Esicxsk\left[S^i\kappa\right]+s^i-z\\
     &=\Esicxsk\left[\left(S^i\kappa-(s^i+z)\right)_+\right]+\Esicxsk\left[\left(S^i\kappa-(s^i-z)\right)_+\right]-|x|.
\end{align*}
Analogously,
\begin{align*}
      \varphi^i_+(z)&=\Esicxsk\left[\left(S^i\tilde{\kappa}+(s^i+z)\right)_+\right]+\Esicxsk\left[\left(S^i\tilde{\kappa}+(s^i-z)\right)_+\right]-|x|.
\end{align*}
\end{remark}

\begin{lemma}\label{le:negSignMax}
With the notation of \Cref{def:pmInvestment}, we have
\begin{equation*}
    \max\left\{\int_{\R_+}\varphi^i_+(z)\,d\muim(z),\int_{\R_+}\varphi^i_-(z)\,d\muim(z)\right\}=\int_{\R_+}\varphi^i_-(z)\,d\muim(z),\quad i=1\ldots, d,
\end{equation*}
provided that $x\neq 0$.
\end{lemma}
\begin{proof}
Let $i\in\{1,\ldots,d\}$ and define $f:\R_+\to\R, f(z):=\varphi^i_-(z)-\varphi^i_+(z)$ for $i=1,\ldots,d$. We will show that
$f(z)\ge0$ for all $z\in\R_+$.
By \cite[Lemma 6.5.]{delbaen2002passport} this follows, if \begin{enumerate}
    \item $\lim_{z\to\infty}f(z)=0$,
    \item for some $z$ big enough, $f(z)>0$,
    \item $f(0)\ge0$,
    \item there is at most one $z_0>0$ such that $f'(z)=0$.
\end{enumerate}
By \Cref{re:varphiRepresentations}, \begin{align*}
f(z)&=\Esicxsk\left[\left(S^i\kappa-(s^i+z)\right)_+\right]+\Esicxsk\left[\left(S^i\kappa-(s^i-z)\right)_+\right]\\
&-\Esicxsk\left[\left(S^i\tilde{\kappa}+(s^i+z)\right)_+\right]-\Esicxsk\left[\left(S^i\tilde{\kappa}+(s^i-z)\right)_+\right]
.\end{align*}
\begin{enumerate}
    \item For large enough $z$ we thus have
    
    \begin{align*}
        \lim_{z\to\infty}f(z)&= \lim_{z\to\infty}\Esicxsk\left[\left(S^i\kappa-(s^i-z)\right)_+-\left(S^i\tilde{\kappa}+(s^i+z)\right)_+\right]\\
        &=  \lim_{z\to\infty}\Esicxsk\left[\left(S^i\kappa-(s^i-z)\right)-\left(S^i\tilde{\kappa}+(s^i+z)\right)\right]=0\\
    \end{align*}
 and thus $\lim_{z\to\infty}f(z)=0$.
 \item By continuity of $f$, there exists $z$ such that $f(z)>0$, due to 3.
    \item \begin{align*}
        f(0)&=\Esicxsk\left[\bigg|S^i\left(\frac{|x|+s^{i}}{s^i}\right)-s^i\bigg|\right]-\Esicxsk\left[\bigg|S^i\left(\frac{|x|-s^i}{s^i}\right)+s^i\bigg|\right]\\
        &=\Esicxsk\left[\bigg|S^i-(s^i+|x|)\bigg|\right]-\Esicxsk\left[\bigg|S^i-(s^i-|x|)\bigg|\right]
    \end{align*} and thus $f(0)>0$ by \Cref{le:medianIneq}.
    \item For any $z>0$
    \begin{align*}
        \frac{\partial}{\partial z}f(z) 
    &=\P\left[S^i\kappa>s^i+z\right]+\P\left[S^i\kappa>s^i-z\right]-\P\left[S^i\tilde{\kappa}>-s^i-z\right]-\P\left[S^i\tilde{\kappa}>-s^i+z\right]\\
    &=-\P\left[s^i-z<S^i\kappa\le s^i+z\right]+\P\left[s^i-z<S^i\tilde{\kappa}\le s^i+z\right]\neq 0
    \end{align*}
since 
$\kappa\neq\tilde{\kappa}$.
\end{enumerate}

Thus $f(z)\ge0$ for every $z\ge0$ and hence
\begin{align*}
    \int_{\R_+}\varphi^i_+(z)\,d\muim(z)<\int_{\R_+}\varphi^i_-(z)\,d\muim(z).
\end{align*}
\end{proof}

\begin{lemma}\label{le:medianIneq}
Let $S\sim\logN(\mu,\sigma)$ with median $m=e^\mu$, and $c_1,c_2\in\R$ s.t. $c_1>\max\{m,c_2\}$. Then
\begin{equation}
    |m-c_1|>|m-c_2|\implies\E[|S-c_1|]>\E[|S-c_2|].
\end{equation}
\end{lemma}
\begin{proof}
    Note that $c_1\ge m$. First, we re-write
    \begin{align*}
        \E[|S-c_1|]=\E[|S-m+\overbrace{m-c_1}^{=-|m-c_1|}|]=& \E[\left(|S-m|+|m-c_1|\right)\1_{\{S\le m\}}]\\
        &+\E[-\left(|S-m|-|m-c_1|\right)\1_{\{ m<S\le c_1\}}]\\
         &+\E[\left(|S-m|-|m-c_1|\right)\1_{\{c_1\le S\}}]\\
         =& \E[|S-m|]-2\E[|S-m|\1_{\{ m<S\le c_1\}}]\\
         &+|m-c_1|\underbrace{\left(\P[S\le c_1]-\P[S\ge c_1]\right)}_{=2\P[m\le S\le c_1]}.
    \end{align*}
    Denoting by $f$ and $F$ the density and cdf of $S$ respectively, we thus get
    \begin{align*}
         \E[|S-c_1|]&=\E[|S-m|]+2\left(\int_m^{c_1}\left[(m-s)-(m-c_1)\right]f(s)\,ds\right)\\
         &=\E[|S-m|]+2\left((c_1-s)F(s)\bigg|^{c_1}_m+  \int_m^{c_1}F(s)\,ds\right) \\
         &=\E[|S-m|]+2\left(0.5(c_1-m)+ \int_m^{c_1}F(s)\,ds\right).
         \end{align*}
    Furthermore we distinguish two cases:
    \begin{enumerate}
        \item $m<c_2<c_1$: Analogously to above we get
        \begin{align*}
             \E[|S-c_2|]&=\E[|S-m|]+2\left(0.5(c_2-m)+ \int_m^{c_2}F(s)\,ds\right)\\
             &<\E[|S-m|]+2\left(0.5(c_1-m)+ \int_m^{c_1}F(s)\,ds\right)\\
             &=\E[|S-c_1|],
        \end{align*}
        since $c_1>c_2$ and $F(s)>0$ for all $s$.
        \item $c_2<m<c_1$: We first re-write
        \begin{align*}
             \E[|S-c_2|]&=\E[\left(|S-m|+|m-c_2|\right)\1_{\{ m\le S\}}]\\
        &+\E[\left(-|S-m|+|m-c_2|\right)\1_{\{ c_2<S\le m\}}]\\
         &+\E[-\left(-|S-m|+|m-c_2|\right)\1_{\{ S\le c_2\}}]\\
         &=\E[|S-m|]-2\E[|S-m|\1_{\{ c_2<S\le m\}}]\\
         &+|m-c_2|\underbrace{\left(\P[c_2\le S]-\P[S\le c_2]\right)}_{=2\P[c_2\le S\le m]}.
        \end{align*}
        Furthermore,
         \begin{align*}
         \E[|S-c_2|]&=\E[|S-m|]+2\left(\int_{c_2}^{m}\left[(s-m)+(m-c_2)\right]f(s)\,ds\right)\\
         &=\E[|S-m|]+2\left((s-c_2)F(s)\bigg|_{c_2}^m-  \int_{c_2}^{m}F(s)\,ds\right) \\
         &=\E[|S-m|]+2\left(0.5(m-c_2)- \int^m_{c_2}F(s)\,ds\right)\\
         &<=\E[|S-m|]+2\left(0.5(m-c_1)\right)\\
         &<\E[|S-c_1|],
         \end{align*}
         since by assumption $|m-c_1|>|m-c_2|$.
    \end{enumerate}
\end{proof}

\begin{lemma}\label{le:varphi_CP_condition}
    Let $i,j\in\{1,\ldots,d\}$, 
    and assume the notations as in the proof of \Cref{le:whichAssetIsMax}, and $\varphi_-$ as in \Cref{def:pmInvestment}. It holds that\begin{align*}
    \varphi^j_-(z)-\varphi^i_-(z)\ge 0, \forall z\ge 0,
\end{align*}
if and only if 
\begin{align*}
CP^j(s^j/\kappa^j)\kappa^j>CP^i(s^i/\kappa^i)
\kappa^i,   
\end{align*}
where $CP^i(K)$ denotes the call price with maturity $\Delta t_n$ on asset $S^i$ with strike $K$.
\end{lemma}

\begin{proof}
    Let $S^i$ respectively $ S^j$ as in \Cref{re:varphiRepresentations} and
\begin{equation*}
    f(z):=  \varphi^j_-(z)-\varphi^i_-(z).
\end{equation*}

Note that for $z$ large enough, 
\begin{align*}
    f(z)&= |x|+ \Esicxsk[j]\left[\left(S^j\kappa^j-(s^j-z)\right)_+\right]-|x|-\Esicxsk[i]\left[\left(S^i{\kappa^i}-(s^i-z)\right)_+\right]\\
        &=  \Esicxsk[j]\left[S^j\kappa^j-(s^j-z)\right]-\Esicxsk[i]\left[S^i{\kappa^i}-(s^i-z)\right]\\
         &= |x|+z-|x|-z=0.\\
\end{align*}
Moreover, by the Black Scholes pricing formula for every $i=1\ldots, d$ and $z\ge0$ we have
\begin{align}
    \Esicxsk\left[\left(S^i\kappa^i-(s^i+z)\right)_+\right]&=\kappa^i\left(s^i \Phi(d_1^i)-\frac{s^i+z}{\kappa^i}\Phi(d_2^i)\right), \\
    &=(s^i+|x|) \Phi(d_1^i)-(s^i+z)\Phi(d_2^i), \\
    d_1^i&=\frac{\log((s^i+|x|)/(s^i+z))+\frac{1}{2}(\sigma^i)^2\Delta t_{k+1}}{\sigi\sqrt{\Delta t_{k+1}}},\notag\\
    d_2^i &=d_1^i-\sigi\sqrt{\Delta t_{k+1}}.\notag\\
\end{align}
Furthermore, for $z=0$, 
\begin{align*}
     \varphi_-^i(z)&=|x|+2\kappa^i\Esicxsk\left[\left(S^i-\frac{s^i}{\kappa^i}\right)_+\right].
\end{align*}

Thus, the following hold.
\begin{enumerate}
\item$\lim_{z\to\infty}f(z)=0.$
\item
Since $f$ is continuous, 3. implies that there exists some $z_0$ such that $f(z_0)\ge 0$.
\item \begin{align*}
    f(0)&=|x|+2\kappa^j\left(\Esicxsk[j]\left[\left(S^j-\frac{s^j}{\kappa^j}\right)_+\right]\right)-|x|-2\kappa^i\left(\Esicxsk\left[\left(S^i-\frac{s^i}{\kappa^i}\right)_+\right]\right)\\
    &=2\left((s^j+|x|) \Phi(d_1^j)-s^j\Phi(d_2^j)-(s^i+|x|) \Phi(d_1^i)+s^i\Phi(d_2^i)\right)\big|_{z=0}
\end{align*}

We have $f(0)> 0$ if and only if
\begin{align*}
&\left((s^j+|x|) \Phi(d_1^j)-s^j\Phi(d_2^j)-(s^i+|x|) \Phi(d_1^i)+s^i\Phi(d_2^i)\right)\big|_{z=0}>0.
\end{align*}
\item For any $z>0$
    \begin{align*}
        \frac{\partial}{\partial z}f(z)
     &=\P\left[S^j\kappa^j>s^j+z\right]+\P\left[S^j\kappa^j>s^j-z\right]-\left(\P\left[S^i{\kappa^i}>s^i+z\right]+\P\left[S^i{\kappa^i}>s^i-z\right]\right)\\
     &=-\P\left[s^j-z<S^j\kappa^j\le s^j+z\right]+\P\left[s^i-z< S^i{\kappa^i}\le s^i+z\right]\\
    \end{align*}
Therefore, $\frac{\partial}{\partial z}f(z)$ can only be zero if and only if $$\P\left[s^j-z<S^j\kappa^j\le s^j+z\right]=\P\left[s^i-z< S^i{\kappa^i}\le s^i+z\right].$$
Thus, unless $s^i=s^j$, and $\sigi=\sigma^j$, for any $z>0$ $\frac{\partial}{\partial z}f(z)\neq 0$. 
\end{enumerate}

Thus, by \cite[Lemma 6.5.]{delbaen2002passport}, it follows from 1.-4. that $f(z)\ge0$ for all $z\ge0$ if and only if
\begin{align*}
&\left((s^j+|x|) \Phi(d_1^j)-s^j\Phi(d_2^j)-(s^i+|x|) \Phi(d_1^i)+s^i\Phi(d_2^i)\right)\big|_{z=0}>0.
\end{align*}
\end{proof}

\begin{lemma}\label{le:varphi_CP_condition_Sufficient}
    Let $i,j\in\{1,\ldots,d\}$,
    and assume the notation as in the proof of \Cref{le:whichAssetIsMax}.
    It holds that\begin{align*}
    \eqref{eq:CPcondition}\implies \int_{\Rp}\varphi^j_-(z)-\varphi^i_-(z)\,dz\ge 0.
\end{align*}

\end{lemma}

\begin{proof}
Recall for all $i$ the definition
  \begin{align*}
             \varphi^i_-(z):=&\Exsk\left[\max\{|\kappai{S^i_{T-1}}-s^i|,z\}\,\sum_{l=1}^d \frac{2\Phi'(d_1^l)}{(z+\Sl_{T-1})\sigl\sqrt{\Delta T}}M_l(z,\Sl)\right]\\
             &+2\Exsk\left[\left|\kappai{S^i_{T-1}}-s^i\right|\sum_{l=1}^d\left(2\Phi(\sigl\sqrt{\Delta T}/2)-1\right)\1_{\{\Sl_{T-1}\CPlOO\text{ is max}\}}\right].
         \end{align*}
    Let $S^i$ respectively $ S^j$ as in \Cref{re:varphiRepresentations} and
\begin{equation*}
    f(z):=  \varphi^j_-(z)-\varphi^i_-(z).
\end{equation*}

         For further derivations, we note that
         \begin{itemize}
             \item[a)]  $\varphi^i_-(\cdot)$ are continuous, and thus also $f$ is continuous. 
             \item[b)]for all $i$, \begin{align}\label{eq:varphii0}
             \varphi^i_-(0)=&\Exsk\left[\left|\kappai{S^i_{T-1}}-s^i\right|\sum_{l=1}^d\left(4\Phi\left(\frac{\sigl\sqrt{\Delta T}}{2}\right)-2+\frac{2\Phi'\left(\frac{\sigl\sqrt{\Delta T}}{2}\right)}{\Sl_{T-1}\sigl\sqrt{\Delta T}}\right)\1_{\{\Sl_{T-1}\CPlOO\text{ is max}\}}\right].
         \end{align}
                \item[c)] For all fixed $\Delta T>0$ $\varphi^9_-(0)$ is finite:
                \begin{align*}
                      \varphi^i_-(0)&\le\Exsk\left[\left|\kappai{S^i_{T-1}}-s^i\right|\right]\max_{l=1,\ldots d}\left(4\Phi\left(\frac{\sigl\sqrt{\Delta T}}{2}\right)-2\right)\\
                      &+\sum_{l=1}^d \underbrace{\Exsk\left[\left|\kappai{S^i_{T-1}}-s^i\right|\frac{2\Phi'\left(\frac{\sigl\sqrt{\Delta T}}{2}\right)}{\Sl_{T-1}\sigl\sqrt{\Delta T}}\right]}_{=:A^{i,l}}\\
                      &<\infty.
                \end{align*}
                \begin{small}Note that term $A^{i,l}$ can be further re-written. For $i\neq l$, due to independence of the assets, we get
        \begin{align*}
            A^{i,l}&=\Exsk\left[\left|\kappai{S^i_{T-1}}-s^i\right|\right]\frac{2\Phi'\left(\frac{\sigl\sqrt{\Delta T}}{2}\right)}{\sigl\sqrt{\Delta T}}\frac{1}{\ssl}\exp((\sigl)^2\Delta T)\\
            &=\left(2CP^i_{\kappai}(1)\ssi-|x|\right)\frac{2\Phi'\left(\frac{\sigl\sqrt{\Delta T}}{2}\right)}{\sigl\sqrt{\Delta T}}\frac{1}{\ssl}\exp((\sigl)^2\Delta T).
        \end{align*}
        Furthermore,
                \begin{align*}
            A^{i,i}&=\Exsk\left[\left|\kappai{\Si_{T-1}}-\ssi\right|\frac{1}{\Si_{T-1}}\right]\frac{2\Phi'\left(\frac{\sigi\sqrt{\Delta T}}{2}\right)}{\sigi\sqrt{\Delta T}}\\
            &=\Exsk\left[\left|\frac{\ssi}{\Si_{T-1}}-\kappai\right|\right]\frac{2\Phi'\left(\frac{\sigi\sqrt{\Delta T}}{2}\right)}{\sigi\sqrt{\Delta T}}\\
            &=\Exsk\left[\left|\Si_{T-1}-\kappai\exp((\sigi)^2\Delta T)\right|\right]\frac{2\Phi'\left(\frac{\sigi\sqrt{\Delta T}}{2}\right)}{\sigi\sqrt{\Delta T}}\frac{1}{\ssi}\exp((\sigi)^2\Delta T)\\
            &={\left(2CP^i_{1}\left(\kappai\exp^{(\sigi)^2\Delta T}\right)\ssi-\ssi+(|x|-\ssi)\exp^{(\sigi)^2\Delta T}\right)}
            \frac{2\Phi'\left(\frac{\sigi\sqrt{\Delta T}}{2}\right)}{\sigi\sqrt{\Delta T}}\frac{1}{\ssi}\exp((\sigi)^2\Delta T).
        \end{align*}
        \end{small}
         \end{itemize}
Note furthermore that for $z$ large enough, 
 \begin{align*}
             \varphi^i_-(z)=&\Exsk\left[\left(\left(\kappai{S^i_{T-1}}-(s^i-z)\right)_++0-\Si_{T-1}\kappai+\ssi\right)\,\sum_{l=1}^d \frac{2\Phi'(d_1^l)}{(z+\Sl_{T-1})\sigl\sqrt{\Delta T}}M_l(z,\Sl_{T-1})\right]\\
             &+2\Exsk\left[\left|\kappai{S^i_{T-1}}-s^i\right|\sum_{l=1}^d\left(2\Phi(\sigl\sqrt{\Delta T}/2)-1\right)\1_{\{\Sl_{T-1}\CPlOO\text{ is max}\}}\right]\\
            &=\Exsk\underbrace{\left[z\,\sum_{l=1}^d \frac{2\overbrace{\Phi'(d_1^l)}^{\overset{z\to\infty} {\to}0}}{(z+\Sl_{T-1})\sigl\sqrt{\Delta T}}\overbrace{M_l(z,\Sl_{T-1})}^{\le 1}\right]}_{\overset{z\to\infty}{\to}0}\\
             &+2\Exsk\left[\left|\kappai{S^i_{T-1}}-s^i\right|\sum_{l=1}^d\left(2\Phi(\sigl\sqrt{\Delta T}/2)-1\right)\1_{\{\Sl_{T-1}\CPlOO\text{ is max}\}}\right],
         \end{align*}
and thus
\begin{align*}
    \lim_{z\to \infty}\varphi^i_-(z)&< \Exsk\left[\left|\kappai{S^i_{T-1}}-s^i\right|\right]
    \cdot\underbrace{\max_{l=1,\ldots d}2\left(2\Phi\left(\frac{\sigl\sqrt{\Delta T}}{2}\right)-1\right)}_{=:F_+},\\
        \lim_{z\to \infty}\varphi^i_-(z)&> \Exsk\left[\left|\kappai{S^i_{T-1}}-s^i\right|\right]
    \cdot\underbrace{\min_{l=1,\ldots d}2\left(2\Phi\left(\frac{\sigl\sqrt{\Delta T}}{2}\right)-1\right)}_{=:F_-}.
\end{align*}
With this we get the necessary condition
\begin{align}\label{eq:Nej}
     \lim_{z\to \infty}\varphi^j_-(z)>  \lim_{z\to \infty}\varphi^i_-(z) &\Rightarrow \Exsk\left[\left|\kappa^j{S^j_{T-1}}-s^j\right|\right]>\frac{F_-}{F_+}\Exsk\left[\left|\kappai{S^i_{T-1}}-s^i\right|\right], \tag{Ne\textsuperscript{j}}
\end{align}
and the sufficient condition
\begin{align}\label{eq:Suj}
     \lim_{z\to \infty}\varphi^j_-(z)>  \lim_{z\to \infty}\varphi^i_-(z)&\Leftarrow \Exsk\left[\left|\kappa^j{S^j_{T-1}}-s^j\right|\right]>\frac{F_+}{F_-}\Exsk\left[\left|\kappai{S^i_{T-1}}-s^i\right|\right] , \tag{Su\textsuperscript{j}}
\end{align}
for $ \lim_{z\to \infty}\varphi^j_-(z)-\varphi^i_-(z)>0$.
Since the negation of conditions \eqref{eq:Suj} and \eqref{eq:Nej} correspond to the necessary condition (Ne\textsuperscript{i}) respectively the sufficient condition (Su\textsuperscript{i}) for $\lim_{z\to \infty}\varphi^j_-(z)<  \lim_{z\to \infty}\varphi^i_-(z) $, we get \eqref{eq:Suj} $\iff$\eqref{eq:Nej}. Moreover,
\begin{align*}
\eqref{eq:Suj} \implies \Exsk\left[\left|\kappa^j{S^j_{T-1}}-s^j\right|\right]>\Exsk\left[\left|\kappai{S^i_{T-1}}-s^i\right|\right]\implies \eqref{eq:Nej}.
\end{align*}
Thus (with reformulations from \Cref{re:varphiRepresentations} for $z=0$) we get that $ \lim_{z\to \infty}\varphi^j_-(z)-\varphi^i_-(z)>0$ if and only if
\begin{align*}
     \Exsk\left[\left|\kappa^j{S^j_{T-1}}-s^j\right|\right]-\Exsk\left[\left|\kappai{S^i_{T-1}}-s^i\right|\right]&>0\\
     \iff  \Exsk\left[\left(\kappa^j{S^j_{T-1}}-s^j\right)_+\right]-\Exsk\left[\left(\kappai{S^i_{T-1}}-s^i\right)_+\right]&>0.
\end{align*}
Therefore, \eqref{eq:CPcondition} is equivalent to $\lim_{z\to\infty}f(z)>0$. By items b) and c), $f(0)=\varphi_-^j(0)-\varphi_-^i(0)=C<\infty$ is equal to some finite constant $C\in\R$, and by continuity of $f$ we get that $\int_{\Rp}f(z)\,dz>0$.

\end{proof}

\begin{lemma}\label{le:CPconditionsEquivalences}
    Let $i,j\in\{1,\ldots,d\}$, and $\kappa$ as in \Cref{re:varphiRepresentations}.
    The following are equivalent
    \begin{enumerate}
    \item \begin{align*}
CP^j(s^j/\kappa^j)\kappa^j>CP^i(s^i/\kappa^i)
\kappa^i, 
\end{align*}
    \item \begin{align*}
s^jCP^j_{\kappa^j}(1)>\ssi CP^i_{\kappai}(1)
, 
\end{align*}
\end{enumerate}
 where $ CP^i_{s}(k)$ denotes the call price of an asset with starting value $s$, volatility $\sigi$, and strike price $k$ for maturity $\Delta T$.
\end{lemma}
\begin{proof}
    To see the equivalence, note that \begin{align*}
    \kappa^i CP^i(s^i/\kappa^i)=\kappai\E\left[\left(\Si-\frac{\ssi}{\kappai}\right)_+\right]=\ssi\E\left[\left(\Si\frac{\kappai}{\ssi}-1\right)_+\right]=\ssi\E\left[\left(\tilde{\Si}-1\right)_+\right]=\CPikapO\ssi,
\end{align*}
where the penultimate equality follows from a change of measure.
\end{proof}

\begin{lemma}\label{le:whichAssetIsMax}
Let $(x,s)$ and $k$ be fixed. Then investing in the $j$\textsuperscript{th} asset is preferred over investing in the $i$\textsuperscript{th} asset, i.e.,
\begin{align*}
     V^{j}_k(x,s)>V^{i}_k(x,s),
\end{align*}
if
\begin{align}\label{eq:CPcondition}
    CP^j(s^j/\kappa^j)\kappa^j>CP^i(s^i/\kappa^i)\kappa^i,
\end{align}
where for each $i$, $\kappa^i=\frac{|x|+s^i}{s^i}$ and
\begin{align*}
CP^j(s^i/\kappa^i)\kappa^i&=(s^i+|x|) \Phi(d_1^i)-s^i\Phi(d_2^i), \\
    d_1^i&=\frac{\log(1+|x|/s^i)+\frac{1}{2}(\sigma^i)^2\Delta t_{k+1}}{\sigi\sqrt{\Delta t_{k+1}}},\\
    d_2^i &=d_1^i-\sigi\sqrt{\Delta t_{k+1}}.\\
\end{align*}
\end{lemma}
\begin{proof}
    We proceed backwards in time. 
    \begin{itemize}
    \item[$\mathbf{ V_T}$] First, for $x\in\Rp, s\in\R^d_+$, $V_T(x,s)=|x|=\int_{\Rp}\max(|x|,z) \delta_0(dz)$, and thus the measure $\mu$ from \Cref{le:Vrepresentation} representing $V_T$ is equal to a Dirac delta at $0$. In particular, it is independent of $s$.
        \item[$\mathbf{ q_T^*}$]  We thus get that 
    \begin{align*}
        V_{T-1}(x,s)&= \max_{i=1\ldots,d}\,V_{T-1}^i(x,s)\\
        &= \max_{i=1\ldots,d}\,\Esimxsk\left[\int_{\R_+}\varphi^i_-(z)\,\mu(dz)\right]\\
        &=\max_{i=1\ldots,d}\,\int_{\R_+}\varphi^i_-(z)\,\mu(dz),
    \end{align*}
    where the last equality follows since $\varphi^i$ are independent of $S^{i-}$. By \Cref{le:varphi_CP_condition},
    \begin{align*}
      \varphi^j_-(z)-\varphi^i_-(z)\ge 0, \forall z\ge 0,
\end{align*}
if and only if 
\begin{align*}
CP^j(s^j/\kappa^j)\kappa^j>CP^i(s^i/\kappa^i)
\kappa^i.   
\end{align*}
Therefore, condition \eqref{eq:CPcondition} is sufficient for $ V^{j}_k(x,s)>V^{i}_k(x,s)$.\footnote{
In fact also conversely,
\begin{align*}
    V^{j}_k(x,s)>V^{i}_k(x,s) \iff \int_{\R_+}\varphi^j_-(z)\,\mu(dz)>\int_{\R_+}\varphi^i_-(z)\,\mu(dz),
\end{align*}
Since $\mu=\delta_0$, it further implies that $\varphi^j_-(0)-\varphi^i_-(0)\ge 0$, which by the proof of \Cref{le:varphi_CP_condition} is equivalent to equation \eqref{eq:CPcondition}.} Thus, $q_{t_N}^*(x,s)$ is given by \eqref{eq:optimalStrat}.

\item[$\mathbf{ V_{T-1}}$] We first define the indicator \begin{align*}
M_i(x,s):=\left\{\begin{matrix}
    1,& \kappa^i CP^i(s^i/\kappa^i) \text{ is max, }\\
    0,& \text{ else.}
\end{matrix}\right.
\end{align*}
Note that by \Cref{le:CPconditionsEquivalences}, equivalently
 \begin{align*}
M_i(x,s):=\left\{\begin{matrix}
    1,& \CPikapO\ssi \text{ is max, }\\
    0,& \text{ else.}
\end{matrix}\right.
\end{align*}
 By the previous step, inserting $q_T^*$ we get that
\begin{align*}
        V_{T-1}(x,s)&= \sum_{i=1}^d\E_{x,s,T-1}\left[\left|x\left(\frac{S^i_{T}}{s^i}\right)-\sign(x)s^i\left(1-\left(\frac{S^i_{T}}{s^i}\right)\right)\right|\right]M_i(x,s)\\
        &= \sum_{i=1}^d\E_{x,s,T-1}\left[\left|\kappa^i{S^i_{T}}-s^i\right|\right]M_i(x,s)\\
         &= \sum_{i=1}^d\si\E_{x,\kappa^i,T-1}\left[\left|\tilde{S^i_{T}}-1\right|\right]M_i(x,s)\\
          &= \sum_{i=1}^d2\si \left(CP^i_{\kappa^i}(1)-|x|\right)M_i(x,s),
    \end{align*}
    where the penultimate equality again follows from a change of measure.
    Thus we obtain the first derivative w.r.t. $x$
    \begin{align*}
        \partial_xV_{T-1}(x,s)&=\sum_{i=1}^d(2\Phi(d_1^i)-1)\sign(x)M_i(x,s),\\
        \end{align*}
        where $\Phi$ is the standard Gaussian cdf, and $d_1^i$ as in the formulation of \Cref{thm:optimalStrat}.
    Moreover, the second (distributional) derivative is given as
     \begin{align*}
        \partial^2_xV_{T-1}(x,s)&=\sum_{i=1}^d\left(\frac{2\Phi'(d_1^i)}{(|x|+\ssi)\sigi\sqrt{\Delta T}}\sign(x)+2\delta_0(x)(2\Phi(d_1^i)-1)\right)M_i(x,s)=:\mu(x).\\
        \end{align*}
        Therefore, the measure $\mu$ from \Cref{le:Vrepresentation} representing $V_{T-1}$ depends on the values of $s$.
        With this we get the representation
        \begin{align*}
            V_{T-1}(x,s) &=\int_{\Rp}\max\{|x|,z\}\,\mu(dz)\\
            &=\sum_{i=1}^d \bigg(\int_{\Rp}\max\{|x|,z\}\frac{2\Phi'(d_1^i)}{(|x|+\ssi)\sigi\sqrt{\Delta T}}\Mi\,dz\\
            &+2|x|\left(2\Phi(\sigi\sqrt{\Delta T}/2)-1\right)\Mi\bigg)
        \end{align*}

         \item[$\mathbf{ q_{T-1}^*}$] We now want to find out which asset we should invest in at time point $k=T-2$, i.e., we want to find $q_{T-1}^{*}$ that solves
         \begin{align*}
             \max_{i=1,\ldots,d}\Exsk\left[V_{T-1}\left(x\left(\frac{S^i_{T-1}}{s^i}\right)+q_{T-1}^is^i\left(1-\left(\frac{S^i_{T-1}}{s^i}\right)\right),\left(s^{i},S_{T-1}^{i-}\right)\right)\right].
         \end{align*}
         With the representation of $V_{T-1}$ from the previous step, this yields the objective
          \begin{align*}
             &\max_{i=1,\ldots,d}\Exsk\left[ \int_{\Rp}\max\{|x\left(\frac{S^i_{T-1}}{s^i}\right)+q_{T-1}^is^i\left(1-\left(\frac{S^i_{T-1}}{s^i}\right)\right)|,z\}\,\mu(dz)\right]\\
             =&\max_{i=1,\ldots,d}\Exsk\left[ \int_{\Rp}\max\{|\kappai{S^i_{T-1}}-s^i|,z\}\,\mu(dz)\right]\\
             =&\max_{i=1,\ldots,d}\Exsk\left[ \int_{\Rp}\max\{|\kappai{S^i_{T-1}}-s^i|,z\}\,\sum_{l=1}^d \frac{2\Phi'(d_1^l)}{(z+\Sl_{T-1})\sigl\sqrt{\Delta T}}M_l(z,\Sl)\,dz\right.\\
             &\left.+2\left|\kappai{S^i_{T-1}}-s^i\right|\sum_{l=1}^d\left(2\Phi(\sigl\sqrt{\Delta T}/2)-1\right)\1_{\{\Sl_{T-1}\CPlOO\text{ is max}\}}\right].
         \end{align*}
         We apply Fubini and define
         \begin{align}\label{eq:Defvarphii}
             \varphi^i_-(z):=&\Exsk\left[\max\{|\kappai{S^i_{T-1}}-s^i|,z\}\,\sum_{l=1}^d \frac{2\Phi'(d_1^l)}{(z+\Sl_{T-1})\sigl\sqrt{\Delta T}}M_l(z,\Sl)\right]\\
             &+2\Exsk\left[\left|\kappai{S^i_{T-1}}-s^i\right|\sum_{l=1}^d\left(2\Phi(\sigl\sqrt{\Delta T}/2)-1\right)\1_{\{\Sl_{T-1}\CPlOO\text{ is max}\}}\right].
         \end{align}
         With this our objective now is to show that asset $j$ solves
                  \begin{align*}
             \max_{i=1,\ldots,d}\int_{\Rp}\varphi^i_-(z)\,dz,
         \end{align*}
         if it fulfils condition \Cref{eq:CPcondition} for all $i\neq j$. 
        By \Cref{le:varphi_CP_condition_Sufficient} this holds true. 
         \item[$\mathbf{ V_t, q_t^*}$] Analogous steps yield the result for $t=T-2,\ldots,0$.

    \end{itemize}

\end{proof}
\newpage
\section{Figures}\label{app:sec:figures}
In this section, we give further figures referenced in the experiments of \Cref{sec:experiments}.
\subsection{2D Market with Independent Assets}
As in \Cref{subsubsec:2Duncorrelated}, we show price surfaces estimated by the trained NN strategies in \Cref{fig:2DvaluesAsym} in an asymmetric BS market with $d=2$ uncorrelated, risky assets with squared volatilities $(\sigma^1)^2=0.04, (\sigma^2)^2=0.03$. We show MC estimates of {$e^{-rT}\mathbb{E}{[(X^{\pith}_T)^{+}\mid X_0^{\pith}=0, S_0=s]}$} for both PG and A2C strategies $\pith$ for a grid of initial asset values $s=(s^1,s^2)\in[0,3]^2$. The price corresponding to the trained critic $\Vph_0$ from \Cref{alg:A2C} is shown in \Cref{subfig:2DvaluesCor0A2CAsym}. (For each initial value $s$, we scale $\Vph(s,0)$ as in \Cref{le:absEquiv} to obtain an estimate $\Vph(s,0)/2$ of $V(s,0)$.) 

\begin{figure}[t!]

\makebox[\textwidth][c]{%
    \subfloat[MC estimate of {$e^{-rT}\mathbb{E}{[(X^{\pith}_T)^{+}\mid X_0^{\pith}=0, S_0=s]}$} with strategies $\pith$ obtained from PG (\Cref{alg:policyGrad}), plotted over a grid of asset values $s=(s^1, s^2)$.\label{subfig:2DvaluesCor0PGAsym}]{%
		\includegraphics[scale=0.60]{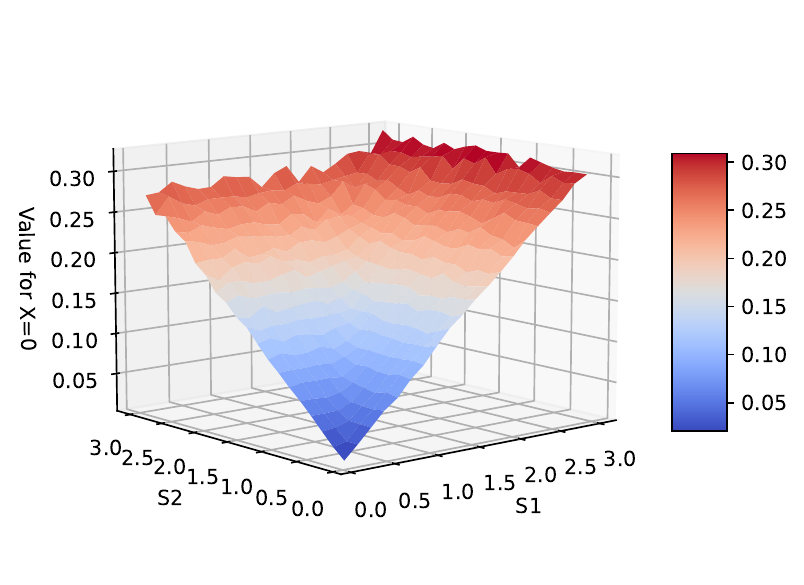}
		}
  }
  \makebox[\textwidth][c]{%
   \subfloat[MC estimate of {$e^{-rT}\mathbb{E}{[(X^{\pith}_T)^{+}\mid X_0^{\pith}=0, S_0=s]}$} with strategies $\pith$ obtained from A2C (\Cref{alg:A2C}) (left), and price surface $\Vph_0(s,0)/2$ corresponding to the trained critic $\Vph_0(s,0)$ (right), plotted over a grid of asset values $s=(s^1, s^2)$.\label{subfig:2DvaluesCor0A2CAsym}]{%
  \includegraphics[scale=0.60]{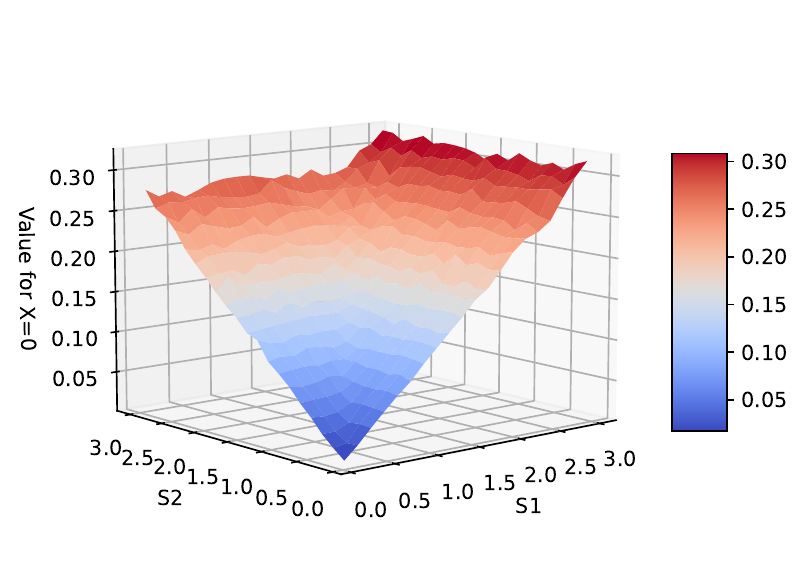}
		\includegraphics[scale=0.60]{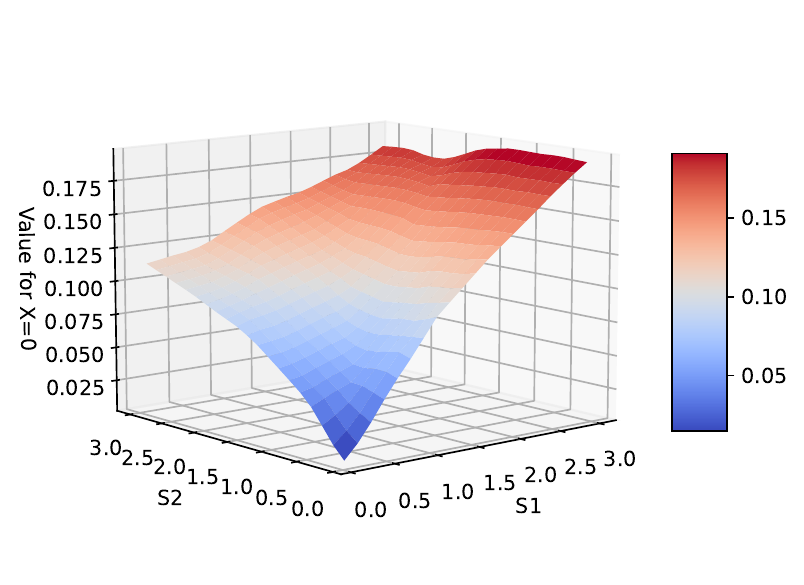}
	
    }
}   
\caption{Estimated price surfaces for a 2D passport option in an asymmetric BS market as in \Cref{sec:preliminaries} with $x_0=0, \sigma^1=0.2, \sigma^2=\sqrt{0.03}$ and $\rho=0$.}
    \label{fig:2DvaluesAsym}
\end{figure}
\clearpage
\subsection{2D Market with Negatively Correlated Assets}
 Analogously to the experiment in \Cref{subsubsec:2Dcorrelated}, we sample a test set of 1000 asset paths and compute the evolutions of portfolio values (PV) $X^{\pith}$ over time until maturity $T=32$, when actions are taken based on the trained network strategies $\pith$. We show these PV evolutions for markets with negative correlations $\rho\in\{-0.1,-0.5,-0.9\}$ in \Cref{fig:2DstratsalongPVnegcor}. In each of the sub-figures of \Cref{fig:2DstratsalongPVnegcor}, the color again signifies the probability that the respective trained network strategy $\pith$ assigns to taking each respective action $q\in\Diamond$ listed in the sub-titles. 
As in the markets with positive correlation of \Cref{subsubsec:2Dcorrelated}, all trained strategies $\pith$ choose to go short for negative portfolio values and long otherwise. Likewise, we observe the same trend that for both network strategies the trading action in asset $S^2$ increases with increasing (negative)correlation in the market. There still is a clear preference for the more volatile asset $S^1$ over $S^2$ (as in the uncorrelated setting shown in \Cref{fig:2DstratsalongPVAsym}), however we see in \Cref{fig:2DstratsalongPVnegcor} that as the correlation in the market increases (i.e., as $\rho$ decreases towards $-1$), the probability of investing in asset $S^2$ increases. 
As in the asymmetric setting with positive correlations, this effect is only slightly visible for $\pith$ trained with the A2C algorithm of \Cref{sec:A2C} (right subplots of \Cref{fig:2DstratsalongPVnegcor}). For the network strategy $\pith$ trained with the PG algorithm of \Cref{sec:PGA} (left subplots of \Cref{fig:2DstratsalongPVnegcor}), the tendency however is less strong as in the positively correlated market. For $\rho \in\{-0.5, -0.9\}$, the network strategy still tells to invest in both assets $S^1$ and $S^2$. 

Furthermore, we also show distributions of terminal values $X_T^{\pith}$ and terminal payoffs $|X_T^{\pith}|$ over 100 000 test paths achieved by different strategies in \Cref{fig:2DPVDistnegcor} for markets with negative correlations. Analogously to the previous experiments of \Cref{subsec:1Dexperiments,subsubsec:2Duncorrelated,subsubsec:2Dcorrelated}, we consider distributions resulting from the following strategies $\pith$: first, the trained PG and A2C strategies obtained from \Cref{alg:policyGrad,alg:A2C} (PG and A2C in \Cref{fig:2DPVDistnegcor}), second, a constant buy-and-hold strategy, i.e., $\pith_t(s,x)(q)=1$ for strategy $q\in\Diamond$ fixed for all $t\in\R_+$, $(s,x)\in\mathcal{X}$, and a strategy that randomly chooses from actions $q\in\Diamond$ in each step, i.e., $\pith_t(s,x)(q)=0.25$ for $q\in\Diamond$ (Const and Rand in \Cref{fig:2DPVDistnegcor}), and third, the strategy of \Cref{thm:optimalStrat} (Opt in \Cref{fig:2DPVDistnegcor}). Note once more that we keep referring to the latter strategy as ``optimal strategy'', where optimal now means that it would be optimal in the same BS market without correlation.
Across correlations, we observe that analogously to the outcome in positively correlated markets of \Cref{fig:2DPVDistposcor}, the trained network strategies $\pith$ and the ``optimal'' strategy of \Cref{thm:optimalStrat} outperform the random and constant strategies, since the estimates of expected terminal payoffs (i.e., the means in the right sub-plots of \Cref{fig:2DPVDistnegcor}) corresponding to the former are significantly larger than the ones resulting from trading with the latter strategies.
In particular, we observe in \Cref{fig:2DPVDistnegcor} that both the PG and the A2C algorithms yield strategies that perform comparably to the strategy of \Cref{thm:optimalStrat} that was optimal only in the market with $\rho=0$.

\begin{figure}[h!]
\makebox[\textwidth][c]{%
\subfloat[$\rho=-0.1$]{%
 \includegraphics[scale=.55]{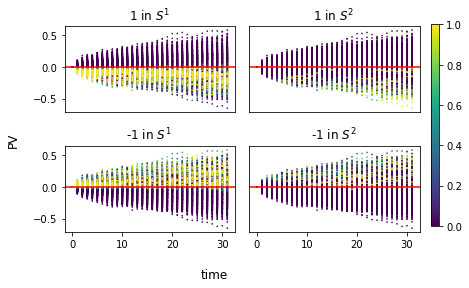}
		\includegraphics[scale=.55]{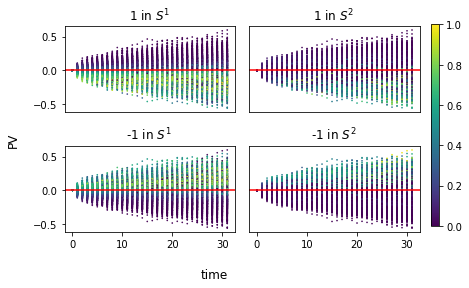}
    }
 }
 \makebox[\textwidth][c]{%
\subfloat[$\rho=-0.5$]{%
 \includegraphics[scale=.55]{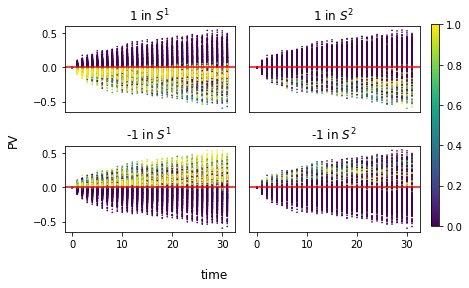}
		\includegraphics[scale=.55]{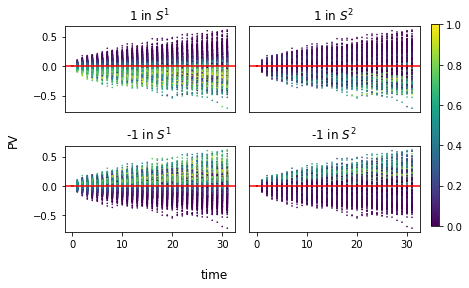}
    }
 }
\makebox[\textwidth][c]{%
\subfloat[$\rho=-0.9$]{%
 \includegraphics[scale=.55]{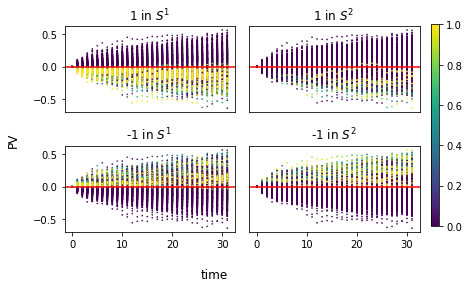}
		\includegraphics[scale=.55]{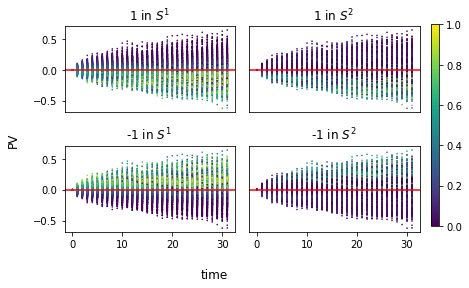}
    }}
    \caption{
   Evolution of portfolio values (PV) $X^{\pith}_t$ over time $t$ until maturity $T=32$ for actions taken according to $\pith$ trained with \Cref{alg:policyGrad} (left) and \Cref{alg:A2C} (right) respectively, over a test set of 1000 asset paths in an asymmetric BS market with $x_0=0, \sigma^1=0.2, \sigma^2=\sqrt{0.03}$ and negative correlations $\rho$. In each of the sub-plots, the probability that the respective network $\pith$ assigns to taking the action $q\in\Diamond$ indicated in the sub-plots title is shown in color.}
    \label{fig:2DstratsalongPVnegcor}
\end{figure}

\begin{figure}[h!]
\makebox[\textwidth][c]{%
\subfloat[$\rho=-0.1$]{%
 \includegraphics[scale=.45]{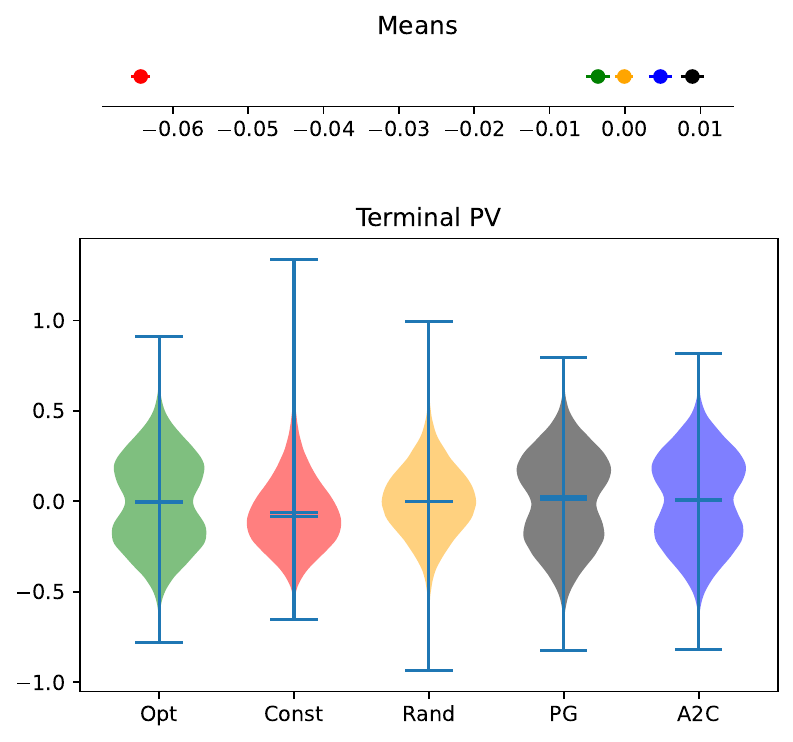}
		\includegraphics[scale=.45]{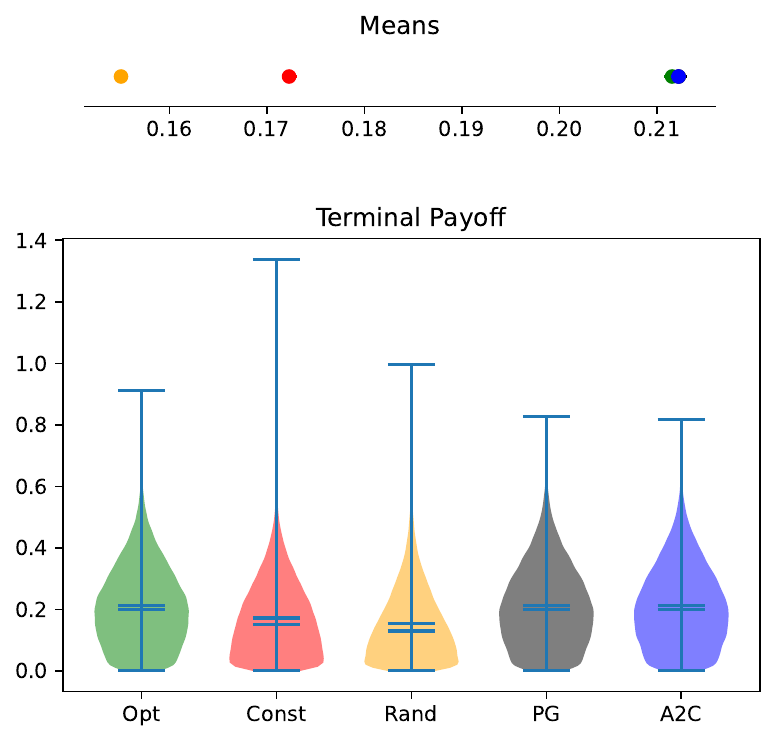}
    }
 }
 \makebox[\textwidth][c]{%
\subfloat[$\rho=-0.5$]{%
 \includegraphics[scale=.45]{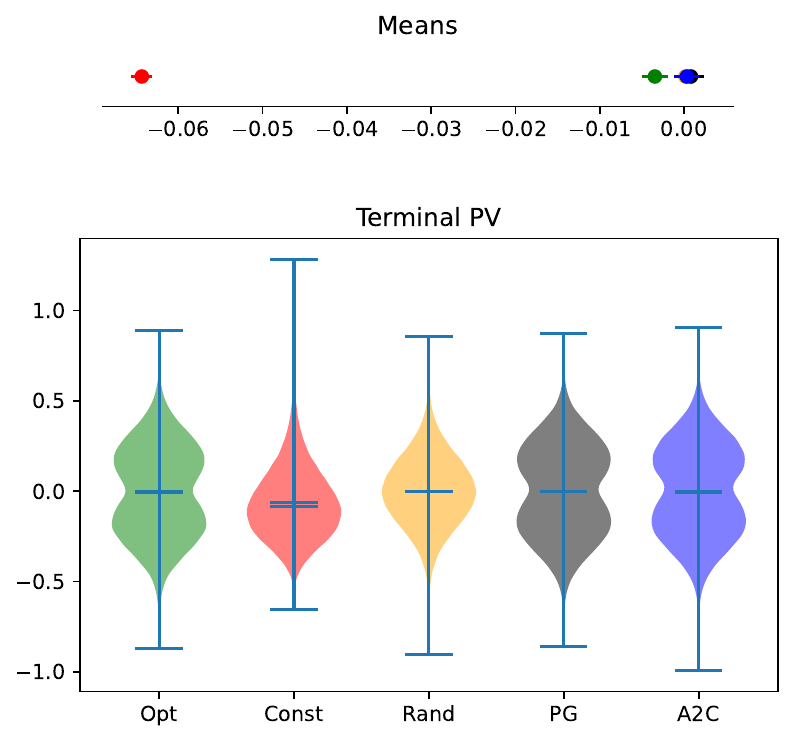}
		\includegraphics[scale=.45]{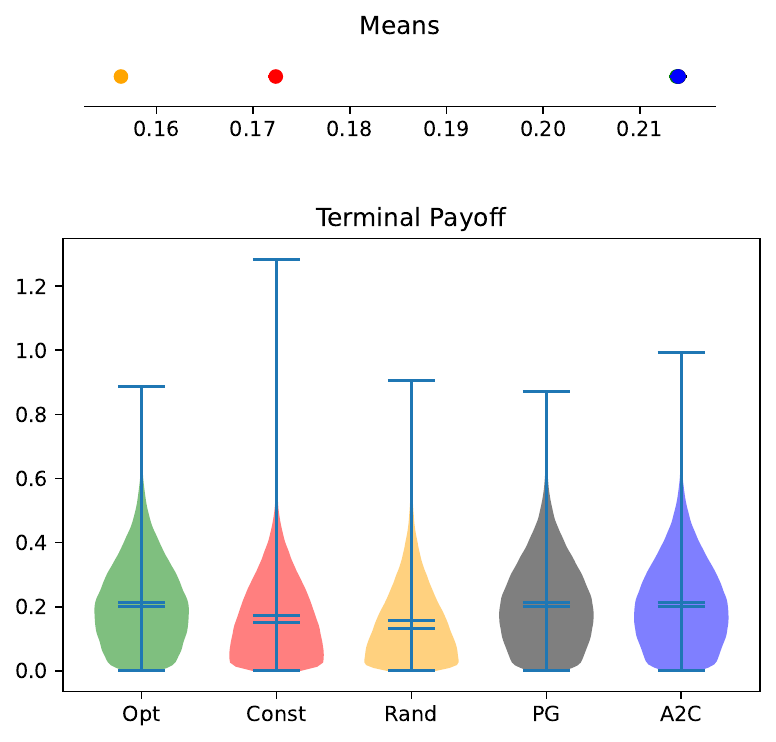}
    }
 }
\makebox[\textwidth][c]{%
\subfloat[$\rho=-0.9$]{%
 \includegraphics[scale=.45]{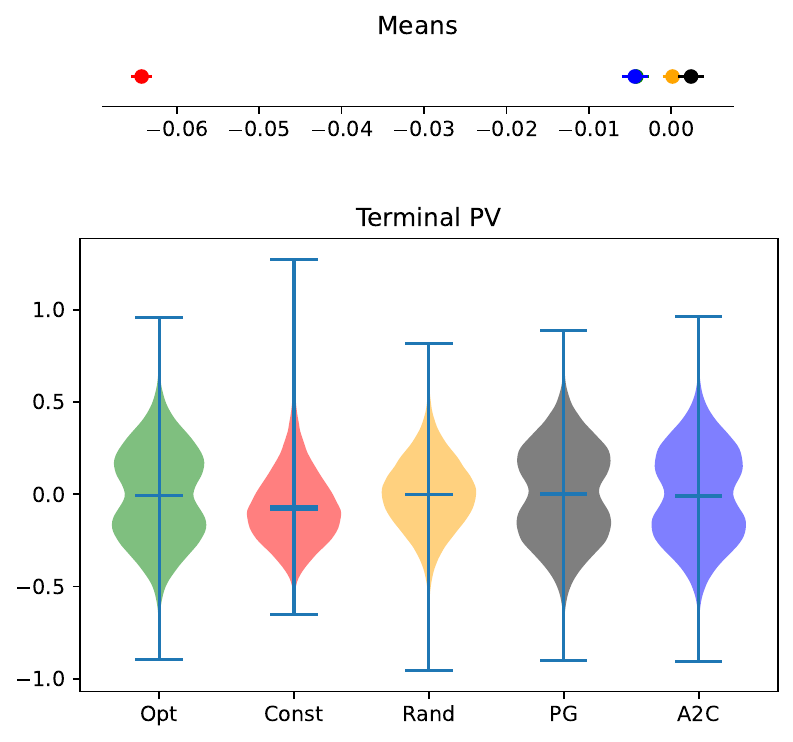}
		\includegraphics[scale=.45]{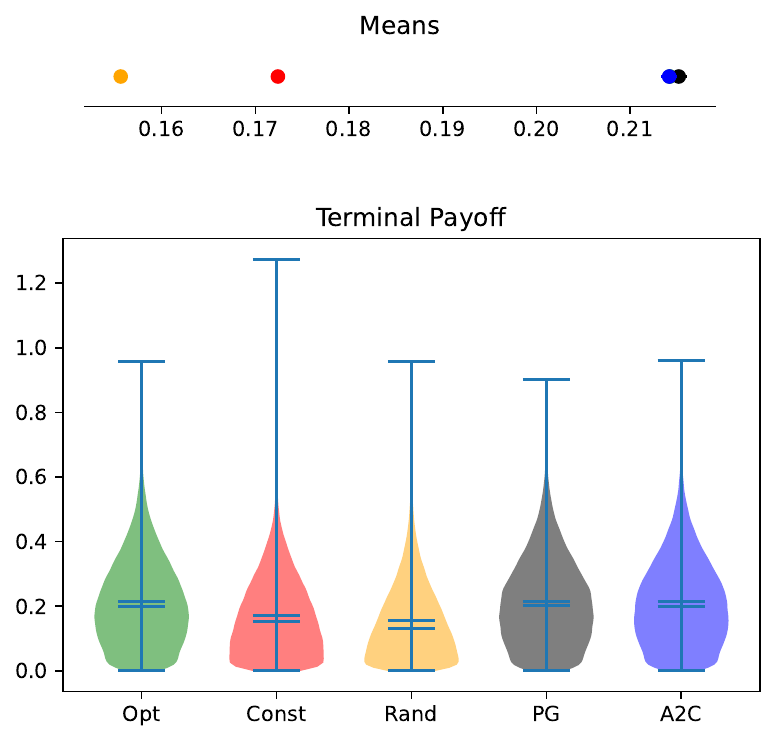}
    }} \caption{
Distributions of terminal PV $X^{\pith}_T$ and terminal payoff $|X^{\pith}_T|$ over a test set of 100 000 asset paths with $x_0=0$, for the strategy of \Cref{thm:optimalStrat} (Opt), a constant (Const), a random (Rand) and both trained NN strategies (PG and A2C) (left to right) in a BS market with $x_0=0, \sigma^1=0.02, \sigma^2=\sqrt{0.03}$ and negative correlations $\rho$. Means with a student-t $95\%$-confidence interval are shown on top.}
    \label{fig:2DPVDistnegcor}
\end{figure}

\end{document}